\documentclass[12pt,english]{article}
\usepackage{lmodern}
\usepackage[T1]{fontenc}
\usepackage[latin9]{inputenc}
\usepackage{geometry}
\usepackage{color}
\usepackage{babel}
\usepackage{mathrsfs}
\usepackage{amsmath}
\usepackage{amsthm}
\usepackage{amssymb}
\usepackage{stmaryrd}
\usepackage{graphicx}
\usepackage{setspace}
\usepackage[authoryear]{natbib}
\usepackage[unicode=true,
 bookmarks=true,bookmarksnumbered=false,bookmarksopen=false,
 breaklinks=false,pdfborder={0 0 0},pdfborderstyle={},backref=false,colorlinks=true]
 {hyperref}
\hypersetup{
 pdfpagelayout=OneColumn, pdfnewwindow=true, pdfstartview=XYZ, plainpages=false, urlcolor=[rgb]{0.647,0.11, 0.188}, linkcolor=[rgb]{0.647,0.11, 0.188}, citecolor=[rgb]{0.647,0.11, 0.188}, hypertexnames=false}
\usepackage{breakurl}

\makeatletter

\newcommand{\noun}[1]{\textsc{#1}}

\providecommand{\tabularnewline}{\\}

  \theoremstyle{definition}
  \newtheorem{defn}{\protect\definitionname}
  \theoremstyle{plain}
  \newtheorem{lem}{\protect\lemmaname}
  \theoremstyle{plain}
  \newtheorem{prop}{\protect\propositionname}
 \theoremstyle{definition}
  \newtheorem{example}{\protect\examplename}
\theoremstyle{plain}
\newtheorem{thm}{\protect\theoremname}
  \theoremstyle{remark}
  \newtheorem{rem}{\protect\remarkname}
  \theoremstyle{plain}
  \newtheorem{cor}{\protect\corollaryname}
  \theoremstyle{plain}
  \newtheorem*{fact*}{\protect\factname}

\usepackage{chngcntr}
\setcitestyle{round}
\usepackage{mathtools}

\makeatother

  \providecommand{\definitionname}{Definition}
  \providecommand{\examplename}{Example}
  \providecommand{\factname}{Fact}
  \providecommand{\lemmaname}{Lemma}
  \providecommand{\propositionname}{Proposition}
  \providecommand{\remarkname}{Remark}
\providecommand{\corollaryname}{Corollary}
\providecommand{\theoremname}{Theorem}

\begin{document}

\title{Learning and Type Compatibility in Signaling Games\thanks{This material was previously part of a larger paper titled ``Type-Compatible
Equilibria in Signaling Games.'' We thank Dan Clark, Laura Doval,
Glenn Ellison, Mira Frick, Ryota Iijima, Lorens Imhof, Yuichiro Kamada,
Robert Kleinberg, David K. Levine, Kevin K. Li, Eric Maskin, Dilip
Mookherjee, Harry Pei, Matthew Rabin, Bill Sandholm, Lones Smith,
Joel Sobel, Philipp Strack, Bruno Strulovici, Tomasz Strzalecki, Jean
Tirole, Juuso Toikka, Alex Wolitzky, and four anonymous referees for
helpful comments and conversations, and National Science Foundation
grant SES 1643517 for financial support. }}

\author{Drew Fudenberg\thanks{Department of Economics, MIT. Email: \texttt{\protect\href{mailto:drew.fudenberg@gmail.com}{drew.fudenberg@gmail.com}}}
\and Kevin He\thanks{Department of Economics, Harvard University. Email: \texttt{\protect\href{mailto:hesichao@gmail.com}{hesichao@gmail.com}} } }

\date{\text{First version: October 12, 2016}\\ \text{This version: }June 30, 2018}
\maketitle
\begin{abstract}
{\normalsize{}Which equilibria will arise in signaling games depends
on how the receiver interprets deviations from the path of play. We
develop a micro-foundation for these off-path beliefs, and an associated
equilibrium refinement, in a model where equilibrium arises through
non-equilibrium learning by populations of patient and long-lived
senders and receivers. In our model, young senders are uncertain about
the prevailing distribution of play, so they rationally send out-of-equilibrium
signals as experiments to learn about the behavior of the population
of receivers. Differences in the payoff functions of the types of
senders generate different incentives for these experiments. Using
the Gittins index \citep{gittins1979bandit}, we characterize which
sender types use each signal more often, leading to a constraint on
the receiver's off-path beliefs based on ``type compatibility'' and
hence a learning-based equilibrium selection. }{\normalsize \par}

{\normalsize{}\thispagestyle{empty}
\setcounter{page}{0}}{\normalsize \par}
\end{abstract}
\begin{flushleft}
{\small{}\newpage{}}
\par\end{flushleft}{\small \par}

\interfootnotelinepenalty=10000

\section{Introduction}

In a signaling game, a privately informed \emph{sender} (for instance
a student) observes their type (e.g. ability) and chooses a signal
(e.g. education level) that is observed by a \emph{receiver }(such
as an employer), who then picks an action without observing the sender's
type. These signaling games can have many perfect Bayesian equilibria,
which are supported by different specifications of how the receiver
would update his beliefs about the sender's type following the observation
of off-path signals that the equilibrium predicts will never occur.
These off-path beliefs are not pinned down by Bayes rule, and solution
concepts such as perfect Bayesian equilibrium and sequential equilibrium
place no restrictions on them. This has led to the development of
equilibrium refinements like \citet{cho_signaling_1987}'s Intuitive
Criterion and \citet{banks_equilibrium_1987}'s divine equilibrium
that reduce the set of equilibria by imposing restrictions on off-path
beliefs, using arguments about how players should infer the equilibrium
meaning of observations that the equilibrium says should never occur. 

This paper uses a learning model to provide a micro-foundation for
restrictions on the off-path beliefs in signaling games, and thus
derive restrictions on which Nash equilibria can emerge from learning.
Our learning model has a continuum of agents who are randomly matched
each period, with a constant inflow of new agents who do not know
the prevailing distribution of strategies and a constant outflow of
equal size. The large population makes it rational for the agents
to ignore repeated-game effects and ensures the aggregate system is
deterministic, while turnover in the population lets us analyze learning
in a stationary model where social steady states exist, even though
individual agents learn.\footnote{It is interesting to note that \citet{spence_job_1973} also interpreted
equilibria as steady states (or ``nontransitory configurations'')
of a learning process, though he did not explicitly specify what sort
of process he had in mind.} To give agents adequate learning opportunities, we assume that their
expected lifetimes are long, so that most agents in the population
live a long time. And to ensure that agents have sufficiently strong
incentives to experiment, we suppose that they are very patient. This
leads us to analyze what we call the ``\emph{patiently stable}''
steady states of our learning model.

Our agents are Bayesians who believe they face a time-invariant distribution
of opponents' play. As in much of the learning-in-games literature
and most laboratory experiments, these agents only learn from their
personal observations and not from sources such as newspapers, parents,
or friends.\footnote{As we explain in Corollary \ref{cor:1}, our main result extends to
environments where some fraction of the population has access to data
about the play of others.} Therefore, patient young senders will rationally try out different
signals to see how receivers react. This implies some ``off-path''
signals that have probability zero in a given equilibrium will occur
with small but positive probabilities in the steady states that approximate
it, so we can use Bayes rule to derive restrictions on the receivers'
typical posterior beliefs following these rare but positive-probability
observations. Moreover, differences in the payoff functions of the
sender types lead them to experiment in different ways. As a consequence,
we can prove that patiently stable steady states must be a subset
of Nash equilibria where the receiver responds to beliefs about the
sender's type that respect a \emph{type compatibility }condition.
This provides a learning-based justification for eliminating certain
``unintuitive'' equilibria in signaling games. These results also
suggest that learning theory could be used to control the rates of
off-path play and hence generate equilibrium refinements in other
games.

\subsection{A Toy Example}

To give some of the intuition for our general results, we study a
particular stage game embedded in an artificially simple learning
model, and explain why optimal experimentation rules out a seemingly
unappealing equilibrium outcome. Consider the following signaling
game: the sender is either the high type $\theta_{H}$ or the low
type $\theta_{L}$, both equally likely. The sender chooses between
two signals,\textbf{ $s\in\{$In}, \textbf{Out}\}. If the sender plays
\textbf{Out}, the game ends and both parties get 0 payoff. If the
sender plays \textbf{In}, the receiver then chooses an action $a\in\{$\textbf{Up,
Down}\}. Payoffs following the signal \textbf{In} depend on the sender's
type and receiver's action, as in the following matrix. 
\begin{center}
\begin{center}%
\begin{tabular}{|c|c|c|}
\hline 
signal: \textbf{In} & action: \textbf{Up} & action:\textbf{ Down}\tabularnewline
\hline 
\hline 
type: $\theta_{H}$ & $2,2$ & $-2,0$\tabularnewline
\hline 
type: $\theta_{L}$ & $1,-1$ & $-3,0$\tabularnewline
\hline 
\end{tabular}\end{center}
\par\end{center}

Both sender types prefer (\textbf{In},\textbf{ Up})\textbf{ }to \textbf{Out
}and prefer \textbf{Out }to (\textbf{In, Down}), while the receiver
prefers \textbf{Up} over \textbf{Down} after signal \textbf{In} if
he believes there is greater than $\frac{1}{3}$ chance that the sender
has type $\theta_{H}$. 

This game has a perfect Bayesian equilibrium (PBE) where both types
choose \textbf{Out} and the receiver plays \textbf{Down} after \textbf{In},
sustained by the belief that anyone who sends \textbf{In }has probability
$p\le\frac{1}{3}$ of being $\theta_{H}$. This updating requires
the receiver to interpret the off-path \textbf{In} as a signal that
the sender is more likely to be $\theta_{L}$, even though $\theta_{H}$
gets 1 more utility than $\theta_{L}$ does from\textbf{ In }regardless
of the receiver's strategy. So, ``both \textbf{Out}'' is eliminated
by the D1 criterion.\footnote{Any receiver play at the off-path signal \textbf{In} that makes it
weakly optimal for $\theta_{L}$ to deviate to \textbf{In} would also
make it strictly optimal for $\theta_{H}$ to deviate. \citet{cho_signaling_1987}'s
D1 criterion therefore requires the receiver to put 0 probability
on $\theta=\theta_{L}$ after \textbf{In}. However, the PBE passes
their Intuitive Criterion.} 

Now suppose there are three infinitely lived agents: $\theta_{H},$
$\theta_{L},$ and R (for receiver). Suppose that in each period $t\in\{1,2,3,...\}$,
the three agents play a simultaneous-move game, where each sender
type $\theta_{i}$ chooses a signal $s_{t}^{i}$, and R chooses a
single action $a_{t}$ to use against both of the senders. (This is
a deterministic analog of the receiver randomly matching with each
type with probability 1/2 without knowing the sender's type.) At the
end of period $t$, R observes the signal choices of both types, while
$\theta_{i}$ observes $a_{t}$ if and only if $s_{t}^{i}=\text{\textbf{In}}$.
That is, each agent only learns from his/her personal experience;
by choosing the ``outside option'' \textbf{Out}, the sender does
not learn how the receiver would have responded to signal \textbf{In}
that period. 

Agents think that each opponent is committed to some mixed strategy
of the stage game and plays this strategy each period, regardless
of their observations of past play: that is, all agents are strategically
myopic in the sense of \citet{fudenberg1993learning} and do not try
to influence the distribution of strategies they will face in future
rounds. At the beginning of $t=1$, each type $\theta_{i}$ is endowed
with a $\text{Beta}(c_{U},c_{D})$ prior about the probability that
R responds to \textbf{In} with \textbf{Up}, with $c_{D}>c_{U}>0,$
so they assign higher probability to \textbf{Down }than to \textbf{Up.}
R starts with two independent priors $\text{Beta}(c_{I}^{H},c_{O}^{H})$
and $\text{Beta}(c_{I}^{L},c_{O}^{L})$ about the probabilities that
$\theta_{H}$ and $\theta_{L}$ choose \textbf{In} each period, where
we only assume $c_{I}^{H},c_{O}^{H},c_{I}^{L},c_{O}^{L}>0$. The independence
assumption means that R does not learn about the behavior of one type
from the play of the other. 

Agents discount payoffs in future periods at rate $0\le\delta<1$
and choose a signal or action each period so as to maximize expected
discounted payoffs. Because expected utility maximizing agents never
strictly prefer to randomize, each of them has a deterministic optimal
policy, so that each discount factor $\delta$ induces a deterministic
infinite history of play $(s_{t}^{H},s_{t}^{L},a_{t})_{t=1}^{\infty}$$=:Y(\delta)$.
When $\delta=0$, the agents play myopically every period, and because
of our assumption that $c_{D}>c_{U}$, both types choose \textbf{Out}
in $t=1$. They thus gain no information about R's play, do not update
their beliefs, and continue playing \textbf{Out} in every future period.
So, the unintuitive ``both \textbf{Out}'' PBE is the learning outcome
when agents are sufficiently impatient. However, we can show for all
large enough $\delta$, that eventually behavior converges to R playing
\textbf{Up} and $\theta_{H}$ playing \textbf{In} each period.\footnote{In practice, the required patience level is not unreasonably high.
When $c_{D}=1.1$ $c_{U}=1,$ $c_{I}^{H}=c_{O}^{L}=1,$ and $c_{O}^{H}=c_{I}^{L}=3,$
for example, $\delta=0$ yields the pathological PBE as the long-run
outcome, but when $\delta\ge0.92$ the long-run outcome involves $s_{t}^{H}=\textbf{In}$
and $a_{t}=\text{\textbf{Up}}$. } 

We give a sketch of the argument, beginning with characterizing agents'
optimal behavior each period. R observes the same information regardless
of his play, so he plays myopically under any $\delta$. Let $p(h_{t})$
be R's Bayesian posterior belief about the probability that an \textbf{In}
sender has type $\theta_{H}$, given history $h_{t}$. Then $a_{t+1}=\text{\textbf{Up}}$
if $p(h_{t})>\frac{1}{3}$ and $a_{t+1}=\text{\textbf{Down}}$ if
$p(h_{t})<\frac{1}{3}$. 

Now we turn to $\theta_{i},$ whose problem involves active experimentation.
Formally, the dynamic optimization problem facing $\theta_{i}$ is
a one-armed Bernoulli bandit. Choosing $s_{t}^{i}=\text{\textbf{Out}}$
is equivalent to taking the safe outside option while choosing $s_{t}^{i}=\text{\textbf{In}}$
is equivalent to pulling the risky arm and getting a payoff depending
on whether the pull results in a success ($a_{t}=\text{\textbf{Up}}$)
or a failure $(a_{t}=\text{\textbf{Down}}$). The optimal policy for
$\theta_{i}$ involves the Gittins index (defined later in Equation
(\ref{eq:gittins_def})). Type $\theta_{i}$ plays \textbf{In} at
those histories where \textbf{In} has a positive Gittins index. 

Once a type chooses to play \textbf{Out }in some period, she receives
no further information and will continue to play \textbf{Out }in all
subsequent periods. Denote the period in $Y(\delta)$ that $\theta_{i}$
first switches from \textbf{In} to \textbf{Out} as $T(i,\delta)\in\mathbb{N}\cup\{\infty\}$,
where $T(i,\delta)=\infty$ means $\theta_{i}$ plays \textbf{In}
forever. The argument that learning eliminates pooling on \textbf{Out}
follows from three observations:

\textbf{Observation 1}. \emph{The high type switches to }\textbf{\emph{Out}}\emph{
later than the low type does, that is,} $T(H,\delta)\ge T(L,\delta)$.
To see why, suppose by way of contradiction that $T(H,\delta)<T(L,\delta).$
Then, in period $t=T(H,\delta),$ both $\theta_{H}$ and $\theta_{L}$
have played \textbf{In} until now and have seen the same history,
so they hold the same belief about R's play. Yet $\theta_{H}$ chooses
\textbf{Out} at this history while $\theta_{L}$ chooses \textbf{In},
meaning $\theta_{H}$ has a negative Gittins index for \textbf{In
}while $\theta_{L}$ has a positive one. This is impossible, since
$\theta_{H}$'s payoff from \textbf{In} is always 1 higher than that
of $\theta_{L}$, so $\theta_{H}$'s index for \textbf{In} is also
always 1 higher than that of $\theta_{L}$ when the two types have
the same belief about R's play. 

\textbf{Observation 2}. \emph{As the high type becomes patient, she
experiments with }\textbf{\emph{In}}\emph{ arbitrarily many times},
that is, $\lim_{\delta\to1}T(H,\delta)=\infty.$ This follows because
for any fixed full-support prior belief of $\theta_{H}$ about R's
mixed strategy, the Gittins index for \textbf{In} stays close to the
``success payoff'' of 2 for a length of time that grows to infinity
as $\delta\to1$, even in the worst case where R plays \textbf{Down}
in every period.

\textbf{Observation 3}. \emph{If the high type plays }\textbf{\emph{In}}\emph{
sufficiently many times and more often than the low type does, then
eventually R will believe that }\textbf{\emph{In}}\emph{ senders have
greater than $\frac{1}{3}$ chance of being $\theta_{H}$,} that is,
there exists $\bar{N}\in\mathbb{N}$ so that $p(h_{T})>\frac{1}{3}$
for any history $h_{T}$ where (i) $\theta_{H}$ played \textbf{In}
at least $\bar{N}$ times and (ii) $\theta_{L}$ played \textbf{In}
no more than $\theta_{H}$ did. This follows from the fact that R's
belief about $\theta_{i}$'s play after $n_{I}^{i}$ instances of
\textbf{In} and $n_{O}^{i}$ instances of \textbf{Out} is $\text{Beta}(c_{I}^{i}+n_{I}^{i},c_{O}^{i}+n_{O}^{i})$. 

From Observation 2, we see that $T(H,\delta)$ is larger than the
$\bar{N}$ of Observation 3 when $\delta$ is sufficiently large.
The history up to period $t$ for any $t\ge\bar{N}$ will therefore
contain at least $\bar{N}$ periods of $\theta_{H}$ playing \textbf{In}
(namely, the very first $\bar{N}$ periods of the game), and by Observation
1 $\theta_{L}$ will have played \textbf{In} no more than $\theta_{H}$
did in this history. So by Observation 3, $p(h_{t})>\frac{1}{3}$
for $t\ge\bar{N}$, meaning $a_{t}=\textbf{Up}$ for $t\ge\bar{N}$.
Since $s_{t}^{H}=\textbf{In}$ for all $t\le\bar{N}$ and observing
\textbf{Up} increases the Gittins index of \textbf{In}, the high type
must always play \textbf{In}. This means $\lim_{t\to\infty}s_{t}^{H}=\textbf{In}$
and $\lim_{t\to\infty}a_{t}=\textbf{Up}$ for large $\delta<1$. 

In this simple learning model, agents are patient and face the same
opponents many times but do not try to influence their future play.
Furthermore, agents believe that opponents' play is stationary but
it changes markedly over time. Finally, the analysis was greatly simplified
because there are only two signals, one of which gives a certain payoff
to the senders; this makes playing \textbf{Out }an absorbing state
and, together with the assumption of Beta priors, lets us explicitly
calculate how the system evolves. This paper's focus is on general
signaling games embedded in a learning model with large populations
and anonymous random matching, eliminating repeated-game effects.
We focus on steady states of the model, where the stationary assumption
is satisfied. Also, we relax the Beta prior assumption and allow learners
to have fairly general non-doctrinaire priors. Many results about
the steady-state model, however, have analogs in the simple model
above. 

Intuitively, $\theta_{H}$ is ``more compatible'' with signal \textbf{In}
than $\theta_{L}$. Definition \ref{def:more_compatible} formalizes
this relation in general signaling games. Observation 1 corresponds
to Lemma \ref{lem:compatible_comonotonic}, which shows that whenever
one type is more compatible than another with a signal, the more compatible
type sends the signal more often. Observation 2 corresponds to Lemma
\ref{lem:nondom_message}, which says a sufficiently patient and long-lived
sender type will experiment many times with all signals that have
the potential to strictly improve that type's equilibrium payoff.
Observation 3 corresponds to Lemma \ref{lem:receiver_learning}, which
says receivers can eventually learn the compatibility relation associated
with each signal, provided senders' play respects the relation and
the more compatible type experiments enough with the signal. Lemmas
\ref{lem:compatible_comonotonic}, \ref{lem:receiver_learning}, and
\ref{lem:nondom_message} are combined to prove the main result of
the paper (Theorem \ref{thm:PS_is_compatible}), a learning-based
refinement in general signaling games. 

\subsection{Outline and Overview of Results}

Section \ref{sec:Model} lays out the notation we will use for signaling
games and introduces our learning model. Section \ref{sec:sender_side}
introduces the Gittins index, which we use to analyze the senders'
learning problem. It also defines type compatibility, which is a partial
order that drives our results. We say that type $\theta^{'}$ is \emph{more
type-compatible with signal $s'$} than type $\theta^{''}$ if, whenever
$s^{'}$ is a weak best response for $\theta^{''}$ against some receiver
behavior strategy, it is a strict best response for $\theta^{'}$
against the same strategy. To relate this static definition to the
senders' optimal dynamic learning behavior, we show that, under our
assumptions, the senders' learning problem is formally a multi-armed
bandit, so the optimal policy of each type is characterized by the
Gittins index. Theorem \ref{thm:index} shows that the compatibility
order on types is equivalent to an order on their Gittins indices:
$\theta^{'}\text{ }$ is more type-compatible with signal $s^{'}$
than type $\theta^{''}$ if and only if, whenever $s^{'}$ has the
(weakly) highest Gittins index for $\theta^{''}$, it has the strictly
highest index for $\theta{}^{'}$, provided the two types hold the
same beliefs and have the same discount factor.

Section \ref{sec:aggregate} studies the aggregate behavior of the
sender and receiver populations. There we define and characterize
the \emph{aggregate responses} of the senders and of the receivers,
which are the analogs of the best-response functions in the one-shot
signaling game. First, we use a coupling argument to extend Theorem
\ref{thm:index} to the aggregate sender behavior, proving that types
who are more compatible with a signal send it more often in aggregate
(Lemma \ref{lem:compatible_comonotonic}). Then we turn to the receivers.
Intuitively, we would expect that when receivers are long-lived, most
of them will have beliefs that respect type compatibility, and we
show that this is the case. More precisely, we show that most receivers
best respond to a posterior belief whose likelihood ratio of $\theta^{'}$
to $\theta^{''}$ dominates the prior likelihood ratio of these two
types whenever they observe a signal $s$ which is more type-compatible
with $\theta^{'}$ than $\theta^{''}$. Lemma \ref{lem:receiver_learning}
shows this is true for any signal that is sent ``frequently enough''
relative to the receivers' expected lifespan, using a result of \citet*{fudenberg_he_imhof_2016}
on updating posteriors after rare events. 

Finally, Section \ref{sec:two_sided} combines the earlier results
to characterize the steady states of the learning model, which can
be viewed as pairs of mutual aggregate responses, analogous to the
definition of Nash equilibrium. We start by proving Lemma \ref{lem:nondom_message},
which shows that any signal that is not \emph{weakly equilibrium dominated}
(see Definition \ref{def:J}) gets sent ``frequently enough'' in
steady state when senders are sufficiently patient and long lived.
Combining the three lemmas discussed above, we establish our main
result: any patiently stable steady state must be a Nash equilibrium
satisfying the additional restriction that the receivers best respond
to certain \emph{admissible beliefs} after every off-path signal (Theorem
\ref{thm:PS_is_compatible}). 

As an example, consider \citet{cho_signaling_1987}'s beer-quiche
game, where it is easy to verify that the strong type is more compatible
with \textbf{Beer} than the weak type. Our results imply that the
strong types will in aggregate send this signal at least as often
as the weak types do, and that a very patient strong type will experiment
with it ``many times.'' As a consequence, when senders are patient,
long-lived receivers are unlikely to revise the probability of the
strong type downwards following an observation of \textbf{Beer}. Thus,
the ``both types eat quiche'' equilibrium is not a patiently stable
steady state of the learning model, as it would require receivers
to interpret \textbf{Beer} as a signal that the sender is weak. 

Finally, Theorem \ref{thm:eqm_dominated_types} provides a stronger
implication of patient stability in generic pure-strategy equilibria,
showing that off-path beliefs must assign probability zero to types
that are equilibrium dominated in the sense of \citet{cho_signaling_1987}. 

\subsection{Related Work}

\citet{fudenbergKreps1988,fudenbergKreps1994learning,fudenberg_learning_1995}
pointed out that experimentation plays an important role in determining
learning outcomes in extensive-form games. As in \citet{fudenberg1993learning},
they studied a model with a single infinitely-lived and strategically
myopic agent in each player role who acts as if the opponent's play
is stationary. Because these models involved accumulating information
over time, they did not have steady states. Our work is closer to
that of \citet{fudenberg_steady_1993} and \citet{fudenberg_superstition_2006}
which also studied learning by Bayesian agents in a large population
who believe that society is in a steady state. A key issue in this
work, and more generally in studying learning in extensive-form games,
is characterizing how much agents will experiment with myopically
suboptimal actions. If agents do not experiment at all, then non-Nash
equilibria can persist, because players can maintain incorrect but
self-confirming beliefs about off-path play. \citet{fudenberg_steady_1993}
showed that patient long-lived agents will experiment enough at their
on-path information sets to learn if they have any profitable deviations,
thus ruling out steady states that are not Nash equilibria. However,
more experimentation than that is needed for learning to generate
the sharper predictions associated with backward induction and sequential
equilibrium. \citet{fudenberg_superstition_2006} showed that patient
rational agents need not do enough experimentation to imply backwards
induction in games of perfect information. Later on, we say about
how the models and proofs of those papers differ from ours.

This paper is also related to the Bayesian learning models of \citet{kalai_rational_1993},
which studied two-player games with one agent on each side, so that
every self-confirming equilibrium is path-equivalent to a Nash equilibrium,
and \citet{esponda_berknash_2016}, which allowed agents to experiment
but did not characterize when and how this occurs. It is also related
to the literature on boundedly rational experimentation in extensive-form
games (e.g. \citet*{jehiel_learning_2005}, \citet*{laslier_stubborn_2014}),
where the experimentation rules of the agents are exogenously specified.
We assume that each sender's type is fixed at birth, as opposed to
being i.i.d. over time. \citet*{dekel_learning_2004} showed some
of the differences this can make using various equilibrium concepts,
but they did not develop an explicit model of non-equilibrium learning. 

For simplicity, we assume here that agents\emph{ }do not know the
payoffs of other players and have full support priors over the opposing
side's behavior strategies. Our companion paper \citet{FudenbergHe2017TCE}
supposed that players assign zero probability to dominated strategies
of their opponents, as in the Intuitive Criterion \citep{cho_signaling_1987},
divine equilibrium \citep{banks_equilibrium_1987}, and rationalizable
self-confirming equilibrium \citep*{dekel1999payoff}. There, we analyzed
how the resulting micro-founded equilibrium refinement compares to
those in past work. 

\section{Model \label{sec:Model}}

\subsection{Signaling Game Notation}

A \emph{signaling game} has two players, a sender (player 1, ``she'')
and a receiver (player 2, ``he''). The sender's type is drawn from
a finite set $\Theta$ according to a prior $\lambda\in\Delta(\Theta$)
with $\lambda(\theta)>0$ for all $\theta$.\footnote{Here and subsequently, $\Delta(X)$ denotes the collection of probability
distributions on the set $X.$} There is a finite set $S$ of signals for the sender and a finite
set $A$ of actions for the receiver.\footnote{To lighten notation we assume that the same set of actions is feasible
following any signal. This is without loss of generality for our results
as we could let the receiver have very negative payoffs when he responds
to a signal with an ``impossible'' action.} The utility functions of the sender and receiver are $u_{1}:\Theta\times S\times A\to\mathbb{R}$
and $u_{2}:\Theta\times S\times A\to\mathbb{R}$ respectively.

When the game is played, the sender knows her type and sends a signal
$s\in S$ to the receiver. The receiver observes the signal, then
responds with an action $a\in A$. Finally, payoffs are realized. 

A \emph{behavior strategy for the sender }$\pi_{1}=(\pi_{1}(\cdot|\theta))_{\theta\in\Theta}$
is a type-contingent mixture over signals $S$. Write $\Pi_{1}$ for
the set of all sender behavior strategies. 

A\emph{ behavior strategy for the receiver} $\pi_{2}=(\pi_{2}(\cdot|s))_{s\in S}$
is a signal-contingent mixture over actions $A$. Write $\Pi_{2}$
for the set of all receiver behavior strategies. 

\subsection{Learning by Individual Agents \label{subsec:learning1} }

We now build a learning model with a given signaling game as the stage
game. In this subsection, we explain an individual agent's learning
problem. In the next subsection, we complete the learning model by
describing a society of learning agents who are randomly matched to
play the signaling game every period. 

Time is discrete and all agents are rational Bayesians with geometrically
distributed lifetimes. They survive between periods with probability
$0\le\gamma<1$ and further discount future utility flows by $0\le\delta<1$,
so their objective is to maximize the expected value of $\sum_{t=0}^{\infty}(\gamma\delta)^{t}\cdot u_{t}$.
Here, $0\le\gamma\delta<1$ is the effective discount factor, and
$u_{t}$ is the payoff $t$ periods from today. 

At birth, each agent is assigned a role in the signaling game: either
as a sender with type $\theta$ or as a receiver. Agents know their
role, which is fixed for life. Every period, each agent is randomly
and anonymously matched with an opponent to play the signaling game,
and the game's outcome determines the agent's payoff that period.
At the end of each period, agents observe the outcomes of their own
matches, that is, the signal sent, the action played in response,
and the sender's type. They do not observe the identity, age, or past
experiences of their opponents, nor does the sender observe how the
receiver would have reacted to a different signal.\footnote{The receiver's payoff reveals the sender's type for generic assignments
of payoffs to terminal nodes. If the receiver's payoff function is
independent of the sender's type, his beliefs about it are irrelevant.
If the receiver does care about the sender's type but observes neither
the sender's type nor his own realized payoff, a great many outcomes
can persist, as in \citet*{dekel_learning_2004}. } Agents update their beliefs and play the signaling game again with
new random opponents next period, provided they are still alive. 

Agents believe they face a fixed but unknown distribution of opponents'
aggregate play, so they believe that their observations will be exchangeable.
We feel that this is a plausible first hypothesis in many situations,
so we expect that agents will maintain their belief in stationarity
when it is approximately correct, but will reject it given clear evidence
to the contrary, as when there is a strong time trend or a high-frequency
cycle. The environment will indeed be constant in the steady states
that we analyze. 

Formally, each sender is born with a prior density function over the
aggregate behavior strategy of the receivers, $g_{1}:\Pi_{2}\to\mathbb{R}_{+,}$
which integrates to 1. Similarly, each receiver is born with a prior
density over the sender's behavior strategies\footnote{Note that the agent's prior belief is over opponents' \emph{aggregate}
play (i.e. $\Pi_{1}$ or $\Pi_{2}$) and not over the prevailing distribution
of behavior strategies in the opponent population (i.e. $\Delta(\Pi_{2})$
or $\Delta(\Pi_{1})$), since under our assumption of anonymous random
matching, these are observationally equivalent for our agents. For
instance, a receiver cannot distinguish between a society where all
type $\theta$ randomize 50-50 between signals $s_{1}$ and $s_{2}$
each period, and another society where half of the type $\theta$
always play $s_{1}$ while the other half always plays $s_{2}$. Note
also that because agents believe the system is in a steady state,
they do not care about calendar time and do not have beliefs about
it. \citet*{fudenbergKreps1994learning} suppose that agents append
a non-Bayesian statistical test of whether their observations are
exchangeable to a Bayesian model that presumes exchangeability. }, $g_{2}:\Pi_{1}\to\mathbb{R}_{+}$. We denote the marginal distribution
of $g_{1}$ on signal $s$ as $g_{1}^{(s)},$ so that $g_{1}^{(s)}(\pi_{2}(\cdot|s))$
is the density of the new senders' prior over how receivers respond
to signal $s$. Similarly, we denote the $\theta$ marginal of $g_{2}$
as $g_{2}^{(\theta)}$, so that $g_{2}^{(\theta)}(\pi_{1}(\cdot|\theta))$
is the new receivers' prior density over $\pi_{1}(\cdot|\theta)\in\Delta(S)$.
\renewcommand{\theenumi}{(\roman{enumi})} 

It is important to remember that $g_{1}$ and $g_{2}$ are beliefs
over opponents' strategies, but not strategies themselves. A new sender
expects the response to $s$ to be $\int\pi_{2}(\cdot|s)\cdot g_{1}(\pi_{2})d\pi_{2}$
while a new receiver expects type $\theta$ to play $\int\pi_{1}(\cdot|\theta)\cdot g_{2}(\pi_{1})d\pi_{1}$. 

We now state a regularity assumption on the agents' priors that will
be maintained throughout. 
\begin{defn}
\label{def:regular_prior} A prior $g=(g_{1},g_{2})$ is \textbf{regular}
if 
\end{defn}
\begin{enumerate}
\item {[}\emph{independence}{]} $g_{1}(\pi_{2})=\underset{s\in S}{\prod}g_{1}^{(s)}(\pi_{2}(\cdot|s))$
and $g_{2}(\pi_{1})=\underset{\theta\in\Theta}{\prod}g_{2}^{(\theta)}(\pi_{1}(\cdot|\theta))$. 
\item {[}$g_{1}$ \emph{non-doctrinaire}{]} $g_{1}$ is continuous and strictly
positive on the interior of $\Pi_{2}.$
\item {[}$g_{2}$ \emph{nice}{]} for each type $\theta$$,$ there are positive
constants $\left(\alpha_{s}^{(\theta)}\right)_{s\in S}$ such that
\[
\pi_{1}(\cdot|\theta)\mapsto\frac{g_{2}^{(\theta)}(\pi_{1}(\cdot|\theta))}{\prod_{s\in S}\pi_{1}(s|\theta){}^{\alpha_{s}^{(\theta)}-1}}
\]
is uniformly continuous and bounded away from zero on the relative
interior of $\Pi_{1}^{(\theta)}$, the set of behavior strategies
of type $\theta$. 
\end{enumerate}
Independence ensures that a receiver does not learn how type $\theta$
plays by observing the behavior of some other type $\theta^{'}\ne\theta$,
and that a sender does not learn how receivers react to signal $s$
by experimenting with some other signal $s^{'}\ne s$. For example,
this means in \citet{cho_signaling_1987}'s beer-quiche game that
the sender does not learn how receivers respond to beer by eating
quiche.\footnote{One could imagine learning environments where the senders believe
that the responses to various signals are correlated, but independence
is a natural special case.} The non-doctrinaire nature of $g_{1}$ and $g_{2}$ implies that
the agents never see an observation that they assigned zero prior
probability, so that they have a well-defined optimization problem
after any history. Non-doctrinaire priors also imply that a large
enough data set can outweigh prior beliefs \citep{diaconis_uniform_1990}.
The niceness assumption in (iii) ensures that $g_{2}$ behaves like
a power function near the boundary of $\Pi_{1}$. Any density that
is strictly positive on $\Pi_{1}$ satisfies this condition, as does
the Dirichlet distribution, which is the prior associated with fictitious
play \citep{fudenberg1993learning}.

The set of histories for an age $t$ sender of type $\theta$ is $Y_{\theta}[t]:=(S\times A)^{t}$,
where each period, the history records the signal sent and the action
that her receiver opponent took in response. The set of all histories
for a type $\theta$ is the union $Y_{\theta}\coloneqq\bigcup_{t=0}^{\infty}Y_{\theta}[t]$.
The dynamic optimization problem of type $\theta$ has an optimal
policy function $\sigma_{\theta}:Y_{\theta}\to S$, where $\sigma_{\theta}(y_{\theta})$
is the signal that a type $\theta$ with history $y_{\theta}$ would
send the next time she plays the signaling game. Analogously, the
set of histories for an age $t$ receiver is $Y_{2}[t]:=(\Theta\times S)^{t}$,
where each period, the history records the type of his sender opponent
and the signal that she sent. The set of all receiver histories is
the union $Y_{2}\coloneqq\bigcup_{t=0}^{\infty}Y_{2}[t]$. The receiver's
learning problem admits an optimal policy function $\sigma_{2}:Y_{2}\to A^{S}$,
where $\sigma_{2}(y_{2})$ is the pure strategy that a receiver with
history $y_{2}$ would commit to next time he plays the game.\footnote{Because our agents are expected-utility maximizers, it is without
loss of generality to assume each agent uses a deterministic policy
rule. If more than one such rule exists, we fix one arbitrarily. Of
course, the optimal policies $\sigma_{\theta}$ and $\sigma_{2}$
depend on the prior $g$ as well as the effective discount factor
$\delta\gamma$. Where no confusion arises, we suppress these dependencies. }

\subsection{Random Matching and Aggregate Play \label{subsec:learning2} }

We analyze learning in a deterministic stationary model with a continuum
of agents, as in \citet{fudenberg_steady_1993,fudenberg_superstition_2006}.
One innovation is that we let lifetimes follow a geometric distribution
instead of the finite and deterministic lifetimes assumed in those
earlier papers, so that we can use the Gittins index. 

The society contains a unit mass of agents in the role of receivers
and mass $\lambda(\theta)$ in the role of type $\theta$ for each
$\theta\in\Theta$. As described in Subsection \ref{subsec:learning1},
each agent has $0\le\gamma<1$ chance of surviving at the end of each
period and complementary chance $1-\gamma$ of dying. To preserve
population sizes, $(1-\gamma)$ new receivers and $\lambda(\theta)(1-\gamma)$
new type $\theta$ are born into the society every period. 

Each period, agents in the society are matched uniformly at random
to play the signaling game. In the spirit of the law of large numbers,
each sender has probability $(1-\gamma)\gamma^{t}$ of matching with
a receiver of age $t$, while each receiver has probability $\lambda(\theta)(1-\gamma)\gamma^{t}$
of matching with a type $\theta$ of age $t.$

A \emph{state }$\psi$ of the learning model is described by the mass
of agents with each possible history. We write it as 
\[
\psi\in\left(\times_{\theta\in\Theta}\Delta(Y_{\theta})\right)\times\Delta(Y_{2}).
\]

We refer to the components of a state $\psi$ by $\psi_{\theta}\in\Delta(Y_{\theta})$
and $\psi_{2}\in\Delta(Y_{2})$. 

Given the agents' optimal policies, each possible history for an agent
completely determines how that agent will play in their next match.
The sender policy functions $\sigma_{\theta}$ are maps from sender
histories to signals,\footnote{Remember that we have fixed deterministic policy functions. }
so they naturally extend to maps from distributions over sender histories
to distributions over signals. That is, given the policy function
$\sigma_{\theta}$, each state $\psi$ induces an aggregate behavior
strategy $\sigma_{\theta}(\psi_{\theta})\in\Delta(S)$ for each type
$\theta$ population, where we extend the domain of $\sigma_{\theta}$
from $Y_{\theta}$ to $\Delta(Y_{\theta})$ in the natural way: 
\begin{equation}
\sigma_{\theta}(\psi_{\theta})(s)\coloneqq\psi_{\theta}\left\{ y_{\theta}\in Y_{\theta}:\sigma_{\theta}(y_{\theta})=s\right\} .\label{eq:induced_sender_strat}
\end{equation}

Similarly, state $\psi$ and the optimal receiver policy $\sigma_{2}$
together induce an aggregate behavior strategy $\sigma_{2}(\psi_{2})$
for the receiver population, where 

\[
\sigma_{2}(\psi_{2})(a|s)\coloneqq\psi_{2}\left\{ y_{2}\in Y_{2}:\sigma_{2}(y_{2})(s)=a\right\} .
\]

We will study the steady states of this learning model, to be defined
more precisely in Section \ref{sec:two_sided}. Loosely speaking,
a steady state is a state $\psi$ that reproduces itself indefinitely
when agents use their optimal policies. Put another way, a steady
state induces a time-invariant distribution over how the signaling
game is played in the society. Suppose society is at steady state
today and we measure what fraction of type $\theta$ sent a certain
signal $s$ in today's matches. After all agents modify their strategies
based on their updated beliefs and all births and deaths take place,
the fraction of type $\theta$ playing $s$ in the matches tomorrow
will be the same as today. 

\section{Senders' Optimal Policies and Type Compatibility \label{sec:sender_side}}

This section studies the senders' learning problem. We will prove
that differences in the payoff structures of the various sender types
generate certain restrictions on their behavior in the learning model.
Subsection \ref{subsec:sender_problem} notes that the senders face
a multi-armed bandit, so the Gittins index characterizes their optimal
policies, and shows how to relate the Gittins index of a signal to
the expected sender payoff versus a particular mixed strategy of the
receiver. In Subsection \ref{subsec:compatible}, we define \emph{type
compatibility}, which formalizes what it means for type $\theta^{'}$
to be more ``compatible'' with a given signal $s$ than type $\theta^{''}$
is. The definition of type compatibility is static, in the sense that
it depends only on the two types' payoff functions in the one-shot
signaling game. Subsection \ref{subsec:gittins} relates type compatibility
to the Gittins index, which applies to the dynamic learning model.
Lemma \ref{lem:compatible_comonotonic}in Section \ref{sec:aggregate}
uses this relationship to show that if type $\theta^{'}$ is more
compatible with signal $s$ than type $\theta^{''},$ then faced with
any fixed distribution of receiver play the type $\theta^{'}$ population
sends $s$ more often in the aggregate than the type $\theta^{''}$
population does. 

\subsection{Optimal Policies and Multi-Armed Bandits\label{subsec:sender_problem} }

Each type $\theta$ sender thinks she is facing a fixed but unknown
aggregate receiver behavior strategy $\pi_{2}$, so each period when
she sends signal $s$, she believes that the response is drawn from
some $\pi_{2}(\cdot|s)\in\Delta(A)$, i.i.d. across periods. Because
her beliefs about the responses to the various signals are independent,
her problem is equivalent to a discounted multi-armed bandit, with
signals $s\in S$ as the arms, where the rewards of arm $s$ are distributed
according to $u_{1}(\theta,s,\pi_{2}(\cdot|s))$. 

Let $\nu_{s}\in\Delta(\Delta(A))$ be a belief over the space of mixed
replies to signal $s$, and let $\nu=(\nu_{s})_{s\in S}$ be a profile
of such beliefs. Write $I(\theta,s,\nu,\beta)$ for the Gittins index
of signal $s$ for type $\theta$, with beliefs $\nu$ over receiver's
play after various signals and with effective discount factor $\beta=\delta\gamma$,
so that 

\begin{equation}
I(\theta,s,\nu,\beta)\coloneqq\ \sup_{\tau>0}\dfrac{\mathbb{E}_{\nu_{s}}\left\{ \sum_{t=0}^{\tau-1}\beta^{t}\cdot u_{1}(\theta,s,a_{s}(t))\right\} }{\mathbb{E}_{\nu_{s}}\left\{ \sum_{t=0}^{\tau-1}\beta^{t}\right\} }.\label{eq:gittins_def}
\end{equation}

Here $a_{s}(t)$ is the receiver's response that the sender observes
the $t$-th time she sends signal $s$, $\tau$ is a stopping time\footnote{That is, whether or not $\tau=t$ depends only on the realizations
of $a_{s}(0),a_{s}(1),...,a_{s}(t-1)$. } and the expectation $\mathbb{E}_{\nu_{s}}$ over the sequence of
responses $\{a_{s}(t)\}_{t\ge0}$ depends on the sender's belief $\nu_{s}$
about responses to signal $s$.\footnote{The Gittins index can be interpreted as the value of an auxiliary
optimization problem, where type $\theta$ chooses each period to
either send signal $s$ and obtain a payoff according to a random
receiver action drawn according to $\pi_{2}(\cdot|s)$, or to stop
forever. The objective of the auxiliary problem is to maximize the
per-period expected discounted payoff until stopping, as the numerator
of Equation (\ref{eq:gittins_def}) describes the expected discounted
sum of payoffs until stopping while the denominator shows the expected
discounted number of periods until stopping. } 

The Gittins index theorem \citep{gittins1979bandit} implies that
after every positive-probability history $y_{\theta},$ the optimal
policy $\sigma_{\theta}$ for a sender of type $\theta$ sends the
signal that has the highest Gittins index for that type under the
profile of posterior beliefs $(\nu_{s})_{s\in S}$ that is induced
by $y_{\theta}$. 

Importantly, we can reformulate the objective function defining the
Gittins index in Equation (\ref{eq:gittins_def}), linking it to the
one-shot signaling game payoff structure. 
\begin{lem}
\label{lem:static} For every signal $s$, stopping time $\tau$,
belief $\nu_{s}$, and discount factor $\beta,$ there exists $\pi_{2,s}(\tau,\nu_{s},\beta)\in\Delta(A)$
so that for every $\theta$,

\[
\dfrac{\mathbb{E}_{\nu_{s}}\left\{ \sum_{t=0}^{\tau-1}\beta^{t}\cdot u_{1}(\theta,s,a_{s}(t))\right\} }{\mathbb{E}_{\nu_{s}}\left\{ \sum_{t=0}^{\tau-1}\beta^{t}\right\} }=u_{1}(\theta,s,\pi_{2,s}(\tau,\nu_{s},\beta))
\]
\end{lem}
That is to say, when the stopping problem in Equation (\ref{eq:gittins_def})
is evaluated at an arbitrary stopping time $\tau,$ the payoff is
equal to sender's expected utility from playing $s$ against the receiver
strategy $\pi_{2,s}(\tau,\nu_{s},\beta)$ in the one-shot signaling
game. 

The proof of Lemma \ref{lem:static} is in Appendix \ref{subsec:Proof-of-Lemma-static}
and shows how to construct $\pi_{2,s}(\tau,\nu_{s},\beta)$, which
can be interpreted as a discounted time average over the receiver
actions that are observed before stopping. To illustrate the construction,
suppose $\nu_{s}$ is supported on two pure receiver strategies after
$s$: either $\pi_{2}(a^{'}|s)=1$ or $\pi_{2}(a^{''}|s)=1,$ with
both strategies equally likely. Suppose also $u_{1}(\theta,s,a^{'})>u_{1}(\theta,s,a^{''}).$
Consider the stopping time $\tau$ that specifies stopping after the
first time the receiver plays $a^{''}$. Then the discounted time
average frequency of $a^{''}$ is:
\[
\frac{\sum_{t=0}^{\infty}\beta^{t}\cdot\mathbb{P}_{\nu_{s}}[\tau\ge t\text{ and receiver plays }a^{''}\text{ in period }t]}{\sum_{t=0}^{\infty}\beta^{t}\cdot\mathbb{P}_{\nu_{s}}[\tau\ge t]}=\frac{0.5}{1+\sum_{t=1}^{\infty}\beta^{t}\cdot0.5}=\frac{1-\beta}{2-\beta}.
\]
 So $\pi_{2,s}(\tau,\nu_{s},\beta)(a^{''})=\frac{1-\beta}{2-\beta}$
and similarly we can calculate that $\pi_{2,s}(\tau,\nu_{s},\beta)(a^{'})=\frac{1}{2-\beta}$,
which shows that $\pi_{2,s}$ indeed corresponds to a mixture over
receiver actions for each $\beta$. As $\beta\to1$, this mixture
converges to the pure strategy of always playing $a^{'}$, so $u_{1}(\theta,s,\pi_{2,s}(\tau,\nu_{s},\beta))$
converges to $u_{1}(\theta,s,a^{'})$, the highest possible payoff
for type $\theta$ after $s$; this parallels the fact that as $\beta$
tends to 1, the Gittins index for $\theta$ after $s$ converges to
the highest payoff in the support of the belief $\nu_{s}$. 

\subsection{\label{subsec:compatible}Type Compatibility in Signaling Games}

We now introduce a notion of the comparative compatibility of two
types with a given signal in the one-shot signaling game. 
\begin{defn}
\label{def:more_compatible} Signal $s^{'}$ is \emph{more type-compatible}
with $\theta^{'}$ than $\theta^{''}$, written as $\theta^{'}\succ_{s^{'}}\theta^{''}$,
if for every $\pi_{2}\in\Pi_{2}$ such that 
\[
u_{1}(\theta^{''},s^{'},\pi_{2}(\cdot|s^{'}))\ge\max_{s^{''}\ne s^{'}}u_{1}(\theta^{''},s^{''},\pi_{2}(\cdot|s^{''})),
\]
we have 
\[
u_{1}(\theta^{'},s^{'},\pi_{2}(\cdot|s^{'}))>\max_{s^{''}\ne s^{'}}u_{1}(\theta^{'},s^{''},\pi_{2}(\cdot|s^{''})).
\]
\end{defn}
In words, $\theta^{'}\succ_{s^{'}}\theta^{''}$ means that whenever
$s^{'}$ is a weak best response for $\theta^{''}$ against some receiver
behavior strategy $\pi_{2}$, it is also a strict best response for
$\theta^{'}$ against $\pi_{2}$.

The following proposition says the compatibility order is transitive
and essentially asymmetric. Its proof is in Appendix \ref{subsec:Proof-of-properties_of_comp_relation}.
\begin{prop}
\label{prop:properties_of_comp_relation} 
\end{prop}
\begin{enumerate}
\item $\succ_{s^{'}}$ is transitive. 
\item Except when $s^{'}$ is either strictly dominant for both $\theta^{'}$
and $\theta^{''}$ or strictly dominated for both $\theta^{'}$ and
$\theta^{''}$, $\theta^{'}\succ_{s^{'}}\theta^{''}$ implies $\theta^{''}\not\succ_{s^{'}}\theta^{'}$. 
\end{enumerate}
To check the compatibility condition, one must consider all strategies
in $\Pi_{2},$ just as the belief restrictions in divine equilibrium
involve all the possible mixed best responses to various beliefs.
However, when the sender's utility function is separable in the sense
that $u_{1}(\theta,s,a)=v(\theta,s)+z(a),$ as in \citet{spence_job_1973}'s
job-market signaling game and in \citet{cho_signaling_1987}'s beer-quiche
game (given below), a sufficient condition for $\theta^{'}\succ_{s^{'}}\theta^{''}$
is 
\[
v(\theta^{'},s^{'})-v(\theta^{''},s^{'})>\max_{s^{''}\ne s^{'}}v(\theta^{'},s^{''})-v(\theta^{''},s^{''}).
\]
This can be interpreted as saying $s^{'}$ is the least costly signal
for $\theta^{'}$ relative to $\theta^{''}.$ In the Online Appendix,
we present a general sufficient condition for $\theta^{'}\succ_{s^{'}}\theta^{''}$
under general payoff functions. 
\begin{example}
\noindent \label{exa:beer-quiche} (\citet{cho_signaling_1987}'s
beer-quiche game) The sender (P1) is either strong ($\theta_{\text{strong}}$)
or weak $(\theta_{\text{weak}}$), with prior probability $\lambda(\theta_{\text{strong}})=0.9.$
The sender chooses to either drink \textbf{Beer} or eat \textbf{Quiche}
for breakfast. The receiver (P2), observing this breakfast choice
but not the sender's type, chooses whether to \textbf{Fight} the sender.
If the sender is $\theta_{\text{weak}}$, the receiver prefers to
\textbf{Fight}. If the sender is $\theta_{\text{strong}}$, the receiver
prefers to \textbf{NotFight}. Also, $\theta_{\text{strong}}$ prefers
\textbf{Beer} for breakfast while $\theta_{\text{weak}}$ prefers
\textbf{Quiche} for breakfast. Both types prefer not being fought
over having their favorite breakfast. 

\noindent \begin{center}

\includegraphics[scale=0.3]{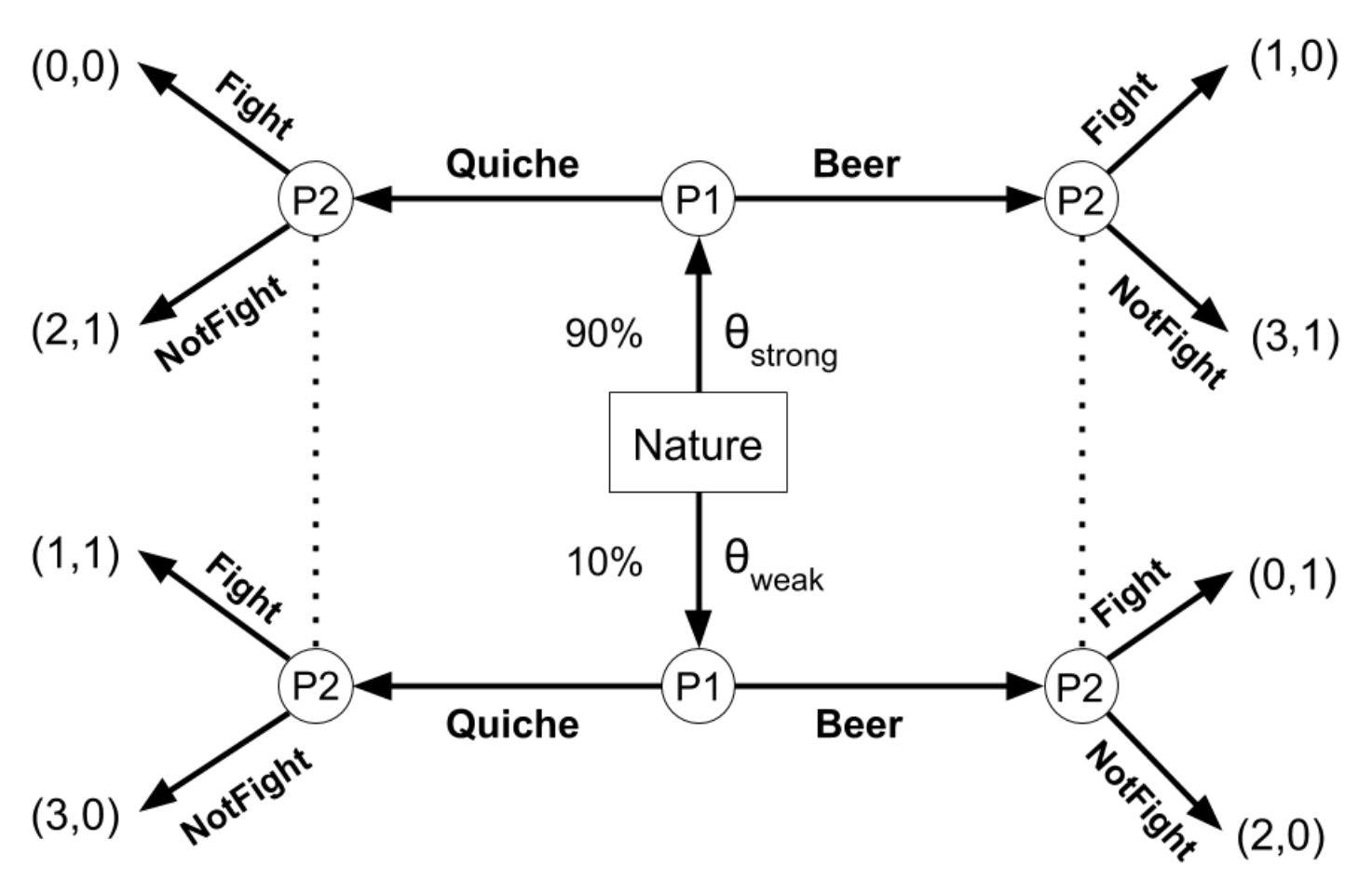}

\noindent \end{center}

This game has separable sender utility with $v(\theta_{\text{strong}},\text{\textbf{Beer}})=v(\theta_{\text{weak}},\text{\textbf{Quiche}})=1$,
$z(\textbf{Fight})=0$ and $z(\textbf{NotFight})=2$. So, we have
$\theta_{\text{strong}}\succ_{\text{\textbf{Beer}}}\theta_{\text{weak}}$.
\hfill{} $\blacklozenge$
\end{example}
It is easy to see that in every Nash equilibrium $\pi^{*},$ if $\theta^{'}\succ_{s^{'}}\theta^{''}$,
then $\pi_{1}^{*}(s^{'}|\theta^{''})>0$ implies $\pi_{1}^{*}(s^{'}|\theta^{'})=1$.
By Bayes rule, this implies that the receiver's equilibrium belief
$p$ after every \emph{on-path} signal $s^{'}$ satisfies the restriction
$\dfrac{p(\theta^{''}|s^{'})}{p(\theta^{'}|s^{'})}\le\dfrac{\lambda(\theta^{''})}{\lambda(\theta^{'})}$
if $\theta^{'}\succ_{s^{'}}\theta^{''}$. Thus in every Nash equilibrium
of the beer-quiche game, if the sender chooses \textbf{Beer} with
positive ex ante probability, then the receiver's odds ratio that
the sender is tough after seeing this signal cannot be less than the
prior odds ratio. Our main result, Theorem \ref{thm:PS_is_compatible},
essentially shows for any strategy profile that can be approximated
by steady-state outcomes with patient and long-lived agents, that
the same compatibility-based restriction is satisfied even for \emph{off-path}
signals. In particular, this allows us to place restrictions on the
receiver's belief after seeing \textbf{Beer} in equilibria where no
type of sender ever plays this signal. 

\subsection{Type compatibility and the Gittins index \label{subsec:gittins}}

We now connect the type compatibility order for a given signal with
the associated Gittins indices. 
\begin{thm}
\noindent \label{thm:index} $\theta^{'}\succ_{s^{'}}\theta^{''}$
if and only if for every $\beta\in[0,1)$ and every profile of beliefs
$\nu$, $I(\theta^{''},s^{'},\nu,\beta)\ge\max_{s^{''}\ne s^{'}}I(\theta^{''},s^{''},\nu,\beta)$
implies $I(\theta^{'},s^{'},\nu,\beta)>\max_{s^{''}\ne s^{'}}I(\theta^{'},s^{''},\nu,\beta)$. 
\end{thm}
That is, $\theta^{'}\succ_{s^{'}}\theta^{''}$ if and only if whenever
$s^{'}$ has the (weakly) highest Gittins index for $\theta^{''}$,
it has the highest index for $\theta,^{'}$ provided the two types
hold the same beliefs and have the same discount factor. The proof
involves reformulating the Gittins index as in Lemma \ref{lem:static},
then applying the compatibility definition.
\begin{proof}
\textbf{Step 1: Only If. }

Suppose $\theta^{'}\succ_{s^{'}}\theta^{''}$ and fix some $\beta\in[0,1)$
and prior belief $\nu$. Suppose $I(\theta^{''},s^{'},\nu,\beta)\ge\max_{s^{''}\ne s^{'}}I(\theta^{''},s^{''},\nu,\beta)$.
We show that $I(\theta^{'},s^{'},\nu,\beta)>\max_{s^{''}\ne s^{'}}I(\theta^{'},s^{''},\nu,\beta)$. 

On any arm $s^{''}\ne s^{'}$, type $\theta^{''}$ could use the (suboptimal)
stopping time $\tau_{s^{''}}^{\theta^{'}}$, which by Lemma \ref{lem:static}
yields an expected per-period payoff of $u_{1}(\theta^{''},s^{''},\pi_{s^{''}}(\nu_{s^{''}},\tau_{s^{''}}^{\theta^{'}},\beta))$.
This is a lower bound for the Gittins index of arm $s^{''}$for type
$\theta^{''}$, so combined with the hypothesis that $I(\theta^{''},s^{'},\nu,\beta)\ge\max_{s^{''}\ne s^{'}}I(\theta^{''},s^{''},\nu,\beta)$,
we get 
\begin{equation}
I(\theta^{''},s^{'},\nu,\beta)\ge\max_{s^{''}\ne s^{'}}u_{1}(\theta^{''},s^{''},\pi_{s^{''}}(\nu_{s^{''}},\tau_{s^{''}}^{\theta^{'}},\beta)).\label{eq:pf_index}
\end{equation}
Now define the receiver strategy $\pi_{2}\in\Pi_{2}$ by $\pi_{2}(\cdot|s^{'})\coloneqq\pi_{s^{'}}(\nu_{s^{'}},\tau_{s^{'}}^{\theta^{''}},\beta)$,
$\pi_{2}(\cdot|s^{''})\coloneqq\pi_{s^{''}}(\nu_{s^{''}},\tau_{s^{''}}^{\theta^{'}},\beta)$
for all $s^{''}\ne s^{'}$. Then Equation (\ref{eq:pf_index}) can
be rewritten as 
\[
u_{1}(\theta^{''},s^{'},\pi_{2}(\cdot|s^{'}))\ge\max_{s^{''}\ne s^{'}}u_{1}(\theta^{''},s^{''},\pi_{2}(\cdot|s^{''})),
\]
that is, $s^{'}$ is weakly optimal for $\theta^{''}$ against $\pi_{2}$.
By the definition of $\theta^{'}\succ_{s^{'}}\theta^{''},$ this implies
$s^{'}$is strictly optimal for $\theta^{'}$ against $\pi_{2}$. 

From the definition of $\pi_{2}$ and Lemma \ref{lem:static}, the
expected utility of $\theta^{'}$ playing any $s^{''}\ne s^{'}$ against
$\pi_{2}$ is equal to the Gittins index of that arm for $\theta^{'},$
namely $I(\theta^{'},s^{''},\nu,\beta)$. On the other hand, $u_{1}(\theta^{'},s^{'},\pi_{2}(\cdot|s^{'}))$
is only a lower bound for $I(\theta^{'},s^{'},\nu,\beta)$. This shows
$I(\theta^{'},s^{'},\nu,\beta)>\max_{s^{''}\ne s^{'}}I(\theta^{'},s^{''},\nu,\beta)$
as desired. 

\textbf{Step 2: If. }

Suppose $\theta^{'}\text{\ensuremath{\not\succ}}\text{ }_{s^{'}}\theta^{''}$.
Then there is some receiver strategy $\pi_{2}^{*}\in\Pi_{2}$ such
that 
\[
u_{1}(\theta^{''},s^{'},\pi_{2}^{*}(\cdot|s^{'}))\ge\max_{s^{''}\ne s^{'}}u_{1}(\theta^{''},s^{''},\pi_{2}^{*}(\cdot|s^{''})),\ \text{and}
\]
\[
u_{1}(\theta^{'},s^{'},\pi_{2}^{*}(\cdot|s^{'}))\leq\max_{s^{''}\ne s^{'}}u_{1}(\theta^{'},s^{''},\pi_{2}^{*}(\cdot|s^{''})).
\]

Let $\nu^{*}$ be any belief that induces $\pi_{2}^{*}$ on average,
that is to say for each $s$, 
\[
\pi_{2}^{*}(\cdot|s)=\int_{\pi_{2,s}\in\Delta(A)}\pi_{2,s}d\nu_{s}^{*}(\pi_{2,s})
\]

Let $\beta=0$. Then $I(\theta,s,\nu^{*},0)=u_{1}(\theta,s,\pi_{2}^{*}(\cdot|s))$
for every $\theta,s$, since the Gittins index is equal to the myopic
payoff when the decision-maker is perfectly impatient. This shows
$I(\theta^{''},s^{'},\nu^{*},0)\ge\max_{s^{''}\ne s^{'}}I(\theta^{''},s^{''},\nu^{*},0)$
and $I(\theta^{'},s^{'},\nu^{*},0)\le\max_{s^{''}\ne s^{'}}I(\theta^{'},s^{''},\nu^{*},0)$. 
\end{proof}

\section{The Aggregate Sender and Receiver Responses \label{sec:aggregate}}

In this section, we will define and analyze the aggregate sender response
$\mathscr{R}_{1}:\Pi_{2}\to\Pi_{1}$ and the aggregate receiver response
$\mathscr{R}_{2}:\Pi_{1}\to\Pi_{2}$. Loosely speaking, these are
the large-populations learning analogs of the best-response functions
in the static signaling game. If we fix the aggregate play of $-i$
population at $\pi_{-i}$ and run the learning model period after
period from an arbitrary initial state, the distribution of play in
$i$ population will approach $\mathscr{R}_{i}[\pi_{-i}]$. Later
in Section \ref{sec:two_sided}, the fixed points of the pair $(\mathscr{R}_{1},\mathscr{R}_{2})$
will characterize the steady states of the learning system.

\subsection{The Aggregate Sender Response \label{subsec:sender_APR}}

To formally define the aggregate sender response, we first introduce
the one-period-forward map. 
\begin{defn}
The \emph{one-period-forward map for type $\theta$}, $f_{\theta}:\Delta(Y_{\theta})\times\Pi_{2}\to\Delta(Y_{\theta})$
is 
\[
f_{\theta}[\psi_{\theta},\pi_{2}](y_{\theta},(s,a)):=\psi_{\theta}(y_{\theta})\cdot\gamma\cdot\boldsymbol{1}\{\sigma_{\theta}(y_{\theta})=s\}\cdot\pi_{2}(a|s)
\]

and $f_{\theta}[\psi_{\theta},\pi_{2}](\emptyset):=1-\gamma$. 
\end{defn}
If the distribution over histories in the type $\theta$ population
is $\psi_{\theta}$ and the receiver population's aggregate play is
$\pi_{2}$, the resulting distribution over histories in the type
$\theta$ population is $f_{\theta}[\psi_{\theta},\pi_{2}]$. Specifically,
there will be a $1-\gamma$ mass of new type $\theta$ who will have
no history. Also, if the optimal first signal of a new type $\theta$
is $s^{'}$, that is if $\sigma_{\theta}(\emptyset)=s^{'},$ then
$f_{\theta}[\psi_{\theta},\pi_{2}](s^{'},a^{'})=\gamma\cdot(1-\gamma)\cdot\pi_{2}(a^{'}|s^{'})$
new senders send $s^{'}$ in their first match, observe action $a^{'}$
in response, and survive. In general, a type $\theta$ who has history
$y_{\theta}$ and whose policy $\sigma_{\theta}(y_{\theta})$ prescribes
playing $s$ has $\pi_{2}(a|s)$ chance of having subsequent history
$(y_{\theta},(s,a))$ provided she survives until next period; the
survival probability corresponds to the factor $\gamma$. 

Write $f_{\theta}^{T}$ for the $T$-fold application of $f_{\theta}$
on $\Delta(Y_{\theta}),$ holding fixed some $\pi_{2}$. Note that
for arbitrary states $\psi$ and $\psi^{'}$, if $(y_{\theta},(s,a))$
is a length-1 history (i.e. $y_{\theta}=\emptyset$), then $\psi_{\theta}(y_{\theta})=\psi'_{\theta}(y_{\theta})$
because both states must assign mass $1-\gamma$ to $\emptyset$,
so $f_{\theta}^{1}[\psi_{\theta},\pi_{2}]$ and $f_{\theta}^{1}[\psi_{\theta}^{'},\pi_{2}]$
agree on $Y_{\theta}[1]$. Iterating, for $T=2,$ $f_{\theta}^{2}[\psi_{\theta},\pi_{2}]$
and $f_{\theta}^{2}[\psi_{\theta}^{'},\pi_{2}]$ agree on $Y_{\theta}[2]$,
because each history in $Y_{\theta}[2]$ can be written as $(y_{\theta},(s,a))$
for $y_{\theta}\in Y_{\theta}[1]$, and $f_{\theta}^{1}[\psi_{\theta},\pi_{2}]$
and $f_{\theta}^{1}[\psi_{\theta}^{'},\pi_{2}]$ match on all $y_{\theta}\in Y_{\theta}[1]$.
Proceeding inductively, we can conclude that $f_{\theta}^{T}(\psi_{\theta},\pi_{2})$
and $f_{\theta}^{T}(\psi_{\theta}^{'},\pi_{2})$ agree on all $Y_{\theta}[t]$
for $t\le T$ for any pair of type $\theta$ states $\psi_{\theta}$
and $\psi_{\theta}^{'}$. This means $\lim_{T\to\infty}f_{\theta}^{T}(\psi_{\theta},\pi_{2})$
exists and is independent of the initial state $\psi_{\theta}$. Denote
this limit as $\psi_{\theta}^{\pi_{2}}.$ It is the long-run distribution
over type $\theta$ histories induced by starting at an arbitrary
state and fixing the receiver population's play at $\pi_{2}$, as
stated formally in the next definition. 
\begin{defn}
The\emph{ aggregate sender response} $\mathscr{R}_{1}:\Pi_{2}\to\Pi_{1}$
is defined by 
\[
\mathscr{R}_{1}[\pi_{2}](s|\theta):=\psi_{\theta}^{\pi_{2}}(y_{\theta}:\sigma_{\theta}(y_{\theta})=s)
\]

where $\psi_{\theta}^{\pi_{2}}:=\lim_{T\to\infty}f_{\theta}^{T}(\psi_{\theta},\pi_{2})$
with $\psi_{\theta}$ any arbitrary $\theta$ state. 
\end{defn}
That is, $\mathscr{R}_{1}[\pi_{2}](\cdot|\theta$) is the long-run
aggregate behavior in the type $\theta$ population when the receivers'
aggregate play is fixed at $\pi_{2}$.
\begin{rem}
\label{rem:ASR_details} Technically, $\mathscr{R}_{1}$ depends on
$g_{1},\delta,$ and $\gamma$, just like $\sigma_{\theta}$ does.
When relevant, we will make these dependencies clear by adding the
appropriate parameters as superscripts to $\mathscr{R}_{1}$, but
we will mostly suppress them to lighten notation. 
\begin{rem}
\label{rem:sender_APR} Although the aggregate sender response is
defined at the aggregate level, $\mathscr{R}_{1}[\pi_{2}](\cdot|\theta)$
also describes the probability distribution of the play of a single
type $\theta$ sender over her lifetime when she faces receiver play
drawn from $\pi_{2}$ every period.\footnote{Observe that $f_{\theta}[\psi_{\theta},\pi_{2}]$ restricted to $Y_{\theta}[1]$
gives the probability distribution over histories for a type $\theta$
who uses $\sigma_{\theta}$ and faces play drawn from $\pi_{2}$ for
one period: it puts weight $\pi_{2}(a^{'}|s^{'})$ on history $(s^{'},a^{'})$
where $s^{'}=\sigma_{\theta}(\emptyset)$. Similarly, $f_{\theta}^{T}[\psi_{\theta},\pi_{2}]$
restricted to $Y_{\theta}[t]$ for any $t\le T$ gives the probability
distribution over histories for someone who uses $\sigma_{\theta}$
and faces play drawn from $\pi_{2}$ for $t$ periods. Since $\psi_{\theta}^{\pi_{2}}$
assigns probability $(1-\gamma)\gamma^{t}$ to the set of histories
$Y_{\theta}[t]$, $\mathscr{R}_{1}[\pi_{2}](\cdot|\theta)=\sigma_{\theta}(\psi_{\theta}^{\pi_{2}})$
is a weighted average over the distributions of period $t$ play ($t=1,2,3,...$)
of someone using $\sigma_{\theta}$ and facing $\pi_{2}$, with weight
$(1-\gamma)\gamma^{t}$ given to the period $t$ distribution. }
\end{rem}
\end{rem}

\subsection{Type Compatibility and the Aggregate Sender Response \label{subsec:sender_reply} }

The next lemma shows how type compatibility translates into restrictions
on the aggregate sender response for different types.
\begin{lem}
\label{lem:compatible_comonotonic}Suppose $\theta^{'}\succ_{s^{'}}\theta^{''}$.
Then for any regular prior $g_{1}$, $0\le\delta,\gamma<1$, and any
$\pi_{2}\in\Pi_{2}$, we have $\mathscr{R}_{1}[\pi_{2}](s^{'}|\theta^{'})\ge\mathscr{R}_{1}[\pi_{2}](s^{'}|\theta^{''})$. 
\end{lem}
Theorem \ref{thm:index} showed that when $\theta^{'}\succ_{s^{'}}\theta^{''}$
and the two types share the same beliefs, if $\theta^{''}$plays $s^{'}$
then $\theta^{'}$ must also play $s^{'}$. But even though new agents
of both types start with the same prior $g_{1}$, their beliefs may
quickly diverge during the learning process due to $\sigma_{\theta^{'}}$
and $\sigma_{\theta^{''}}$ prescribing different experiments after
the same history. This lemma shows that compatibility still imposes
restrictions on the aggregate play of the sender population: Regardless
of the aggregate play $\pi_{2}$ in the receiver population, the frequencies
that $s^{'}$ appears in the aggregate responses of different types
are always co-monotonic with the compatibility order $\succ_{s^{'}}$. 

To gain intuition for Lemma \ref{lem:compatible_comonotonic}, consider
two new senders with types $\theta_{\text{strong}}$ and $\theta_{\text{weak}}$
who are learning to play the beer-quiche game from Example \ref{exa:beer-quiche}.
Suppose they have uniform priors over the responses to each signal,
and that they face a sequence of receivers programmed to play \textbf{Fight}
after \textbf{Beer} and \textbf{NotFight} after \textbf{Quiche}. Since
observing \textbf{Fight} is the worst possible news about a signal's
payoff, the Gittins index of a signal decreases when \textbf{Fight}
is observed. Conversely, the Gittins index of a signal increases after
each observation of \textbf{NotFight}.\footnote{This follows from \citet{bellman1956problem}'s Theorem 2 on Bernoulli
bandits. } Thus given the assumed play of the receivers, there are $n_{1},n_{2}\ge0$
such that type $\theta_{\text{strong}}$ play \textbf{Beer} for $n_{1}$
periods (and observe $n_{1}$ instances of \textbf{Fight}) and then
switch to \textbf{Quiche} forever after, while type $\theta_{\text{weak}}$
will play \textbf{Beer} for $n_{2}$ periods before switching to \textbf{Quiche}
forever after. Now we claim that $n_{1}\ge n_{2}$. To see why, suppose
instead that $n_{1}<n_{2}$, and let $\nu$ be the posterior belief
about receivers' aggregate play induced from $n_{1}$ periods of observing
\textbf{Fight} after \textbf{Beer}. After $n_{1}$ periods, both types
would share the belief $\nu$. Then at belief $\nu$ type $\theta_{\text{weak}}$
must play \textbf{Beer} while type $\theta_{\text{strong}}$ plays
\textbf{Quiche}, so signal \textbf{Beer} must have the highest Gittins
index for $\theta_{\text{weak}}$ but not for $\theta_{\text{strong}}$.
But this would contradict Theorem \ref{thm:index}. 

The proof of Lemma \ref{lem:compatible_comonotonic} relies on the
similar idea of fixing a particular ``programming'' of receiver
play and studying the induced paths of experimentation for different
types. In the aggregate learning model, the sequence of responses
that a given sender encounters in her life depends on the realization
of the random matching process, because different receivers have different
histories and respond differently to a given signal. We can index
all possible sequences of random matching realizations using a device
we call the ``pre-programmed response path''. To show that more
compatible types play a given signal more often, it suffices to show
this comparison holds on each pre-programmed response path, thus coupling
the learning processes of types $\theta^{'}$ and $\theta^{''}.$
We will show that the intuition above extends to signaling games with
any number of signals and to any pre-programmed response path. 
\begin{defn}
A \emph{pre-programmed response path} $\mathfrak{a}=(a_{1,s},a_{2,s},...,)_{s\in S}$
is an element in $\times_{s\in S}\left(A^{\infty}\right)$. 
\end{defn}
A pre-programmed response path is an $|S|$-tuple of infinite sequences
of receiver actions, one sequence for each signal. For a given pre-programmed
response path $\mathfrak{a}$, we can imagine starting with a new
type $\theta$ and generating receiver play each period in the following
programmatic manner: when the sender plays $s$ for the $j$-th time,
respond with receiver action $a_{j,s}$. (If the sender sends $s^{''}$
five times and then sends $s^{'}\ne s^{''}$, the response she gets
to $s^{'}$ is $a_{1,s^{'}}$, not $a_{6,s^{'}}$.) For a type $\theta$
who applies $\sigma_{\theta}$ each period, $\mathfrak{a}$ induces
a deterministic history of experiments and responses, which we denote
$y_{\theta}(\mathfrak{a}$). The induced history $y_{\theta}(\mathfrak{a})$
can be used to calculate $\overline{\mathscr{R}}_{1}[\mathfrak{a}](\cdot|\theta)$,
the distribution of signals over the lifetime of a type $\theta$
induced by the pre-programmed response path $\mathfrak{a}$. Namely,
$\overline{\mathscr{R}}_{1}[\mathfrak{a}](\cdot|\theta)$ is simply
a mixture over all signals sent along the history $y_{\theta}(\mathfrak{a})$,
with weight $(1-\gamma)\gamma^{t-1}$ given to the signal in period
$t$. 

Now consider a type $\theta$ facing actions generated i.i.d. from
the receiver behavior strategy $\pi_{2}$ each period, as in the interpretation
of $\mathscr{R}_{1}$ in Remark \ref{rem:sender_APR}. This data-generating
process is equivalent to drawing a random pre-programmed response
path $\mathfrak{a}$ at time 0 according to a suitable distribution,
then producing all receiver actions using $\mathfrak{a}$. That is,
$\mathscr{R}_{1}[\pi_{2}](\cdot|\theta)=\int\overline{\mathscr{R}}_{1}[\mathfrak{a}](\cdot|\theta)d\pi_{2}(\mathfrak{a})$
where we abuse notation and use $d\pi_{2}(\mathfrak{a})$ to denote
the distribution over pre-programmed response paths associated with
$\pi_{2}$. Importantly, any two types $\theta^{'}$ and $\theta^{''}$
face the same distribution over pre-programmed response paths, so
to prove the proposition it suffices to show $\overline{\mathscr{R}}_{1}[\mathfrak{a}](s^{'}|\theta^{'})\ge\overline{\mathscr{R}}_{1}[\mathfrak{a}](s^{'}|\theta^{''})$
for all $\mathfrak{a}$. 
\begin{proof}
For $t\ge0$, write $y_{\theta}^{t}$ for the truncation of infinite
history $y_{\theta}$ to the first $t$ periods, with $y_{\theta}^{\infty}:=y_{\theta}$.
Given a finite or infinite history $y_{\theta}^{t}$ for type $\theta$,
the signal counting function $\#(s|y_{\theta}^{t})$ returns how many
times signal $s$ has appeared in $y_{\theta}^{t}$. (We need this
counting function since the receiver play generated by a pre-programmed
response path each period depends on how many times each signal has
been sent so far.) 

As discussed above, we need only show $\overline{\mathscr{R}}_{1}[\mathfrak{a}](s^{'}|\theta^{'})\ge\overline{\mathscr{R}}_{1}[\mathfrak{a}](s^{'}|\theta^{''})$.
Let $\mathfrak{a}$ be given and write $T_{j}^{\theta}$ for the period
in which type $\theta$ sends signal $s^{'}$ for the $j$-th time
in the induced history $y_{\theta}(\mathfrak{a})$. If no such period
exists, then set $T_{j}^{\theta}=\infty$. Since $\overline{\mathscr{R}}_{1}[\mathfrak{a}](\cdot|\theta)$
is a weighted average over signals in $y_{\theta}(\mathfrak{a})$
with decreasing weights given to later signals, to prove $\overline{\mathscr{R}}_{1}[\mathfrak{a}](s^{'}|\theta^{'})\ge\overline{\mathscr{R}}_{1}[\mathfrak{a}](s^{'}|\theta^{''})$
it suffices to show that $T_{j}^{\theta^{'}}\le T_{j}^{\theta^{''}}$
for every $j$. Towards this goal, we will prove a sequence of statements
by induction: 

\texttt{Statement $j$}: Provided $T_{j}^{\theta^{''}}$ is finite,
$\#\left(s^{''}\ |\ y_{\theta^{'}}^{T_{j}^{\theta^{'}}}(\mathfrak{a})\right)\le\#\left(s^{''}\ |\ y_{\theta^{''}}^{T_{j}^{\theta^{''}}}(\mathfrak{a})\right)$for
all $s^{''}\ne s^{'}$. 

For every $j$ where $T_{j}^{\theta^{''}}<\infty$, \texttt{statement
$j$} implies that the number of periods type $\theta^{'}$ spent
sending each signal $s^{''}\ne s^{'}$ before sending $s^{'}$ for
the $j$-th time is fewer than the number of periods $\theta^{''}$
spent doing the same. Therefore it follows that $\theta^{'}$ sent
$s^{'}$ for the $j$-th time sooner than $\theta^{''}$ did, that
is $T_{j}^{\theta^{'}}\le T_{j}^{\theta^{''}}$. Finally, if $T_{j}^{\theta^{''}}=\infty$,
then evidently $T_{j}^{\theta'}\le\infty=T_{j}^{\theta^{''}}.$ 

It now remains to prove the sequence of statements by induction. 

\texttt{Statement 1} is the base case. By way of contradiction, suppose
$T_{1}^{\theta^{''}}<\infty$ and 
\[
\#\left(s^{''}\ |\ y_{\theta^{'}}^{T_{1}^{\theta^{'}}}(\mathfrak{a})\right)>\#\left(s^{''}\ |\ y_{\theta^{''}}^{T_{1}^{\theta^{''}}}(\mathfrak{a})\right)
\]
 for some $s^{''}\ne s^{'}$. Then there is some earliest period $t^{*}<T_{1}^{\theta^{'}}$
where 
\[
\#\left(s^{''}\ |\ y_{\theta^{'}}^{t^{*}}(\mathfrak{a})\right)>\#\left(s^{''}\ |\ y_{\theta^{''}}^{T_{1}^{\theta^{''}}}(\mathfrak{a})\right),
\]
where type $\theta^{'}$ played $s^{''}$ in period $t^{*}$, $\sigma_{\theta^{'}}(y_{\theta^{'}}^{t^{*}-1}(\mathfrak{a}))=s^{''}$.

But by construction, by the end of period $t^{*}-1$ type $\theta^{'}$
has sent $s^{''}$ exactly as many times as type $\theta^{''}$ has
sent it by period $T_{1}^{\theta^{''}}-1$, so that 

\[
\#\left(s^{''}\ |\ y_{\theta^{'}}^{t^{*}-1}(\mathfrak{a})\right)=\#\left(s^{''}\ |\ y_{\theta^{''}}^{T_{1}^{\theta^{''}}-1}(\mathfrak{a})\right).
\]
Furthermore, neither type has sent $s^{'}$ yet, so also
\[
\#\left(s^{'}\ |\ y_{\theta^{'}}^{t^{*}-1}(\mathfrak{a})\right)=\#\left(s^{'}\ |\ y_{\theta^{''}}^{T_{1}^{\theta^{''}}-1}(\mathfrak{a})\right).
\]
Therefore, type $\theta^{'}$ holds the same posterior over the receiver's
reaction to signals $s^{'}$ and $s^{''}$ at period $t^{*}-1$ as
type $\theta^{''}$ does at period $T_{1}^{\theta^{''}}-1$. So\footnote{In the following equation and elsewhere in the proof, we abuse notation
and write $I(\theta,s,y)$ to mean $I(\theta,s,g_{1}(\cdot|y),\delta\gamma)$,
which is the Gittins index of type $\theta$ for signal $s$ at the
posterior obtained from updating the prior $g_{1}$ using history
$y$, with effective discount factor $\delta\gamma$. } by Theorem \ref{thm:index},
\begin{equation}
s^{'}\in\underset{\hat{s}\in S}{\arg\max}\ I\left(\theta^{''},\hat{s},y_{\theta^{''}}^{T_{1}^{\theta^{''}}-1}(\mathfrak{a})\right)\implies I(\theta^{'},s^{'},y_{\theta^{'}}^{t^{*}-1}(\mathfrak{a}))>I(\theta^{'},s^{''},y_{\theta^{'}}^{t^{*}-1}(\mathfrak{a})).\label{eq:core_dp1}
\end{equation}
However, by construction of $T_{1}^{\theta^{''}}$, we have $\sigma_{\theta^{''}}\left(y_{\theta^{''}}^{T_{1}^{\theta^{''}}-1}(\mathfrak{a})\right)=s^{'}$.
By the optimality of the Gittins index policy, the left-hand side
of Equation (\ref{eq:core_dp1}) is satisfied. But, again by the optimality
of the Gittins index policy, the right-hand side of Equation (\ref{eq:core_dp1})
contradicts $\sigma_{\theta^{'}}(y_{\theta^{'}}^{t^{*}-1}(\mathfrak{a}))=s^{''}$.
Therefore we have proven \texttt{Statement 1}. 

Now suppose \texttt{Statement $j$} holds for all $j\le K$. We show
\texttt{Statement $K+1$} also holds. If $T_{K+1}^{\theta^{''}}$
is finite, then $T_{K}^{\theta^{''}}$ is also finite. The inductive
hypothesis then shows 

\[
\#\left(s^{''}\ |\ y_{\theta^{'}}^{T_{K}^{\theta^{'}}}(\mathfrak{a})\right)\le\#\left(s^{''}\ |\ y_{\theta^{''}}^{T_{K}^{\theta^{''}}}(\mathfrak{a})\right)
\]
for every $s^{''}\ne s^{'}$. Suppose there is some $s^{''}\ne s^{'}$
such that
\[
\#\left(s^{''}\ |\ y_{\theta^{'}}^{T_{K+1}^{\theta^{'}}}(\mathfrak{a})\right)>\#\left(s^{''}\ |\ y_{\theta^{''}}^{T_{K+1}^{\theta^{''}}}(\mathfrak{a})\right).
\]
Together with the previous inequality, this implies type $\theta^{'}$
played $s^{''}$ for the\\
 $\left[\#\left(s^{''}\ |\ y_{\theta^{''}}^{T_{K+1}^{\theta^{''}}}(\mathfrak{a})\right)+1\right]$-th
time sometime between playing $s^{'}$ for the $K$-th time and playing
$s^{'}$ for the ($K+1$)-th time. That is, if we put
\[
t^{*}\coloneqq\min\left\{ t:\#(s^{''}\ |\ y_{\theta^{'}}^{t}(\mathfrak{a})))>\#\left(s^{''}\ |\ y_{\theta^{''}}^{T_{K+1}^{\theta^{''}}}(\mathfrak{a})\right)\right\} ,
\]
then $T_{K}^{\theta^{'}}<t^{*}<T_{K+1}^{\theta^{'}}$. By the construction
of $t^{*}$, 
\[
\#\left(s^{''}\ |\ y_{\theta^{'}}^{t^{*}-1}(\mathfrak{a})\right)=\#\left(s^{''}\ |\ y_{\theta^{''}}^{T_{K+1}^{\theta^{''}}-1}(\mathfrak{a})\right),\text{ and also}
\]
\[
\#\left(s^{'}\ |\ y_{\theta^{'}}^{t^{*}-1}(\mathfrak{a})\right)=K=\#\left(s^{'}\ |\ y_{\theta^{''}}^{T_{K+1}^{\theta^{''}}-1}(\mathfrak{a})\right).
\]

Therefore, type $\theta^{'}$ holds the same posterior over the receiver's
reaction to signals $s^{'}$ and $s^{''}$ at period $t^{*}-1$ as
type $\theta^{''}$ does at period $T_{K+1}^{\theta^{''}}-1$. As
in the base case, we can invoke Theorem \ref{thm:index} to show that
it is impossible for $\theta^{'}$ to play $s^{''}$ in period $t^{*}$
while $\theta^{''}$ plays $s^{'}$ in period $T_{K+1}^{\theta^{''}}$.
This shows \texttt{statement $j$ }is true for every $j$ by induction. 
\end{proof}

\subsection{The Aggregate Receiver Response}

We now turn to the receivers' problem. Each new receiver thinks he
is facing a fixed but unknown aggregate sender behavior strategy $\pi_{1}$,
with belief over $\pi_{1}$ given by his regular prior $g_{2}$. To
maximize his expected utility, the receiver must learn to infer the
type of the sender from the signal, using his personal experience. 

Unlike the senders whose optimal policies may involve experimentation,
the receivers' problem only involves passive learning. Since the receiver
observes the same information in a match regardless of his action,
the optimal policy $\sigma_{2}(y_{2})$ simply best responds to the
posterior belief induced by history $y_{2}$. 
\begin{defn}
The \emph{one-period-forward map for receivers} $f_{2}:\Delta(Y_{2})\times\Pi_{1}\to\Delta(Y_{2})$
is 
\[
f_{2}[\psi_{2},\pi_{1}](y_{2},(\theta,s)):=\psi_{2}(y_{2})\cdot\gamma\cdot\lambda(\theta)\cdot\pi_{1}(s|\theta)
\]

and $f_{2}(\emptyset):=1-\gamma$. 
\end{defn}
As with the one-period-forward maps $f_{\theta}$ for senders, $f_{2}[\psi_{2},\pi_{1}]$
describes the new distribution over receiver histories tomorrow if
the distribution over histories in the receiver population today is
$\psi_{2}$ and the sender population's aggregate play is $\pi_{1}.$
We write $\psi_{2}^{\pi_{1}}:=\lim_{T\to\infty}f_{2}^{T}(\psi_{2},\pi_{1})$
for the long-run distribution over $Y_{2}$ induced by fixing sender
population's play at $\pi_{1}$, which is independent of the particular
choice of initial state $\psi_{2}$. 
\begin{defn}
The\emph{ aggregate receiver response} $\mathscr{R}_{2}:\Pi_{1}\to\Pi_{2}$
is 
\[
\mathscr{R}_{2}[\pi_{1}](a|s):=\psi_{2}^{\pi_{1}}(y_{2}:\sigma_{2}(y_{2})(s)=a),
\]

where $\psi_{2}^{\pi_{1}}:=\lim_{T\to\infty}f_{2}^{T}(\psi_{2},\pi_{1})$
with $\psi_{2}$ any arbitrary receiver state. 
\end{defn}
We are interested in the extent to which $\mathscr{R}_{2}[\pi_{1}]$
responds to inequalities of the form $\pi_{1}(s^{'}|\theta^{'})\ge\pi_{1}(s^{'}|\theta^{''})$
embedded in $\pi_{1}$, such as those generated when $\theta^{'}\succ_{s^{'}}\theta^{''}$
(Lemma \ref{lem:compatible_comonotonic}). To this end, for any two
types $\theta^{'},\theta^{''}$ we define $P_{\theta^{'}\triangleright\theta^{''}}$
as those beliefs where the odds ratio of $\theta^{'}$ to $\theta^{''}$
exceeds their prior odds ratio, that is
\begin{equation}
P_{\theta^{'}\triangleright\theta^{''}}:=\left\{ p\in\Delta(\Theta):\frac{p(\theta^{''})}{p(\theta^{'})}\le\frac{\lambda(\theta^{''})}{\lambda(\theta^{'})}\right\} .\label{eq:odds_ratio_P}
\end{equation}
If $\pi_{1}(s^{'}|\theta^{'})\ge\pi_{1}(s^{'}|\theta^{''}),$ $\pi_{1}(s^{'}|\theta^{'})>0,$
and receiver knows $\pi_{1}$, then receiver's posterior belief about
sender's type after observing $s^{'}$ falls in the set $P_{\theta^{'}\triangleright\theta^{''}}$.
The next lemma shows that under the additional provisions that $\pi_{1}(s^{'}|\theta^{'})$
is ``large enough'' and receivers are sufficiently long-lived, $\mathscr{R}_{2}[\pi_{1}]$
will best respond to $P_{\theta^{'}\triangleright\theta^{''}}$ with
high probability when $s^{'}$ is sent. 

For $P\subseteq\Delta(\Theta)$, we let\footnote{We abuse notation here and write $u_{2}(p,s,a^{'})$ to mean $\sum_{\theta\in\Theta}u_{2}(\theta,s,a^{'})\cdot p(\theta)$. }
$\text{BR}(P,s)\coloneqq\bigcup_{p\in P}\left(\underset{a'\in A}{\arg\max}\ \text{\ensuremath{u_{2}(p,s,a')}}\right);$
this is the set of best responses to $s$ supported by some belief
in $P$. 
\begin{lem}
\label{lem:receiver_learning}Let regular prior $g_{2}$, types $\theta^{'},\theta^{''}$,
and signal $s^{'}$ be fixed. For every $\epsilon>0$, there exist
$C>0$ and $\underline{\gamma}<1$ so that for any $0\le\delta<1$,
$\underline{\gamma}\le\gamma<1$, and $n\ge1$, if $\pi_{1}(s^{'}|\theta^{'})\ge\pi_{1}(s^{'}|\theta^{''})$
and $\pi_{1}(s^{'}|\theta^{'})\ge(1-\gamma)nC$, then 
\[
\mathscr{R}_{2}[\pi_{1}](\text{BR}(P_{\theta^{'}\triangleright\theta^{''}},s^{'})\ |\ s^{'})\ge1-\frac{1}{n}-\epsilon.
\]
\end{lem}
This lemma gives a lower bound on the probability that $\mathscr{R}_{2}[\pi_{1}]$
best responds to $P_{\theta^{'}\triangleright\theta^{''}}$ after
signal $s^{'}$. Note that the bound only applies for survival probabilities
$\gamma$ that are close enough to 1, because when receivers have
short lifetimes they need not get enough data to outweigh their prior.
Note also that more of the receivers learn the compatibility condition
when $\pi_{1}(s^{'}|\theta^{'})$ is large compared to $(1-\gamma)$
and almost all of them do in the limit of $\text{\ensuremath{n\shortrightarrow\infty.}}$The
proof of Lemma \ref{lem:receiver_learning} relies on Theorem 2 from
\citet*{fudenberg_he_imhof_2016} about updating Bayesian posteriors
after rare events, where the rare event corresponds to observing $\theta^{'}$
play $s^{'}$. The details are in Appendix \ref{subsec:Proof-of-Lemma-Receiver}. 

To interpret the condition $\pi_{1}(s^{'}|\theta^{'})\ge(1-\gamma)nC,$
recall that an agent with survival chance $\gamma$ has a typical
lifespan of $\frac{1}{1-\gamma}$. If $\pi_{1}$ describes the aggregate
play in the sender population, then on average a type $\theta^{'}$
plays $s^{'}$ for $\frac{1}{1-\gamma}\cdot\pi_{1}(s^{'}|\theta^{'})$
periods in her life. So when a typical type $\theta^{'}$ plays $s^{'}$
for $nC$ periods, this lemma provides a bound of $1-\frac{1}{n}-\epsilon$
on the share of the receiver responses that lie in $\text{BR}(P_{\theta^{'}\triangleright\theta^{''}},s^{'}).$
Note that the hypothesis $\theta^{'}$ plays $s^{'}$ for $nC$ periods
does not require that $\pi_{1}(s^{'}|\theta^{'})$ is bounded away
from 0 as $\gamma\to1.$ To preview, Lemma \ref{lem:nondom_message}
in the next section will establish that signals that are not weakly
equilibrium dominated for a given type are played sufficiently often
that Lemma \ref{lem:receiver_learning} has bite when both $\delta$
and $\gamma$ are close to 1.

\section{Steady State Implications for Aggregate Play\label{sec:two_sided} }

Section \ref{sec:aggregate} separately examined the senders' and
receivers' learning problems. In this section, we turn to the two-sided
learning problem. We will first define steady-state strategy profiles,
which are signaling game strategy profiles $\pi^{*}$ where $\pi_{1}^{*}$
and $\pi_{2}^{*}$ are mutual aggregate responses, and then characterize
the steady states using our previous results. 

\subsection{Steady States, $\delta$-Stability, and Patient Stability }

We introduced the one-period-forward maps $f_{\theta}$ and $f_{2}$
in Section \ref{sec:aggregate}, which describe the deterministic
transition between state $\psi^{t}$ this period to state $\psi^{t+1}$
next period through the learning dynamics and the birth-death process.
More precisely, $\psi_{\theta}^{t+1}=f_{\theta}(\psi_{\theta}^{t},\sigma_{2}(\psi_{2}^{t}))$
and $\psi_{2}^{t+1}=f_{2}(\psi_{2}^{t},(\sigma_{\theta}(\psi_{\theta}^{t}))_{\theta\in\Theta}).$
A steady state is a fixed point $\psi^{*}$ of this transition map,
. 
\begin{defn}
A state $\psi^{*}$ is a \emph{steady state} if $\psi_{\theta}^{*}=f_{\theta}(\psi_{\theta}^{*},\sigma_{2}(\psi_{2}^{*}))$
for every $\theta$ and $\psi_{2}^{*}=f_{2}(\psi_{2}^{*},(\sigma_{\theta}(\psi_{\theta}^{*}))_{\theta\in\Theta})$.
The set of all steady states for regular prior $g$ and $0\le\delta,\gamma<1$
is denoted $\Psi^{*}(g,\delta,\gamma)$, while the set of steady-state
strategy profiles is $\Pi^{*}(g,\delta,\gamma):=\{\sigma(\psi^{*}):\psi^{*}\in\Psi^{*}(g,\delta,\gamma)\}$. 
\end{defn}
The strategy profiles associated with steady states represent time-invariant
distributions of play, as the information lost when agents die each
period exactly balances out the information agents gain through learning
that period. This means the exchangeability assumption of the learners
will be satisfied in any steady state. 

We now give an equivalent characterization $\Pi^{*}(g,\delta,\gamma)$
in terms of $\mathscr{R}_{1}$ and $\mathscr{R}_{2}$. The proof is
in Appendix \ref{subsec:pf_mutual_AR}. 
\begin{prop}
\label{prop:mutual_AR} $\pi^{*}\in\Pi^{*}(g,\delta,\gamma)$ if and
only if $\mathscr{R}_{1}^{g,\delta,\gamma}(\pi_{2}^{*})=\pi_{1}^{*}$
and $\mathscr{R}_{2}^{g,\delta,\gamma}(\pi_{1}^{*})=\pi_{2}^{*}$. 
\end{prop}
(Note that here we make the dependence of $\mathscr{R}_{1}$ and $\mathscr{R}_{2}$
on parameters $(g,\delta,\gamma)$ explicit to avoid confusion.) That
is, a steady-state strategy profile is a pair of mutual aggregate
replies. 

The next proposition guarantees that there always exists at least
one steady-state strategy profile. 
\begin{prop}
\label{thm:existence}$\Pi^{*}(g,\delta,\gamma)$ is nonempty and
compact in the norm topology. 
\end{prop}
The proof is in the Online Appendix. We establish that $\Psi^{*}(g,\delta,\gamma)$
is nonempty and compact in the $\ell_{1}$ norm on the space of distributions,
which immediately implies the same properties for $\Pi^{*}(g,\delta,\gamma)$.
Intuitively, if lifetimes are finite, the set of histories is finite,
so the set of states is of finite dimension. Here the one-period-forward
map $f=((f_{\theta})_{\theta\in\Theta},f_{2})$ is continuous, so
the usual version of Brouwer's fixed-point theorem applies. With geometric
lifetimes, very old agents are rare, so truncating the agents' lifetimes
at some large $T$ yields a good approximation. Instead of using these
approximations directly, our proof shows that under the $\ell_{1}$
norm $f$ is continuous, and that (because of the geometric lifetimes)
the feasible states form a compact locally convex Hausdorff space.
This lets us appeal to a fixed-point theorem for that domain. 

We now focus on the iterated limit 
\[
\lim_{\delta\to1}\lim_{\gamma\to1}\Pi^{*}(g,\delta,\gamma),
\]
 that is, the set of steady-state strategy profiles for $\delta$
and $\gamma$ near 1, where we first send $\gamma$ to 1 holding $\delta$
fixed, and then send $\delta$ to 1. 
\begin{defn}
For each $0\le\delta<1$, a strategy profile $\pi^{*}$ is \emph{$\delta$-stable
under $g$ }if there is a sequence $\gamma_{k}\to1$ and an associated
sequence of steady-state strategy profiles $\pi^{(k)}\in\Pi^{*}(g,\delta,\gamma_{k})$,
such that $\pi^{(k)}\to\pi^{*}$. Strategy profile $\pi^{*}$ is \emph{patiently
stable under $g$ }if there is a sequence $\delta_{k}\to1$ and an
associated sequence of strategy profiles $\pi^{(k)}$ where each $\pi^{(k)}$
is $\delta_{k}$-stable under $g$ and $\pi^{(k)}\to\pi^{*}$. Strategy
profile $\pi^{*}$ is \emph{patiently stable} if it is patiently stable
under some regular prior $g$. 
\end{defn}
Heuristically, patiently stable strategy profiles are the limits of
learning outcomes when agents become infinitely patient (so that senders
are willing to make many experiments) and long lived (so that agents
on both sides can learn enough for their data to outweigh their prior).
As in past work on steady-state learning \citep{fudenberg_steady_1993,fudenberg_superstition_2006},
the reason for this order of limits is to ensure that most agents
have enough data that they stop experimenting and play myopic best
responses.\footnote{If agents did not eventually stop experimenting as they age, then
even if most agents have approximately correct beliefs, aggregate
play need not be close to a Nash equilibrium because most agents would
not be playing a (static) best response to their beliefs.} We do not know whether our results extend to the other order of limits;
we explain the issues involved below, after sketching the intuition
for Proposition \ref{prop:nash}. 

\subsection{Preliminary Results on $\delta$-Stability and Patient Stability}

When $\gamma$ is near 1, agents correctly learn the consequences
of the strategies they play frequently. But for a fixed patience level
they may choose to rarely or never experiment, and so can maintain
incorrect beliefs about the consequences of strategies that they do
not play. The next result formally states this, which parallels \citet{fudenberg_steady_1993}'s
result that $\delta$-stable strategy profiles are self-confirming
equilibria. 
\begin{prop}
\label{thm:fixed_delta}Suppose strategy profile $\pi^{*}$ is $\delta$-stable
under a regular prior. Then for every type $\theta$ and signal $s$
with $\pi_{1}^{*}(s|\theta)>0$, $s$ is a best response to some $\pi_{2}\in\Pi_{2}$
for type $\theta$, and furthermore $\pi_{2}(\cdot|s)=\pi_{2}^{*}(\cdot|s)$.
Also, for any signal $s$ such that $\pi_{1}^{*}(s|\theta)>0$ for
at least one type $\theta$, $\pi_{2}^{*}(\cdot|s)$ is supported
on pure best responses to the Bayesian belief generated by $\pi_{1}^{*}$
after $s$. 
\end{prop}
We prove this result in the Online Appendix. The idea of the proof
is the following: If signal $s$ has positive probability in the limit,
then it is played many times by the senders, so the receivers eventually
learn the correct posterior distribution for $\theta$ given $s.$
As the receivers have no incentive to experiment, their actions after
$s$ will be a best response to this correct posterior belief. For
the senders, suppose $\pi_{1}^{*}(s|\theta)>0,\text{ }$ but $s$
is not a best response for type $\theta$ to any $\pi_{2}\in\Pi_{2}$
that matches $\pi_{2}^{*}(\cdot|s)$. Yet if a sender has played $s$
many times then with high probability her belief about $\pi_{2}(\cdot|s)$
is close to $\pi_{2}^{*}(\cdot|s)$, so playing $s$ is not myopically
optimal. This would imply that type $\theta$ has persistent option
value for signal $s$, which contradicts the fact that this option
value must converge to 0 with the sample size. 
\begin{rem}
This proposition says that each sender type is playing a best response
to a belief about the receiver's play that is correct on the equilibrium
path, and that the receivers are playing an aggregate best response
to the aggregate play of the senders. Thus the $\delta$-stable outcomes
are a version of self-confirming equilibrium where different types
of sender are allowed to have different beliefs. Moreover, as the
next example shows, this sort of heterogeneity in the senders' beliefs
about the aggregate strategy of the receivers can endogenously arise
in a $\delta$-stable strategy profile even when all types of new
senders start with the same prior over how the receivers play. \footnote{\citet*{dekel_learning_2004} defined type-heterogeneous self-confirming
equilibrium in static Bayesian games. As they noted, this sort of
heterogeneity is natural when the type of each agent is fixed, but
not if each agent's type is drawn i.i.d. in each period. To extend
their definition to signaling games, we can define the ``signal functions''
$y_{i}(a,\theta)$ from that paper to respect the extensive form of
the game. See also \citet{fudenberg_kamada_2018}. }
\end{rem}
\begin{example}
\noindent Consider the following game: 

\noindent \begin{center}

\includegraphics[scale=0.3]{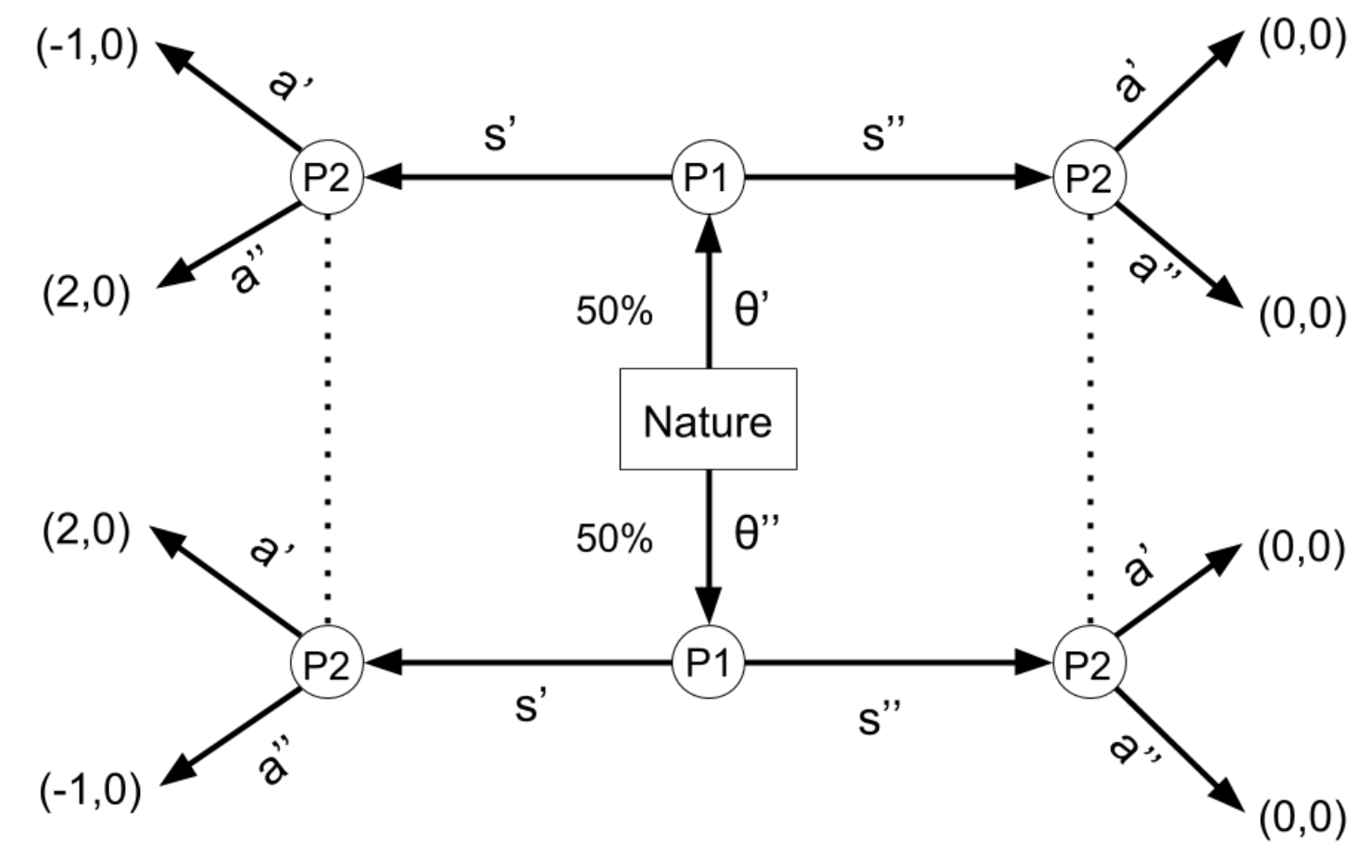}

\noindent \end{center}

The receiver is indifferent between all responses. Fix any regular
prior $g_{2}$ for the receiver and any regular prior $g_{1}^{(s^{''})}$
for the sender. Let $g_{1}^{(s^{'})}$ be Beta(1, 3) on $a^{'}$ and
$a^{''}$ respectively. We claim that it is $\delta$-stable when
$\delta=0$ for both types to send $s^{''}$ and for the receiver
to respond to every signal with $a^{'}$, which is a type-heterogeneous
rationalizable self-confirming equilibrium. However, this pooling
behavior cannot occur in a Nash equilibrium or in a unitary self-confirming
equilibrium, where both sender types must hold the same belief about
how the receiver responds to $s^{'}$. 

To establish this claim, note that since $\delta=0$ each sender plays
the myopically optimal signal after every history. For any $\gamma$,
there is a steady state where the receivers' policy responds to every
signal with $a^{'}$ after every history, type $\theta^{''}$ senders
play $s^{''}$ after every history and never update their prior belief
about how receivers react to $s^{'}$, and type $\theta'$ senders
with fewer than 6 periods of experience play $s^{'}$ but switch to
playing $s^{''}$ forever starting at age 7. The behavior of the $\theta^{'}$
agents is optimal because after $k$ periods of playing $s^{'}$ and
seeing response $a^{'}$ every period, the sender's posterior belief
about $\pi_{2}(\cdot|s^{'})$ is $\text{Beta}(1+k,3)$, so the expected
payoff from playing $s^{'}$ next period is

\[
\frac{1+k}{4+k}(-1)+\frac{3}{4+k}(2).
\]

This expression is positive when $0\le k\le5$ but negative when $k=6$.
The fraction of type $\theta^{'}$ aged 6 and below approaches 0 as
$\gamma\to1$, hence we have constructed a sequence of steady-state
strategy profiles converging to the $s^{''}$ pooling equilibrium.
So even though both types start with the same prior $g_{1}$, their
beliefs about how the receivers react to $s^{'}$ eventually diverge.
\hfill{} $\blacklozenge$
\end{example}
In contrast to the plethora of $\delta$-stable profiles, we now show
that only Nash equilibrium profiles can be steady-state outcomes as
$\delta$ tends to 1. Moreover, this limit also rules out strategy
profiles in which the sender's strategy can only be supported by the
belief that the receiver would play a dominated action in response
to some of the unsent signals. 
\begin{defn}
\label{def:PBE_hetero} In a signaling game, a \emph{perfect Bayesian
equilibrium with heterogeneous off-path beliefs }is a strategy profile
$(\pi_{1}^{*},\pi_{2}^{*})$ such that: 
\begin{itemize}
\item For each $\theta\in\Theta,$ $u_{1}(\theta;\pi^{*})=\max_{s\in S}u_{1}(\theta,s,\pi_{2}^{*}(\cdot|s))$. 
\item For each on-path signal $s$, $u_{2}(p^{*}(\cdot|s),s,\pi_{2}^{*}(\cdot|s))=\underset{\hat{a}\in A}{\max}\ u_{2}(p^{*}(\cdot|s),s,\hat{a})$. 
\item For each off-path signal $s$ and each $a\in A$ with $\pi_{2}^{*}(a|s)>0$,
there exists a belief $p\in\Delta(\Theta)$ such that $u_{2}(p,s,a)=\underset{\hat{a}\in A}{\max}\ u_{2}(p,s,\hat{a})$. 
\end{itemize}
Here $u_{1}(\theta;\pi^{*})$ refers to type $\theta$'s payoff under
$\pi^{*},$ and $p^{*}(\cdot|s)$ is the Bayesian posterior belief
about sender's type after signal $s$, under strategy $\pi_{1}^{*}$. 
\end{defn}
The first two conditions imply that the profile is a Nash equilibrium.
The third condition resembles that of perfect Bayesian equilibrium,
but is somewhat weaker as it allows the receiver's play after an off-path
signal $s$ to be a mixture over several actions, each of which is
a best response to a different belief about the sender's type. This
means $\pi_{2}^{*}(\cdot|s)\in\Delta(\text{BR}(\Delta(\Theta),s))$,
but $\pi_{2}^{*}(\cdot|s)$ itself may not be a best response to any
unitary belief about the sender's type. 
\begin{prop}
\label{prop:nash} If strategy profile $\pi^{*}$ is patiently stable,
then it is a perfect Bayesian equilibrium with heterogeneous off-path
beliefs. 
\end{prop}
\begin{proof}
In the Online Appendix, we prove that patiently stable profiles must
be Nash equilibria. This argument follows the proof strategy of \citet{fudenberg_steady_1993},
which derived a contradiction via excess option values. In outline,
if $\pi^{*}$ is patiently stable, each player's strategy is a best
response to a belief that is correct about the opponent's on-path
play. Thus if $\pi^{*}$ is not a Nash equilibrium, some type should
perceive a persistent option value to experimenting with some signal
that she plays with probability 0. But this would contradict the fact
that the option values evaluated at sufficiently long histories must
go to 0. We now explain why a patiently stable profile $\pi^{*}$
must satisfy the third condition in Definition \ref{def:PBE_hetero}.
After observing any history $y_{2}$, a receiver who started with
a regular prior thinks every signal has positive probability in his
next match. So, his optimal policy prescribes for each signal $s$
a best response to that receiver's posterior belief about the sender's
type upon seeing signal $s$ after history $y_{2}$. For any regular
prior $g$, $0\le\delta,\gamma<1$, and any sender aggregate play
$\pi_{1}$, we thus deduce $\mathscr{R}_{2}^{g,\delta,\gamma}[\pi_{1}](\cdot|s)$
is entirely supported on $\text{BR}(\Delta(\Theta),s)$. This means
the the same is true about the aggregate receiver response in every
steady state and hence in every patiently stable strategy profile. 
\end{proof}
In \citet{fudenberg_steady_1993}, this argument relies on the finite
lifetime of the agents only to ensure that ``almost all'' histories
are long enough, by picking a large enough lifetime. We can achieve
the analogous effect in our geometric-lifetime model by picking $\gamma$
close to 1. Our proof uses the fact that if $\delta$ is fixed and
$\gamma\to1,$ then the number of experiments that a sender needs
to exhaust her option value is negligible relative to her expected
lifespan, so that most senders play approximate best responses to
their current beliefs. The same conclusion does not hold if we fix
$\gamma$ and let $\delta\to1,$ even though the optimal sender policy
only depends on the product $\delta\gamma$, because for a fixed sender
policy the induced distribution on sender play depends on $\gamma$
but not on $\delta.$ 

\subsection{Patient Stability Implies the Compatibility Criterion }

Proposition \ref{prop:nash} allows the receiver to sustain his off-path
actions using any belief $p\in\Delta(\Theta)$. We now turn to our
main result, which focuses on refining off-path beliefs. We prove
that patient stability selects a strict subset of the Nash equilibria,
namely those that satisfy the\emph{ compatibility criterion. }
\begin{defn}
\label{def:J}For a fixed strategy profile $\pi^{*}$, let $u_{1}(\theta;\pi^{*})$
denote the payoff to type $\theta$ under $\pi^{*},$ and let 

\begin{eqnarray*}
J(s,\pi^{*}) & \coloneqq & \left\{ \theta\in\Theta:\underset{a\in A}{\max}\ u_{1}(\theta,s,a)>u_{1}(\theta;\pi^{*})\right\} 
\end{eqnarray*}
be the set of types for which \emph{some} response to signal $s$
is strictly better than their payoff under $\pi^{*}.$ Signal $s$
is \emph{weakly equilibrium dominated} for types in the complement
of $J(s,\pi^{*})$. 

The \emph{admissible beliefs at signal $s$} \emph{under profile}
$\pi^{*}$ are 

\[
P(s,\pi^{*})\coloneqq\bigcap\left\{ P_{\theta^{'}\triangleright\theta^{''}}:\theta^{'}\succ_{s}\theta^{''}\text{ and }\theta^{'}\in J(s,\pi^{*})\right\} 
\]

where $P_{\theta^{'}\triangleright\theta^{''}}$ is defined in Equation
(\ref{eq:odds_ratio_P}). 
\end{defn}
That is, $P(s,\pi^{*})$ is the joint belief restriction imposed by
a family of $P_{\theta^{'}\triangleright\theta^{''}}$ for $(\theta^{'},\theta^{''})$
satisfying two conditions: $\theta^{'}$ is more type-compatible with
$s$ than $\theta^{''},$ and furthermore the more compatible type
$\theta^{'}$ belongs to $J(s,\pi^{*})$. If there are no pairs $(\theta^{'},\theta^{''})$
satisfying these two conditions, then (by convention of intersection
over no elements) $P(s,\pi^{*})$ is defined as $\Delta(\Theta)$.
In any signaling game and for any $\pi^{*}$, the set $P(s,\pi^{*})$
is always nonempty because it always contains the prior $\lambda$. 
\begin{defn}
Strategy profile $\pi^{*}$ \emph{satisfies the compatibility criterion
}if $\pi_{2}(\cdot|s)\in\Delta(\text{BR}(P(s,\pi^{*}),s))$ for every
$s$. 
\end{defn}
Like divine equilibrium but unlike the Intuitive Criterion or \citet{cho_signaling_1987}'s
D1 criterion, the compatibility criterion says only that some signals
should not increase the relative probability of ``implausible''
types, as opposed to requiring that these types have probability 0. 

One might imagine a version of the compatibility criterion where the
belief restriction $P_{\theta^{'}\triangleright\theta^{''}}$ applies
whenever $\theta^{'}\succ_{s}\theta^{''}$. To understand why we require
the additional condition that $\theta^{'}\in J(s,\pi^{*})$ in the
definition of admissible beliefs, recall that Lemma \ref{lem:receiver_learning}
only gives a learning guarantee in the receiver's problem when $\pi_{1}(s|\theta^{'})$
is ``large enough'' for the more type-compatible $\theta^{'}$.
In the extreme case where $s$ is a strictly dominated signal for
$\theta^{'}$, she will never play it during learning. It turns out
that if $s$ is weakly equilibrium dominated for $\theta^{'}$, then
$\theta^{'}$ may still not experiment very much with it. On the other
hand, the next lemma provides a lower bound on the frequency that
$\theta^{'}$ experiments with $s^{'}$ when $\theta^{'}\in J(s^{'},\pi^{*})$
and $\delta$ and $\gamma$ are close to 1. 
\begin{lem}
\label{lem:nondom_message} Fix a regular prior $g$ and a strategy
profile $\pi^{*}$ where for some type $\theta^{'}$ and signal $s^{'}$,
$\theta^{'}\in J(s^{'},\pi^{*})$. There exist a number $\epsilon\in(0,1)$
and threshold functions $\bar{\delta}:\mathbb{N}\to(0,1)$ and $\bar{\gamma}:\mathbb{N}\times(0,1)\to(0,1)$
such that whenever $\pi\in\Pi^{*}(g,\delta,\gamma)$ with $\delta\ge\bar{\delta}(N)$
and $\gamma\ge\bar{\gamma}(N,\delta)$ and $\pi$ is no more than
$\epsilon$ away from $\pi^{*}$ in $L_{1}$ distance\footnote{\label{fn:l1_norm} The $L_{1}$ distance is 
\[
d(\pi,\pi^{*})=\sum_{\theta\in\Theta}\sum_{s\in S}|\pi_{1}(s|\theta)-\pi_{1}^{*}(s|\theta)|+\sum_{s\in S}\sum_{a\in A}|\pi_{2}(a|s)-\pi_{2}^{*}(a|s)|.
\]
}, we have $\pi_{1}(s^{'}|\theta^{'})\ge(1-\gamma)\cdot N.$
\end{lem}
Note that since $\pi_{1}(s|\theta^{'})$ is between 0 and 1, we know
that $(1-\bar{\gamma}(N,\delta))\cdot N<1$ for each $N$. 

The proof of this lemma is in the Online Appendix. To gain an intuition
for it, suppose that not only is $s^{'}$ equilibrium undominated
in $\pi^{*},$ but furthermore $s^{'}$ can lead to the highest signaling
game payoff for type $\theta^{'}$ under some receiver response $a^{'}$.
Because the prior is non-doctrinaire, the Gittins index of each signal
in the learning problem approaches its highest possible payoff in
the stage game as the sender becomes infinitely patient. Therefore,
for every $N\in\mathbb{N}$, when $\gamma$ and $\delta$ are close
enough to 1, a new type $\theta^{'}$ will play $s^{'}$ in each of
the first $N$ periods of her life, regardless of what responses she
receives during that time. These $N$ periods account for roughly
$(1-\gamma)\cdot N$ fraction of her life, proving the lemma in this
special case. It turns out that even if $s^{'}$ does not lead to
the highest potential payoff in the signaling game, long-lived players
will have a good estimate of their steady-state payoff. So, type $\theta'$
will still play any $s^{'}$ that is equilibrium undominated in strategy
profile $\pi^{*}$ at least $N$ times in any steady states that are
sufficiently close to $\pi^{*}$, though these $N$ periods may not
occur at the beginning of her life. 
\begin{thm}
\label{thm:PS_is_compatible} Every patiently stable strategy profile
$\pi^{*}$ satisfies the compatibility criterion. 
\end{thm}
The proof combines Lemma \ref{lem:compatible_comonotonic}, Lemma
\ref{lem:receiver_learning}, and Lemma \ref{lem:nondom_message}.
Lemma \ref{lem:compatible_comonotonic} shows that types that are
more compatible with $s^{'}$ play it more often. Lemma \ref{lem:nondom_message}
says that types for whom $s^{'}$ is not weakly equilibrium dominated
will play it ``many times.'' Finally, Lemma \ref{lem:receiver_learning}
shows that the ``many times'' here is sufficiently large that most
receivers correctly believe that more compatible types play $s^{'}$
more than less compatible types do, so their posterior odds ratio
for more versus less compatible types exceeds the prior ratio. 
\begin{proof}
Suppose $\pi^{*}$ is patiently stable under regular prior $g$. Fix
an $s^{'}$ and an action $\hat{a}\notin\text{BR}(P(s^{'},\pi^{*}),s^{'})$.
Let $h>0$ be given. We will show that $\pi_{2}^{*}(\hat{a}|s^{'})<h$.
Since the choices of $s^{'}$, $\hat{a}$, and $h>0$ are arbitrary,
we will have proven the theorem. 

\textbf{Step 1}: Setting some constants. 

In the statement of Lemma \ref{lem:receiver_learning}, for each pair
$\theta^{'},\theta^{''}$ such that $\theta^{'}\succ_{s^{'}}\theta^{''}\text{ and }\theta^{'}\in J(s^{'},\pi^{*})$,
put $\epsilon=\frac{h}{2|\Theta|^{2}}$ and find $C_{\theta^{'},\theta^{''}}$
and $\underline{\gamma}_{\theta^{'},\theta^{''}}$ so that the result
holds. Let $C$ be the maximum of all such $C_{\theta^{'},\theta^{''}}$
and $\underline{\gamma}$ be the maximum of all such $\underline{\gamma}_{\theta^{'},\theta^{''}}$.
Also find $\underline{n}\ge1$ so that 
\begin{equation}
1-\frac{1}{\underline{n}}>1-\frac{h}{2|\Theta|^{2}}.\label{eq:frac_old}
\end{equation}

In the statement of Lemma \ref{lem:nondom_message}, for each $\theta^{'}$
such that $\theta^{'}\succ_{s^{'}}\theta^{''}$ for at least one $\theta^{''}$,
find $\epsilon_{\theta^{'}},\bar{\delta}_{\theta^{'}}(\underline{n}C)$,
$\bar{\gamma}_{\theta^{'}}(\underline{n}C,\delta)$ so that the lemma
holds. Write $\epsilon^{*}>0$ as the minimum of all such $\epsilon_{\theta^{'}}$
and let $\bar{\delta}^{*}(\underline{n}C)$ and $\bar{\gamma}^{*}(\underline{n}C,\delta)$
represent the maximum of $\delta_{\theta^{'}}$ and $\gamma_{\theta^{'}}$
across such $\theta^{'}$. 

\textbf{Step 2}: Finding a steady-state profile with large $\delta,\gamma$
that approximates $\pi^{*}$. 

Since $\pi^{*}$ is patiently stable under $g$, there exists a sequence
of strategy profiles $\pi^{(j)}\to\pi^{*}$ where $\pi^{(j)}$ is
$\delta_{j}$-stable under $g$ with $\delta_{j}\to1$. Each $\pi^{(j)}$
can be written as the limit of steady-state strategy profiles. That
is, for each $j$, there exists $\gamma_{j,k}\to1$ and a sequence
of steady-state profiles $\pi^{(j,k)}\in\Pi^{*}(g,\delta_{j},\gamma_{j,k})$
such that $\lim_{k\to\infty}\pi^{(j,k)}=\pi^{(j)}$. 

The convergence of the array $\pi^{(j,k)}$ to $\pi^{*}$ means we
may find $\underline{j}\in\mathbb{N}$ and function $k(j)$ so that
whenever $j\ge\underline{j}$ and $k\ge k(j),$ $\pi^{(j,k)}$ is
no more than $\min(\epsilon^{*},\frac{h}{2|\Theta|^{2}})$ away from
$\pi^{*}$. Find $j^{\circ}\ge\underline{j}$ large enough so $\delta^{\circ}:=\delta_{j^{\circ}}>\bar{\delta}^{*}(\underline{n}C)$,
and then find a large enough $k^{\circ}>k(j^{\circ})$ so that $\gamma^{\circ}:=\gamma_{j^{\circ},k^{\circ}}>\max(\bar{\gamma}^{*}(\underline{n}C,\delta^{\circ}),\underline{\gamma})$.
So we have identified a steady-state profile $\pi^{\circ}:=\pi^{(j^{\circ},k^{\circ})}\in\Pi^{*}(g,\delta^{\circ},\gamma^{\circ})$
which approximates $\pi^{*}$ to within $\min(\epsilon^{*},\frac{h}{2|\Theta|^{2}})$. 

\textbf{Step 3}: Applying properties of $\mathscr{R}_{1}$ and $\mathscr{R}_{2}$. 

For each pair $\theta^{'},\theta^{''}$ such that $\theta^{'}\succ_{s^{'}}\theta^{''}\text{ and }\theta^{'}\in J(s^{'},\pi^{*})$,
we will bound the probability that $\pi_{2}^{\circ}(\cdot|s^{'})$
does not best respond to $P_{\theta^{'}\triangleright\theta^{''}}$
by $\frac{h}{|\Theta|^{2}}$. Since there are at most $|\Theta|\cdot(|\Theta|-1)$
such pairs in the intersection defining $P(s^{'},\pi^{*})$, this
would imply that $\pi_{2}^{\circ}(\hat{a}|s^{'})<[|\Theta|\cdot(|\Theta|-1)]\cdot\frac{h}{|\Theta|^{2}}$
since $\hat{a}\notin\text{BR}(P(s^{'},\pi^{*}),s^{'})$. And since
$\pi_{2}^{\circ}$ is no more than $\frac{h}{2|\Theta|^{2}}$ away
from $\pi_{2},$ this would show $\pi_{2}(\hat{a}|s^{'})<h$. 

By construction $\pi^{\circ}$ is closer than $\epsilon_{\theta^{'}}$
to $\pi^{*}$, and furthermore $\delta^{\circ}\ge\bar{\delta}_{\theta^{'}}(\underline{n}C)$
and $\gamma^{\circ}\ge\bar{\gamma}_{\theta^{'}}(\underline{n}C,\delta^{\circ})$.
By Lemma \ref{lem:nondom_message}, $\pi_{1}^{\circ}(s^{'}|\theta^{'})\ge\underline{n}C(1-\gamma^{\circ})$.
At the same time, $\pi_{1}^{\circ}=\mathscr{R}_{1}[\pi_{2}^{\circ}]$
and $\theta^{'}\succ_{s^{'}}\theta^{''}$, so Lemma \ref{lem:compatible_comonotonic}
implies that $\pi_{1}^{\circ}(s^{'}|\theta^{'})\ge\pi_{1}^{\circ}(s^{'}|\theta^{''})$.
Turning to the receiver side, $\pi_{2}^{\circ}=\mathscr{R}_{2}[\pi_{1}^{\circ}]$
with $\pi_{^{\circ}1}$ satisfying the conditions of Lemma \ref{lem:receiver_learning}
associated with $\epsilon=\frac{h}{2|\Theta|^{2}}$ and $\gamma^{\circ}\ge\underline{\gamma}$.
Therefore, we conclude 
\[
\pi_{2}^{\circ}(\text{BR}(P_{\theta^{'}\triangleright\theta^{''}},s^{'})\ |\ s^{'})\ge1-\frac{1}{\underline{n}}-\frac{h}{2|\Theta|^{2}}.
\]
But by construction of $\underline{n}$ in Equation (\ref{eq:frac_old}),
$1-\frac{1}{\underline{n}}>1-\frac{h}{2|\Theta|^{2}}$. So the LHS
is at least $1-\frac{h}{|\Theta|^{2}}$, as desired. 
\end{proof}
\begin{rem}
\label{rem:general_learning_model} More generally, consider \emph{any}
model for our populations of agents with geometrically distributed
lifetimes that generates aggregate response functions $\mathscr{R}_{1}$
and $\mathscr{R}_{2}$. Defining the steady states under $(g,\delta,\gamma)$
as the strategy profiles $\pi^{*}$ such that $\mathscr{R}_{1}^{g,\delta,\gamma}(\pi_{2}^{*})=\pi_{1}^{*}$
and $\mathscr{R}_{2}^{g,\delta,\gamma}(\pi_{1}^{*})=\pi_{2}^{*}$,
the proof of Theorem \ref{thm:PS_is_compatible} applies to the patiently
stable profiles of the new learning model provided that $\mathscr{R}_{1}$
satisfies the conclusion of Lemma \ref{lem:compatible_comonotonic},
$\mathscr{R}_{2}$ satisfies the conclusion of Lemma \ref{lem:receiver_learning},
and Lemma \ref{lem:nondom_message} is valid for $(\theta^{'},s^{'})$
pairs such that $\theta^{'}\succ_{s^{'}}\theta^{''}$ for at least
one type $\theta^{''}$and $\text{ }\theta^{'}\in J(s^{'},\pi^{*})$.
\end{rem}
We outline two such more general learning models below. (The proof
is in the Online Appendix.) 
\begin{cor}
\label{cor:1} With either of the following modifications of the steady-state
learning model from Section \ref{sec:Model}, every patiently stable
strategy profile still satisfies the compatibility criterion. 
\begin{enumerate}
\item \textbf{Heterogeneous priors}. There is a finite collection of regular
sender priors $\{g_{1,k}\}_{k=1}^{n}$ and a finite collection of
regular receiver priors $\{g_{2,k}\}_{k=1}^{n}$. Upon birth, an agent
is endowed with a random prior, where the distributions over priors
are $\mu_{1}$ and $\mu_{2}$ for senders and receivers. An agent's
prior is independent of her payoff type, and furthermore no one ever
observes another person's prior. 
\item \textbf{Social learning}. Suppose $1-\alpha$ fraction of the senders
are ``normal learners'' as described in Section \ref{sec:Model},
but the remaining $0<\alpha<1$ fraction are ``social learners.''
At the end of each period, a social learner can observe the extensive-form
strategies of her matched receiver and of $c>0$ other matches sampled
uniformly at random. Each sender knows whether she is a normal learner
or a social learner upon birth, which is uncorrelated with her payoff
type. Receivers cannot distinguish between the two kinds of senders. 
\end{enumerate}
\end{cor}
\setcounter{example}{0}
\begin{example}
[Continued] The beer-quiche game of Example \ref{exa:beer-quiche}
has two components of Nash equilibria: ``beer-pooling equilibria''
where both types play \textbf{Beer} with probability 1, and ``quiche-pooling
equilibria'' where both types play \textbf{Quiche} with probability
1. In a quiche-pooling equilibrium $\pi^{*}$, type $\theta_{\text{strong}}$'s
equilibrium payoff is 2, so $\theta_{\text{strong}}\in J(\text{\textbf{Beer}},\pi^{*})$
since $\theta_{\text{strong}}$'s highest possible payoff under \textbf{Beer}
is 3, and we have already shown that $\text{\ensuremath{\theta}}_{\text{strong}}\succ_{\text{\textbf{Beer}}}\theta_{\text{weak}}$.
So, 

\[
P(\text{\textbf{Beer}},\pi^{*})=\left\{ p\in\Delta(\Theta):\frac{p(\text{\ensuremath{\theta}}_{\text{weak}})}{p(\text{\ensuremath{\theta}}_{\text{strong}})}\le\frac{\lambda(\text{\ensuremath{\theta}}_{\text{weak}})}{\text{\ensuremath{\lambda}(\ensuremath{\theta}}_{\text{strong}})}=1/9\right\} .
\]

\textbf{Fight} is not a best response after \textbf{Beer} to any such
belief, so equilibria in which \textbf{Fight} occurs with positive
probability after \textbf{Beer} do not satisfy the compatibility criterion,
and thus no quiche-pooling equilibrium is patiently stable. Since
the set of patiently stable outcomes is a nonempty subset of the set
of Nash equilibria, pooling on beer is the unique patiently stable
outcome.

By Corollary \ref{cor:1}, quiche-pooling equilibria are still not
patiently stable in more general learning models involving either
heterogeneous priors or social learners.\hfill{} $\blacklozenge$
\end{example}
\setcounter{example}{2}

\subsection{Patient Stability and Equilibrium Dominance }

In generic signaling games, equilibria where the receiver plays a
pure strategy must satisfy a stronger condition than the compatibility
criterion to be patiently stable. 
\begin{defn}
Let 
\[
\widetilde{J}(s,\pi^{*})\coloneqq\left\{ \theta\in\Theta:\underset{a\in A}{\max}\ u_{1}(\theta,s,a)\geq u_{1}(\theta;\pi^{*})\right\} .
\]

If $\widetilde{J}(s^{'},\pi^{*})$ is nonempty, define the \emph{strongly
admissible beliefs at signal $s^{'}$} \emph{under profile} $\pi^{*}$
to be 

\[
\tilde{P}(s^{'},\pi^{*})\coloneqq\text{\ensuremath{\Delta(\widetilde{J}(s^{'},\pi^{*}))} }\bigcap\left\{ P_{\theta^{'}\triangleright\theta^{''}}:\theta^{'}\succ_{s^{'}}\theta^{''}\right\} 
\]

where $P_{\theta^{'}\triangleright\theta^{''}}$ is defined in Equation
(\ref{eq:odds_ratio_P}). Otherwise, define $\tilde{P}(s^{'},\pi^{*}):=\Delta(\Theta)$. 
\end{defn}
Here, $\widetilde{J}(s,\pi^{*})$ is the set of types for which \emph{some}
response to signal $s$ is at least as good as their equilibrium payoff
under $\pi^{*}$ \textemdash{} that is, the set of types for whom
$s$ is not equilibrium dominated in the sense of \citet{cho_signaling_1987}.
Note that $\widetilde{P},$ unlike $P,$ assigns probability 0 to
equilibrium-dominated types, which is the belief restriction of the
Intuitive Criterion. 
\begin{defn}
A Nash equilibrium $\pi^{*}$ is \emph{on-path strict for the receiver}
if for every on-path signal $s^{*},$ $\pi_{2}(a^{*}|s^{*})=1$ for
some $a^{*}\in A$ and $u_{2}(s^{*},a^{*},\pi_{1})>\max_{a\ne a^{*}}u_{2}(s^{*},a,\pi_{1})$. 
\end{defn}
Of course, the receiver cannot have strict ex ante preferences over
play at unreached information sets; this condition is called ``on-path
strict'' because it places no restrictions on the receiver's incentives
after off-path signals. In generic signaling games, all pure-strategy
equilibria are on-path strict for the receiver, but the same is not
true for mixed-strategy equilibria. 
\begin{defn}
A strategy profile $\pi^{*}$ satisfies the \emph{strong compatibility
criterion} if at every signal $s^{'}$ we have

\[
\pi_{2}^{*}(\cdot|s^{'})\in\Delta(\text{BR}(\widetilde{P}(s^{'},\pi^{*}),s^{'})).
\]
\end{defn}
It is immediate that the strong compatibility criterion implies the
compatibility criterion, since it places more stringent restrictions
on the receiver's behavior. It is also immediate that the strong compatibility
criterion implies the Intuitive Criterion.
\begin{thm}
\label{thm:eqm_dominated_types} Suppose $\pi^{*}$ is on-path strict
for the receiver and patiently stable. Then it satisfies the strong
compatibility criterion. 
\end{thm}
The proof of this theorem appears in Appendix \ref{subsec:eqm_dominated}.
The main idea is that when off-path signal $s^{'}$ is equilibrium
dominated in $\pi^{*}$ for type $\theta^{\text{D}}$ but not even
weakly equilibrium dominated for type $\theta^{\text{U}}$, type $\theta^{\text{U}}$
will experiment ``infinitely more often'' with $s^{'}$ than $\theta^{\text{D}}$
does. Indeed, we can provide an upper bound on the steady-state probability
that $\theta^{\text{D}}$ ever switches away from its equilibrium
signal $s^{*}$ after trying it for the first time\footnote{This upper bound does not apply when $\pi^{*}$ is not on-path strict
for the receiver. When $\pi^{*}$ involves the receiver strictly mixing
between several responses after $s^{*}$, some of these responses
might make $\theta^{\text{D}}$ strictly worse off than her worst
payoff after $s^{'}$, so there is non-vanishing probability that
that $\theta^{\text{D}}$ observes a large number of these bad responses
in a row and then stops playing $s^{*}$. }, which is also an upper bound on how often $\theta^{\text{D}}$ experiments
with $s^{'}$, while Lemma \ref{lem:nondom_message} provides a lower
bound for how often$\theta^{\text{U}}$ plays $s^{'}$. We show there
is a sequence of steady-state profiles $\pi^{(k)}\in\Pi^{*}(g,\delta_{k},\gamma_{k})$
with $\gamma_{k}\to1$ and $\pi^{(k)}\to\pi^{*}$ where the ratio
of the lower bound to the upper bound goes to infinity. Applying Theorem
2 of \citet*{fudenberg_he_imhof_2016}, we can then prove receivers
will infer that an $s^{'}$-sender is ``infinitely more likely''
to be $\theta^{\text{U}}$ than $\theta^{\text{D}}$, which means
receivers must assign probability 0 to $\theta^{\text{D}}$ after
$s^{'}$ in equilibrium $\pi^{*}$. 
\begin{rem}
As noted by \citet{fudenbergKreps1988} and \citet*{sobel1990fixed},
it seems ``intuitive'' that learning and rational experimentation
should lead receivers to assign probability 0 to types that are equilibrium
dominated, so it might seem surprising that this theorem needs the
additional assumption that the equilibrium is on-path strict for the
receiver. However, in our model senders start out initially uncertain
about the receivers' play, and so even types for whom a signal is
equilibrium dominated might initially experiment with it. Showing
that these experiments do not lead to ``perverse'' responses by
the receivers requires some arguments about the \emph{relative} probabilities
with which equilibrium-dominated types and non-equilibrium-dominated
types play off-path signals. When the equilibrium involves on-path
receiver randomization, a nontrivial fraction of receivers could play
an action after a type's equilibrium signal that the type finds strictly
worse than her worst payoff under an off-path signal. In this case,
we do not see how to show that the probability she ever switches away
from her equilibrium signal tends to 0 with patience, since the event
of seeing a large number of these unfavorable responses in a row has
probability bounded away from 0 even when the receiver population
plays exactly their equilibrium strategy. However, we do not have
a counterexample to show that the conclusion of the theorem fails
without on-path strictness for the receiver.
\end{rem}
\begin{example}
\label{exa:modifed_beer_quiche}In the following modified beer-quiche
game, the payoffs of fighting a type $\theta_{\text{weak}}$ who drinks
beer have been substantially increased relative to Example \ref{exa:beer-quiche},
so that \textbf{Fight} is now a best response to the prior belief
$\lambda$ after \textbf{Beer}. 

\begin{center}

\includegraphics[scale=0.3]{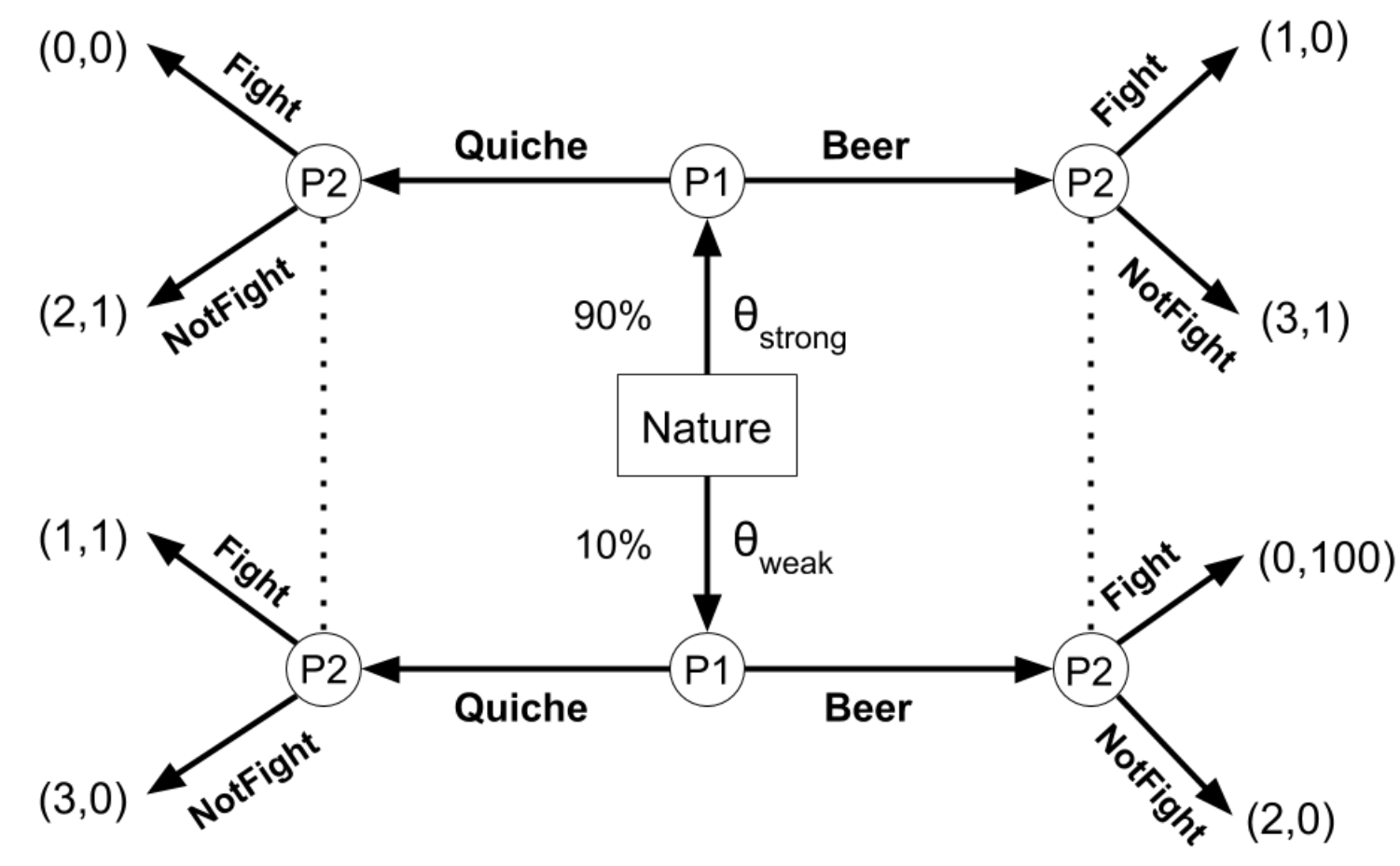}

\noindent \end{center}

Since the prior $\lambda$ is always an admissible belief in any signaling
game after any signal, the Nash equilibrium $\pi^{*}$ where both
types play \textbf{Quiche} (supported by the receiver playing \textbf{Fight}
after \textbf{Beer}) is not ruled out by the compatibility criterion,
unlike in Example \ref{exa:beer-quiche}. However, this equilibrium
is ruled out by the strong compatibility criterion. To see why, note
that this pooling equilibrium is on-path strict for the receiver,
because the receiver has a strict preference for \textbf{NotFight}
at the only on-path signal, \textbf{Quiche}. Moreover, $\pi^{*}$
does not satisfy the strong compatibility criterion, because $\widetilde{J}(\text{\textbf{Beer}},\pi^{*})=\{\theta_{\text{strong}}\}$
implies the only strongly admissible belief after \textbf{Beer} assigns
probability 1 to the sender being $\theta_{\text{strong}}$. Thus
Theorem \ref{thm:eqm_dominated_types} implies that this equilibrium
is not patiently stable.\hfill{} $\blacklozenge$
\end{example}

\section{Discussion }

Our learning model supposes that the agents have geometrically distributed
lifetimes, which is one of the reasons that the senders' optimization
problems can be solved using the Gittins index. If agents were to
have fixed finite lifetimes, as in \citet{fudenberg_steady_1993,fudenberg_superstition_2006},
their optimization problem would not be stationary, and the finite-horizon
analog of the Gittins index is only approximately optimal for the
finite-horizon multi-armed bandit problem \citep{nino2011computing}.
Applying the geometric-lifetime framework to steady-state learning
models for other classes of extensive-form games could prove fruitful,
especially for games where we need to compare the behavior of various
players or player types, and in studies of other sorts of dynamic
decisions. 

Theorem \ref{thm:index} provides a comparison between the dynamic
behavior of two agents in a geometric-lifetime bandit problem based
on their static preferences over the prizes. As an immediate application,
consider a principal-agent setting where the agent faces a multi-armed
bandit with arms $s\in S$, where $s$ leads a prize drawn from $Z_{s}$
according to some distribution. The principal knows the agent's per-period
utility function $u:\cup_{s}Z_{s}\to\mathbb{R}$, but not the agent's
beliefs over the prize distributions of different arms or agent's
discount factor. Suppose the principal observes the agent choosing
arm 1 in the first period. The principal can impose taxes and subsidies
on the different prizes and arms, changing the agent's utility function
to $\tilde{u}$. For what taxes and subsidies would the agent still
have chosen arm 1 in the first period, irrespective of her initial
beliefs and discount factor? According to Theorem \ref{thm:index},
the answer is precisely those taxes and subsidies such that arm 1
is more type-compatible with $\tilde{u}$ than $u$. 

Our results provide an upper bound on the set of patiently stable
strategy profiles in a signaling game. In \citet*{FudenbergHe2017TCE},
we provided a lower bound for the same set, as well as a sharper upper
bound under additional restrictions on the priors. But together, these
results will not give an exact characterization of patiently stable
outcomes. Nevertheless, our results do show how the theory of learning
in games provides a foundation for refining the set of equilibria
in signaling games.

In future work, we hope to investigate a learning model featuring
temporary sender types. Instead of the sender's type being assigned
at birth and fixed for life, at the start of each period each sender
takes an i.i.d. draw from $\lambda$ to discover her type for that
period. When the players are impatient, this yields different steady
states than the fixed-type model here, as noted by \citet*{dekel_learning_2004}.
This model will require different tools to analyze, since the sender's
problem becomes a restless bandit.

\bibliographystyle{ecta}
\bibliography{Gittins_eqm}

\begin{thebibliography}{23}
\newcommand{\enquote}[1]{``#1''}
\expandafter\ifx\csname natexlab\endcsname\relax\def\natexlab#1{#1}\fi

\bibitem[\protect\citeauthoryear{Banks and Sobel}{Banks and
  Sobel}{1987}]{banks_equilibrium_1987}
\textsc{Banks, J.~S. and J.~Sobel} (1987): \enquote{Equilibrium {Selection} in
  {Signaling} {Games},} \emph{Econometrica}, 55, 647--661.

\bibitem[\protect\citeauthoryear{Bellman}{Bellman}{1956}]{bellman1956problem}
\textsc{Bellman, R.} (1956): \enquote{A {Problem} in the {Sequential} {Design}
  of {Experiments},} \emph{Sankhy{\=a}: The Indian Journal of Statistics
  (1933-1960)}, 16, 221--229.

\bibitem[\protect\citeauthoryear{Billingsley}{Billingsley}{1995}]{billingsley1995probability}
\textsc{Billingsley, P.} (1995): \emph{Probability and Measure}, John Wiley \&
  Sons.

\bibitem[\protect\citeauthoryear{Cho and Kreps}{Cho and
  Kreps}{1987}]{cho_signaling_1987}
\textsc{Cho, I.-K. and D.~M. Kreps} (1987): \enquote{Signaling {Games} and
  {Stable} {Equilibria},} \emph{Quarterly Journal of Economics}, 102, 179--221.

\bibitem[\protect\citeauthoryear{Dekel, Fudenberg, and Levine}{Dekel
  et~al.}{1999}]{dekel1999payoff}
\textsc{Dekel, E., D.~Fudenberg, and D.~K. Levine} (1999): \enquote{Payoff
  {Information} and {Self-Confirming} {Equilibrium},} \emph{Journal of Economic
  Theory}, 89, 165--185.

\bibitem[\protect\citeauthoryear{Dekel, Fudenberg, and Levine}{Dekel
  et~al.}{2004}]{dekel_learning_2004}
---\hspace{-.1pt}---\hspace{-.1pt}--- (2004): \enquote{Learning to play
  {Bayesian} games,} \emph{Games and Economic Behavior}, 46, 282--303.

\bibitem[\protect\citeauthoryear{Diaconis and Freedman}{Diaconis and
  Freedman}{1990}]{diaconis_uniform_1990}
\textsc{Diaconis, P. and D.~Freedman} (1990): \enquote{On the {Uniform}
  {Consistency} of {Bayes} {Estimates} for {Multinomial} {Probabilities},}
  \emph{Annals of Statistics}, 18, 1317--1327.

\bibitem[\protect\citeauthoryear{Esponda and Pouzo}{Esponda and
  Pouzo}{2016}]{esponda_berknash_2016}
\textsc{Esponda, I. and D.~Pouzo} (2016): \enquote{Berk-{Nash} {Equilibrium}:
  {A} {Framework} for {Modeling} {Agents} {With} {Misspecified} {Models},}
  \emph{Econometrica}, 84, 1093--1130.

\bibitem[\protect\citeauthoryear{Fudenberg and He}{Fudenberg and
  He}{2017}]{FudenbergHe2017TCE}
\textsc{Fudenberg, D. and K.~He} (2017): \enquote{Learning and Equilibrium
  Refinements in Signalling Games,} \emph{Mimeo}.

\bibitem[\protect\citeauthoryear{Fudenberg, He, and Imhof}{Fudenberg
  et~al.}{2017}]{fudenberg_he_imhof_2016}
\textsc{Fudenberg, D., K.~He, and L.~A. Imhof} (2017): \enquote{Bayesian
  posteriors for arbitrarily rare events,} \emph{Proceedings of the National
  Academy of Sciences}, 114, 4925--4929.

\bibitem[\protect\citeauthoryear{Fudenberg and Kreps}{Fudenberg and
  Kreps}{1988}]{fudenbergKreps1988}
\textsc{Fudenberg, D. and D.~M. Kreps} (1988): \enquote{A Theory of Learning,
  Experimentation, and Equilibrium in Games,} \emph{Mimeo}.

\bibitem[\protect\citeauthoryear{Fudenberg and Kreps}{Fudenberg and
  Kreps}{1993}]{fudenberg1993learning}
---\hspace{-.1pt}---\hspace{-.1pt}--- (1993): \enquote{Learning {Mixed}
  {Equilibria},} \emph{Games and Economic Behavior}, 5, 320--367.

\bibitem[\protect\citeauthoryear{Fudenberg and Kreps}{Fudenberg and
  Kreps}{1994}]{fudenbergKreps1994learning}
---\hspace{-.1pt}---\hspace{-.1pt}--- (1994): \enquote{Learning in
  Extensive-Form Games, II: Experimentation and Nash Equilibrium,}
  \emph{Mimeo}.

\bibitem[\protect\citeauthoryear{Fudenberg and Kreps}{Fudenberg and
  Kreps}{1995}]{fudenberg_learning_1995}
---\hspace{-.1pt}---\hspace{-.1pt}--- (1995): \enquote{Learning in
  {Extensive-Form} {Games} {I}. {Self}-{Confirming} {Equilibria},} \emph{Games
  and Economic Behavior}, 8, 20--55.

\bibitem[\protect\citeauthoryear{Fudenberg and Levine}{Fudenberg and
  Levine}{1993}]{fudenberg_steady_1993}
\textsc{Fudenberg, D. and D.~K. Levine} (1993): \enquote{Steady {State}
  {Learning} and {Nash} {Equilibrium},} \emph{Econometrica}, 61, 547--573.

\bibitem[\protect\citeauthoryear{Fudenberg and Levine}{Fudenberg and
  Levine}{2006}]{fudenberg_superstition_2006}
---\hspace{-.1pt}---\hspace{-.1pt}--- (2006): \enquote{Superstition and
  {Rational} {Learning},} \emph{American Economic Review}, 96, 630--651.

\bibitem[\protect\citeauthoryear{Gittins}{Gittins}{1979}]{gittins1979bandit}
\textsc{Gittins, J.~C.} (1979): \enquote{Bandit {Processes} and {Dynamic}
  {Allocation} {Indices},} \emph{Journal of the Royal Statistical Society.
  Series B (Methodological)}, 148--177.

\bibitem[\protect\citeauthoryear{Jehiel and Samet}{Jehiel and
  Samet}{2005}]{jehiel_learning_2005}
\textsc{Jehiel, P. and D.~Samet} (2005): \enquote{Learning to Play Games in
  Extensive Form by Valuation,} \emph{Journal of Economic Theory}, 124,
  129--148.

\bibitem[\protect\citeauthoryear{Kalai and Lehrer}{Kalai and
  Lehrer}{1993}]{kalai_rational_1993}
\textsc{Kalai, E. and E.~Lehrer} (1993): \enquote{Rational {Learning} {Leads}
  to {Nash} {Equilibrium},} \emph{Econometrica}, 61, 1019--1045.

\bibitem[\protect\citeauthoryear{Laslier and Walliser}{Laslier and
  Walliser}{2015}]{laslier_stubborn_2014}
\textsc{Laslier, J.-F. and B.~Walliser} (2015): \enquote{Stubborn {Learning},}
  \emph{Theory and Decision}, 79, 51--93.

\bibitem[\protect\citeauthoryear{Ni{\~n}o-Mora}{Ni{\~n}o-Mora}{2011}]{nino2011computing}
\textsc{Ni{\~n}o-Mora, J.} (2011): \enquote{Computing a {Classic} {Index} for
  {Finite-Horizon} {Bandits},} \emph{INFORMS Journal on Computing}, 23,
  254--267.

\bibitem[\protect\citeauthoryear{Sobel, Stole, and Zapater}{Sobel
  et~al.}{1990}]{sobel1990fixed}
\textsc{Sobel, J., L.~Stole, and I.~Zapater} (1990):
  \enquote{{Fixed-Equilibrium} {Rationalizability} in {Signaling} {Games},}
  \emph{Journal of Economic Theory}, 52, 304--331.

\bibitem[\protect\citeauthoryear{Spence}{Spence}{1973}]{spence_job_1973}
\textsc{Spence, M.} (1973): \enquote{Job {Market} {Signaling},} \emph{Quarterly
  Journal of Economics}, 87, 355--374.

\end{thebibliography}

\appendix
\counterwithin{lem}{section}

\section{Appendix \textendash{} Relegated Proofs}

\subsection{\label{subsec:Proof-of-properties_of_comp_relation} Proof of Proposition
\ref{prop:properties_of_comp_relation}}

\textbf{Proposition \ref{prop:properties_of_comp_relation}}:
\begin{enumerate}
\item $\succ_{s^{'}}$ is transitive. 
\item Except when $s^{'}$ is either strictly dominant for both $\theta^{'}$
and $\theta^{''}$ or strictly dominated for both $\theta^{'}$ and
$\theta^{''}$, $\theta^{'}\succ_{s^{'}}\theta^{''}$ implies $\theta^{''}\not\succ_{s^{'}}\theta^{'}$. 
\end{enumerate}
\begin{proof}
To show (i), suppose $\theta^{'}\succ_{s^{'}}\theta^{''}$ and $\theta^{''}\succ_{s^{'}}\theta^{'''}.$
For any $\pi_{2}\in\Pi_{2}$ where $s^{'}$ is weakly optimal for
$\theta^{'''},$ it must be strictly optimal for $\theta^{''},$ hence
also strictly optimal for $\theta^{'}$. This shows $\theta^{'}\succ_{s^{'}}\theta^{'''}$. 

To establish (ii), partition the set of receiver strategies as $\Pi_{2}=\Pi_{2}^{+}\cup\Pi_{2}^{0}\cup\Pi_{2}^{-},$
where the three subsets refer to receiver strategies that make $s^{'}$
strictly better, indifferent, or strictly worse than the best alternative
signal for $\theta^{''}$. If the set $\Pi_{2}^{0}$ is nonempty,
then $\theta^{'}\succ_{s^{'}}\theta^{''}$ implies $\theta^{''}\not\succ_{s^{'}}\theta^{'}$.
This is because against any $\pi_{2}\in\Pi_{2}^{0}$, signal $s^{'}$
is strictly optimal for $\theta^{'}$ but only weakly optimal for
$\theta^{''}$. At the same time, if both $\Pi_{2}^{+}$ and $\Pi_{2}^{-}$
are nonempty, then $\Pi_{2}^{0}$ is nonempty. This is because both
$\pi_{2}\mapsto u_{1}(\theta^{''},s^{'},\pi_{2}(\cdot|s^{'}))$ and
$\pi_{2}\mapsto\max_{s^{''}\ne s^{'}}u_{1}(\theta^{''},s^{''},\pi_{2}(\cdot|s^{''}))$
are continuous functions, so for any $\pi_{2}^{+}\in\Pi_{2}^{+}$
and $\pi_{2}^{-}\in\Pi_{2}^{-},$ there exists $\alpha\in(0,1)$ so
that $\alpha\pi_{2}^{+}+(1-\alpha)\pi_{2}^{-}\in\Pi_{2}^{0}$. If
only $\Pi_{2}^{+}$ is nonempty and $\theta^{'}\succ_{s^{'}}\theta^{''}$,
then $s^{'}$ is strictly dominant for both $\theta^{'}$ and $\theta^{''}$.
If only $\Pi_{2}^{-}$ is nonempty, then we can have $\theta^{''}\succ_{s^{'}}\theta^{'}$
only when $s^{'}$ is never a weak best response for $\theta^{'}$
against any $\pi_{2}\in\Pi_{2}$. 
\end{proof}

\subsection{\label{subsec:Proof-of-Lemma-static} Proof of Lemma \ref{lem:static}}

\textbf{Lemma \ref{lem:static}}: For every signal $s$, stopping
time $\tau$, belief $\nu_{s}$, and discount factor $\beta,$ there
exists $\pi_{2,s}(\tau,\nu_{s},\beta)\in\Delta(A)$ so that for every
$\theta$,

\[
\dfrac{\mathbb{E}_{\nu_{s}}\left\{ \sum_{t=0}^{\tau-1}\beta^{t}\cdot u_{1}(\theta,s,a_{s}(t))\right\} }{\mathbb{E}_{\nu_{s}}\left\{ \sum_{t=0}^{\tau-1}\beta^{t}\right\} }=u_{1}(\theta,s,\pi_{2,s}(\tau,\nu_{s},\beta))
\]

\begin{proof}
\textbf{Step 1: Induced mixed actions. }

A belief $\nu_{s}$ and a stopping time $\tau_{s}$ together define
a stochastic process $(A_{t})_{t\ge0}$ over the space $A\cup\{\emptyset\}$,
where $A_{t}\in A$ corresponds to the receiver action seen in period
$t$ if $\tau_{s}$ has not yet stopped ($\tau_{s}>t$), and $A_{t}\coloneqq\emptyset$
if $\tau_{s}$ has stopped $(\tau_{s}\le t$). Enumerating $A=\{a_{1},...,a_{n}\}$,
we write $p_{t,i}\coloneqq\mathbb{P}_{\nu_{s}}\left[A_{t}=a_{i}\right]$
for $1\le i\le n$ to record the probability of seeing receiver action
$a_{i}$ in period $t$ and $p_{t,0}\coloneqq\mathbb{P}_{\nu_{s}}[A_{t}=\emptyset]=\mathbb{P}_{\nu_{s}}[\tau_{s}\le t]$
for the probability of seeing no receiver action in period $t$ due
to $\tau_{s}$ having stopped. 

Given $\nu_{s}$ and $\tau_{s}$, we define the\emph{ induced mixed
actions after signal $s$,} $\pi_{2,s}(\nu_{s},\tau_{s},\beta)\in\Delta(A)$
by: 

\[
\pi_{2,s}(\nu_{s},\tau_{s},\beta)(a)\coloneqq\frac{\sum_{t=0}^{\infty}\beta^{t}p_{t,i}}{\sum_{t=0}^{\infty}\beta^{t}(1-p_{t,0})}\ \text{for \ensuremath{i} such that \ensuremath{a=a_{i}}}.
\]

As $\sum_{i=1}^{n}p_{t,i}=1-p_{t,0}$ for each $t\ge0$, it is clear
that $\pi_{2,s}(\nu_{s},\tau_{s},\beta)$ puts nonnegative weights
on actions in $A$ that sum to 1, so $\pi_{2,s}(\nu_{s},\tau_{s},\beta)\in\Delta(A)$
may indeed be viewed as a mixture over receiver actions. 

\bigskip{}

\textbf{Step 2: Induced mixed actions and per-period payoff. }

We now show that for any $\beta$ and any stopping time $\tau_{s}$
for signal $s$, the normalized payoff in the stopping problem is
equal to the utility of playing $s$ against $\pi_{2,s}(\nu_{s},\tau_{s},\beta)$
for one period, that is, 
\[
u_{1}(\theta,s,\pi_{2,s}(\nu_{s},\tau_{s},\beta))=\mathbb{E}_{\nu_{s}}\left\{ \sum_{t=0}^{\tau_{s}-1}\beta^{t}\cdot u_{1}(\theta,s,a_{s}(t))\right\} \ /\ \mathbb{E}_{\nu_{s}}\left\{ \sum_{t=0}^{\tau_{s}-1}\beta^{t}\right\} .
\]
 To see why this is true, rewrite the denominator of the right-hand
side as
\[
\mathbb{E}_{\nu_{s}}\left\{ \sum_{t=0}^{\tau_{s}-1}\beta^{t}\right\} =\mathbb{E}_{\nu_{s}}\left\{ \sum_{t=0}^{\infty}\left[1_{\tau_{s}>t}\right]\cdot\beta^{t}\right\} =\sum_{t=0}^{\infty}\beta^{t}\cdot\mathbb{P}_{\nu_{s}}\left[\tau_{s}>t\right]=\sum_{t=0}^{\infty}\beta^{t}(1-p_{t,0}),
\]
 and rewrite the numerator as 
\begin{align*}
\mathbb{E}_{\nu_{s}}\left\{ \sum_{t=0}^{\tau_{s}-1}\beta^{t}\cdot u_{1}(\theta,s,a_{s}(t))\right\}  & =\sum_{t=0}^{\infty}\beta^{t}\cdot\left(\underset{\text{get 0 if already stopped}}{\underbrace{p_{t,0}\cdot0}}+\underset{\text{else, \ensuremath{a_{s}(t)} distributed as \ensuremath{(p_{t,i})}}}{\underbrace{\sum_{i=1}^{n}p_{t,i}\cdot u_{1}(\theta,s,a_{i})}}\right)\\
 & =\sum_{i=1}^{n}\left(\sum_{t=0}^{\infty}\beta^{t}\cdot p_{t,i}\right)\cdot u_{1}(\theta,s,a_{i}).
\end{align*}
So overall, we get as desired:

\begin{align*}
\mathbb{E}_{\nu_{s}}\left\{ \sum_{t=0}^{\tau_{s}-1}\beta^{t}\cdot u_{1}(\theta,s,a_{s}(t))\right\} \ /\ \mathbb{E}_{\nu_{s}}\left\{ \sum_{t=0}^{\tau_{s}-1}\beta^{t}\right\}  & =\sum_{i=1}^{n}\left[\frac{\left(\sum_{t=0}^{\infty}\beta^{t}\cdot p_{t,i}\right)}{\sum_{t=0}^{\infty}\beta^{t}(1-p_{t,0})}\right]\cdot u_{1}(\theta,s,a_{i})\\
 & =u_{1}(\theta,s,\pi_{2,s}(\nu_{s},\tau_{s},\beta)).
\end{align*}
\end{proof}

\subsection{\label{subsec:Proof-of-Lemma-Receiver} Proof of Lemma \ref{lem:receiver_learning}}

\textbf{Lemma \ref{lem:receiver_learning}}: Let regular prior $g_{2}$,
types $\theta^{'},\theta^{''}$, and signal $s^{'}$ be fixed. For
every $\epsilon>0$, there exists $C>0$ and $\underline{\gamma}<1$
so that for any $0\le\delta<1$, $\underline{\gamma}\le\gamma<1$,
and $n\ge1$, if $\pi_{1}(s^{'}|\theta^{'})\ge\pi_{1}(s^{'}|\theta^{''})$
and $\pi_{1}(s^{'}|\theta^{'})\ge(1-\gamma)nC$, then 
\[
\mathscr{R}_{2}[\pi_{1}](\text{BR}(P_{\theta^{'}\triangleright\theta^{''}},s^{'})\ |\ s^{'})\ge1-\frac{1}{n}-\epsilon.
\]

We invoke Theorem 2 of \citet*{fudenberg_he_imhof_2016}, which in
our setting says:
\begin{quote}
\emph{Let regular prior $g_{2}$ and signal $s^{'}$ be fixed. Let
$0<\epsilon,h<1$. There exists $C$ such that whenever $\pi_{1}(s^{'}|\theta^{'})\ge\pi_{1}(s^{'}|\theta^{''})$
and $t\cdot\pi_{1}(s^{'}|\theta^{'})\ge C$, we get}

\emph{
\[
\psi_{2}^{\pi_{1}}\left(y_{2}\in Y_{2}[t]:\frac{p(\theta^{''}|s^{'};y_{2})}{p(\theta^{'}|s^{'};y_{2})}\le\frac{1}{1-h}\cdot\frac{\lambda(\theta^{''})}{\lambda(\theta^{'})}\right)/\psi_{2}^{\pi_{1}}(Y_{2}[t])\ge1-\epsilon
\]
}

\emph{where $p(\theta|s;y_{2})$ refers to the conditional probability
that a sender of $s$ is type $\theta$ according to the posterior
belief induced by history $y_{2}$. }
\end{quote}
That is, if at age $t$ a receiver would have observed in expectation
$C$ instances of type $\theta^{'}$ sending $s^{'}$, then the belief
of at least $1-\epsilon$ fraction of age $t$ receivers (essentially)
falls in $P_{\theta^{'}\triangleright\theta^{''}}$ after seeing the
signal $s^{'}$. The proof of Lemma \ref{lem:receiver_learning} calculates
what fraction of receivers meets this ``age requirement.'' 
\begin{proof}
We will show the following stronger result:

Let regular prior $g_{2}$, types $\theta^{'},\theta^{''}$, and signal
$s^{'}$ be fixed. For every $\epsilon>0$, there exists $C>0$ so
that for any $0\le\delta,\gamma<1$ and $n\ge1$, if $\pi_{1}(s^{'}|\theta^{'})\ge\pi_{1}(s^{'}|\theta^{''})$
and $\pi_{1}(s^{'}|\theta^{'})\ge(1-\gamma)nC$, then 
\[
\mathscr{R}_{2}[\pi_{1}](\text{BR}(P_{\theta^{'}\triangleright\theta^{''}},s^{'})\ |\ s^{'})\ge\gamma^{\left\lceil \frac{1}{n(1-\gamma)}\right\rceil }-\epsilon
\]

The lemma follows because we may pick a large enough $\underline{\gamma}<1$
so that $\gamma^{\left\lceil \frac{1}{n(1-\gamma)}\right\rceil }>1-\frac{1}{n}$
for all $n\ge1$ and $\gamma\ge\underline{\gamma}$. 

For each $0<h<1$, define $P_{\theta^{'}\triangleright\theta^{''}}^{h}:=\left\{ p\in\Delta(\Theta):\frac{p(\theta^{''})}{p(\theta^{'})}\le\frac{1}{1-h}\cdot\frac{\lambda(\theta^{''})}{\lambda(\theta^{'})}\right\} ,$with
the convention that $\frac{0}{0}=0.$ Then it is clear that each $P_{\theta^{'}\triangleright\theta^{''}}^{h}$,
as well as $P_{\theta^{'}\triangleright\theta^{''}}$ itself, is a
closed subset of $\Delta(\Theta)$. Also, $P_{\theta^{'}\triangleright\theta^{''}}^{h}\to P_{\theta^{'}\triangleright\theta^{''}}$
as $h\to0$. 

Fix action $a\in A$. If for all $\bar{h}>0$ there exists some $0<h\le\bar{h}$
so that $a\in\text{BR}(P_{\theta^{'}\triangleright\theta^{''}}^{h},s^{'})$,
then $a\in\text{BR}(P_{\theta^{'}\triangleright\theta^{''}},s^{'})$
also due to best-response correspondence having a closed graph. This
means that, for each $a\notin\text{BR}(P_{\theta^{'}\triangleright\theta^{''}},s^{'})$,
there exists $\bar{h}_{a}>0$ so that $a\notin\text{BR}(P_{\theta^{'}\triangleright\theta^{''}}^{h},s^{'})$
whenever $0<h\le\bar{h}_{a}$. Let $\bar{h}:=\min_{a\notin\text{BR}(P_{\theta^{'}\triangleright\theta^{''}},s^{'})}\bar{h}_{a}$.
Let $\epsilon>0$ be given and apply Theorem 2 of \citet*{fudenberg_he_imhof_2016}
with $\epsilon$ and $\bar{h}$ to find constant $C$. 

When $\pi_{1}(s^{'}|\theta^{'})\ge\pi_{1}(s^{'}|\theta^{''})$ and
$\pi_{1}(s^{'}|\theta^{'})\ge(1-\gamma)nC$, consider an age $t$
receiver for $t\ge\left\lceil \frac{1}{n(1-\gamma)}\right\rceil $.
Since $t\cdot\pi_{1}(s^{'}|\theta^{'})\ge C,$ Theorem 2 of \citet*{fudenberg_he_imhof_2016}
implies there is probability at least $1-\epsilon$ this receiver's
belief about the types who send $s^{'}$ falls in $P_{\theta^{'}\triangleright\theta^{''}}^{\bar{h}}$.
By construction of $\bar{h},$ $\text{BR}(P_{\theta^{'}\triangleright\theta^{''}}^{\bar{h}},s^{'})=\text{BR}(P_{\theta^{'}\triangleright\theta^{''}},s^{'})$,
so $1-\epsilon$ of age $t$ receivers have a history $y_{2}$ where
$\sigma_{2}(y_{2})(s^{'})\in\text{BR}(P_{\theta^{'}\triangleright\theta^{''}},s^{'})$. 

Since agents survive between periods with probability $\gamma,$ the
mass of the receiver population aged $\left\lceil \frac{1}{n(1-\gamma)}\right\rceil $
or older is $(1-\gamma)\cdot\sum_{t=\left\lceil \frac{1}{n(1-\gamma)}\right\rceil }^{\infty}\gamma^{t}=\gamma^{\left\lceil \frac{1}{n(1-\gamma)}\right\rceil }$.This
shows 
\[
\mathscr{R}_{2}[\pi_{1}](\text{BR}(P_{\theta^{'}\triangleright\theta^{''}},s^{'})\ |\ s^{'})\ge\gamma^{\frac{1}{n(1-\gamma)}}\cdot(1-\epsilon)\ge\gamma^{\left\lceil \frac{1}{n(1-\gamma)}\right\rceil }-\epsilon
\]
as desired. 
\end{proof}

\subsection{\label{subsec:pf_mutual_AR} Proof of Proposition \ref{prop:mutual_AR}}

\textbf{Proposition \ref{prop:mutual_AR}}: $\pi^{*}\in\Pi^{*}(g,\delta,\gamma)$
if and only if $\mathscr{R}_{1}^{g,\delta,\gamma}[\pi_{2}^{*}]=\pi_{1}^{*}$
and $\mathscr{R}_{2}^{g,\delta,\gamma}[\pi_{1}^{*}]=\pi_{2}^{*}$. 
\begin{proof}
\textbf{If}: Suppose $\pi^{*}$ is such that $\mathscr{R}_{1}[\pi_{2}^{*}]=\pi_{1}^{*}$
and $\mathscr{R}_{2}[\pi_{1}^{*}]=\pi_{2}^{*}$. Consider the state
$\psi^{*}$ defined as $\psi_{\theta}^{*}:=\psi_{\theta}^{\pi_{2}^{*}}$
for each $\theta$ and $\psi_{2}^{*}:=\psi_{2}^{\pi_{1}^{*}}$. Then,
by construction $\sigma_{\theta}(\psi_{\theta}^{\pi_{2}^{*}})=\pi_{\theta}^{*}$
and $\sigma_{2}(\psi_{2}^{\pi_{1}^{*}})=\pi_{2}^{*}$, so the state
$\psi^{*}$ gives rise to $\pi^{*}$. To verify that $\psi^{*}$ is
a steady state, we can expand by the definition of $\psi_{\theta}^{\pi_{2}^{*}}$,
\[
f_{\theta}(\psi_{\theta}^{\pi_{2}^{*}},\pi_{2}^{*})=f_{\theta}\left(\lim_{T\to\infty}f_{\theta}^{T}(\tilde{\psi}_{\theta},\pi_{2}^{*}),\pi_{2}^{*}\right),
\]
where $\tilde{\psi}_{\theta}$ is any arbitrary initial state. 

Since $f_{\theta}$ is continuous\footnote{This is implied by Step 1 of the proof of Proposition \ref{thm:existence}
in the Online Appendix, which shows $f_{\theta}$ is continuous at
all states that assign $(1-\gamma)\gamma^{t}$ mass to the set of
length-$t$ histories.} at $\psi_{\theta}^{\pi_{2}^{*}}$ in $L_{1}$ distance defined in
Footnote \ref{fn:l1_norm}, $\lim_{T\to\infty}f_{\theta}^{T}(\tilde{\psi}_{\theta},\pi_{2}^{*})=\psi_{\theta}^{\pi_{2}^{*}}$
is a fixed point of $f_{\theta}(\cdot,\pi_{2}^{*}).$ To see this,
write $\psi_{\theta}^{(T)}:=f_{\theta}^{T}(\tilde{\psi}_{\theta},\pi_{2}^{*})$
for each $T\ge1$ and let $\epsilon>0$ be given. Continuity of $f_{\theta}$
implies there is $\zeta>0$ so that $d(f_{\theta}(\psi_{\theta}^{\pi_{2}^{*}},\pi_{2}^{*}),f_{\theta}(\psi_{\theta}^{(T)},\pi_{2}^{*}))<\epsilon/2$
whenever $d(\psi_{\theta}^{\pi_{2}^{*}},\psi_{\theta}^{(T)})<\zeta$.
So pick a large enough $T$ so that $d(\psi_{\theta}^{\pi_{2}^{*}},\psi_{\theta}^{(T)})<\zeta$
and also $d(\psi_{\theta}^{\pi_{2}^{*}},\psi_{\theta}^{(T+1)})<\epsilon/2$.
Then 
\[
d(f_{\theta}(\psi_{\theta}^{\pi_{2}^{*}},\pi_{2}^{*}),\psi_{\theta}^{\pi_{2}^{*}})\le d(f_{\theta}(\psi_{\theta}^{\pi_{2}^{*}},\pi_{2}^{*}),f_{\theta}(\psi_{\theta}^{(T)},\pi_{2}^{*}))+d(\psi_{\theta}^{(T+1)},\psi_{\theta}^{\pi_{2}^{*}})<\epsilon/2+\epsilon/2.
\]

Since $\epsilon>0$ was arbitrary, we have shown that $f_{\theta}(\psi_{\theta}^{\pi_{2}^{*}},\pi_{2}^{*})=\psi_{\theta}^{\pi_{2}^{*}}$
and a similar argument shows $f_{2}(\psi_{2}^{\pi_{1}^{*}},\pi_{1}^{*})=\psi_{2}^{\pi_{1}^{*}}.$
This tells us $\psi^{*}=((\psi_{\theta}^{\pi_{2}^{*}})_{\theta\in\Theta},\psi_{2}^{\pi_{1}^{*}})$
is a steady state. 

\textbf{Only if}: Conversely, suppose $\pi^{*}\in\Pi^{*}(g,\delta,\gamma).$
Then there exists a steady state $\psi^{*}\in\Psi^{*}(g,\delta,\gamma)$
such that $\pi^{*}=\sigma(\psi^{*})$. This means $f_{\theta}(\psi_{\theta}^{*},\pi_{2}^{*})=\psi_{\theta}^{*}$,
so iterating shows 
\[
\psi_{\theta}^{\pi_{2}^{*}}:=\lim_{T\to\infty}f_{\theta}^{T}(\psi_{\theta}^{*},\pi_{2}^{*})=\psi_{\theta}^{*}.
\]
Since $\mathscr{R}_{1}[\pi_{2}^{*}](\cdot|\theta):=\sigma_{\theta}(\psi_{\theta}^{\pi_{2}^{*}})$,
the above implies $\mathscr{R}_{1}[\pi_{2}^{*}](\cdot|\theta)=\sigma_{\theta}(\psi_{\theta}^{*})=\pi_{1}^{*}(\cdot|\theta)$
by the choice of of $\psi^{*}$. We can similarly show $\mathscr{R}_{2}[\pi_{1}^{*}]=\pi_{2}^{*}$. 
\end{proof}

\subsection{\label{subsec:eqm_dominated}Proof of Theorem \ref{thm:eqm_dominated_types}}

Throughout this subsection, we will make use of the following version
of Hoeffding's inequality. 
\begin{fact*}
(Hoeffding's inequality) Suppose $X_{1},...,X_{n}$ are independent
random variables on $\mathbb{R}$ such that $a_{i}\le X_{i}\le b_{i}$
with probability 1 for each $i$. Write $S_{n}\coloneqq\sum_{i=1}^{n}X_{i}$.
Then, 
\[
\mathbb{P}\left[|S_{n}-\mathbb{E}[S_{n}]|\ge d\right]\le2\exp\left(-\frac{2d^{2}}{\sum_{i=1}^{n}(b_{i}-a_{i})^{2}}\right).
\]
\end{fact*}
\begin{lem}
\label{lem:appendix_receiver_rate}In strategy profile $\pi^{*}$,
suppose $s^{*}$ is on-path and $\pi_{2}^{*}(a^{*}|s^{*})=1$, where
$a^{*}$ is a strict best response to $s^{*}$ given $\pi_{1}^{*}$.
Then there exists $N\in\mathbb{R}$ so that, for any regular prior
and any sequence of steady-state strategy profiles $\pi^{(k)}\in\Pi^{*}(g,\delta_{k},\gamma_{k})$
where $\gamma_{k}\to1$,$\pi^{(k)}\to\pi^{*}$, there exists $K\in\mathbb{N}$
such that whenever $k\ge K$, we have $\pi_{2}^{(k)}(a^{*}|s^{*})\ge1-(1-\gamma_{k})\cdot N$. 
\end{lem}
\begin{proof}
Since $a^{*}$ is a strict best response after $s^{*}$ for $\pi_{1}^{*}$,
there exists $\epsilon>0$ so that $a^{*}$ will continue to be a
strict best response after $s^{*}$ for any $\pi_{1}^{'}\in\Pi_{1}$
where for every $\theta\in\Theta$, $|\pi_{1}^{'}(s^{*}|\theta)-\pi_{1}^{*}(s^{*}|\theta)|<3\epsilon.$

Since $\pi^{(k)}\to\pi^{*}$, find large enough $K$ such that $k\ge K$
implies for every $\theta\in\Theta$, $\left|\pi_{1}^{(k)}(s^{*}|\theta)-\pi_{1}^{*}(s^{*}|\theta)\right|<\epsilon.$

Write $e_{n,\theta}^{\text{obs}}$ for the probability that an age-$n$
receiver has encountered type $\theta$ fewer than $\frac{1}{2}n\lambda(\theta)$
times. We will find a number $N^{\text{obs}}<\infty$ so that 
\[
\sum_{\theta\in\Theta}\sum_{n=0}^{\infty}e_{n,\theta}^{\text{obs}}\le N^{\text{obs}}.
\]

Fix some $\theta\in\Theta$. Write $Z_{t}^{(\theta)}\in\{0,1\}$ as
the indicator random variable for whether the receiver sees a type
$\theta$ in period $t$ of his life and write $S_{n}\coloneqq\sum_{t=1}^{n}Z_{t}^{(\theta)}$
for the total number of type $\theta$ encountered up to age $n$.
We have $\mathbb{E}[S_{n}]=n\lambda(\theta)$, so we can use Hoeffding's
inequality to bound $e_{n,\theta}^{\text{obs}}$. 
\begin{align*}
e_{n,\theta}^{\text{obs}} & \le\mathbb{P}\left[|S_{n}-\mathbb{E}[S_{n}]|\ge\frac{1}{2}n\lambda(\theta)\right]\\
 & \le2\exp\left(-\frac{2\cdot[\frac{1}{2}n\lambda(\theta)]^{2}}{n}\right).
\end{align*}
This shows $e_{n,\theta}^{\text{obs}}$ tends to 0 at the same rate
as $\exp(-n)$, so 
\[
\sum_{n=0}^{\infty}e_{n,\theta}^{\text{obs}}\le\sum_{n=0}^{\infty}2\exp\left(-\frac{2\cdot[\frac{1}{2}n\lambda(\theta)]^{2}}{n}\right)\eqqcolon N_{\theta}^{\text{obs}}<\infty.
\]
 So we set $N^{\text{obs}}\coloneqq\sum_{\theta\in\Theta}N_{\theta}^{\text{obs}}$. 

Next, write $e_{n,\theta}^{\text{bias},k}$ for the probability that,
after observing $\left\lfloor \frac{1}{2}n\lambda(\theta)\right\rfloor $
i.i.d. draws from $\pi_{1}^{(k)}(\cdot|\theta)$, the empirical frequency
of signal $s^{*}$ differs from $\pi_{1}^{(k)}(s^{*}|\theta)$ by
more than $2\epsilon$. So again, write $Z_{t}^{\theta,k}\in\{0,1\}$
to indicate if the $t$-th draw resulted in signal $s^{*}$, with
$\mathbb{E}\left[Z_{t}^{\theta,k}\right]=\pi_{1}^{(k)}(s^{*}|\theta)$,
and put $S_{n,k}\coloneqq\sum_{t=1}^{\left\lfloor \frac{1}{2}n\lambda(\theta)\right\rfloor }Z_{t}^{\theta,k}$
for total number of $s^{*}$ out of $\left\lfloor \frac{1}{2}n\lambda(\theta)\right\rfloor $
draws. We have $\mathbb{E}[S_{n,k}]=\left\lfloor \frac{1}{2}n\lambda(\theta)\right\rfloor \cdot\pi_{1}^{(k)}(s^{*}|\theta)$,
but $\left|\pi_{1}^{(k)}(s^{*}|\theta)-\pi_{1}^{*}(s^{*}|\theta)\right|<\epsilon$
whenever $k\ge K$. That means,
\begin{align*}
e_{n,\theta}^{\text{bias},k}\coloneqq & \mathbb{P}\left[\left|\frac{S_{n,k}}{\left\lfloor \frac{1}{2}n\lambda(\theta)\right\rfloor }-\pi_{1}^{*}(s^{*}|\theta)\right|\ge2\epsilon\right]\\
\le & \mathbb{P}\left[\left|\frac{S_{n,k}}{\left\lfloor \frac{1}{2}n\lambda(\theta)\right\rfloor }-\pi_{1}^{(k)}(s^{*}|\theta)\right|\ge\epsilon\right]\text{ if \ensuremath{k\ge K}}\\
= & \mathbb{P}\left[\left|S_{n,k}-\mathbb{E}[S_{n,k}]\right|\ge\left\lfloor \frac{1}{2}n\lambda(\theta)\right\rfloor \cdot\epsilon\right]\\
\le & 2\exp\left(-\frac{2\cdot(\left\lfloor \frac{1}{2}n\lambda(\theta)\right\rfloor \cdot\epsilon)^{2}}{\left\lfloor \frac{1}{2}n\lambda(\theta)\right\rfloor }\right)\text{ by Hoeffding's inequality.}
\end{align*}

Let $N_{\theta}^{\text{bias}}\coloneqq\sum_{n=1}^{\infty}2\exp\left(-\frac{2\cdot(\left\lfloor \frac{1}{2}n\lambda(\theta)\right\rfloor \cdot\epsilon)^{2}}{\left\lfloor \frac{1}{2}n\lambda(\theta)\right\rfloor }\right)$,
with $N_{\theta}^{\text{bias}}<\infty$ since the summand tends to
0 at the same rate as$\exp(-n)$. This argument shows that, whenever
$k\ge K$, we have $\sum_{n=1}^{\infty}e_{n,\theta}^{\text{bias},k}\le N_{\theta}^{\text{bias}}$.
Now let $N^{\text{bias}}\coloneqq\sum_{\theta\in\Theta}N_{\theta}^{\text{bias}}$. 

Finally, since $g$ is regular, we appeal to Proposition 1 of \citet*{fudenberg_he_imhof_2016}
to see that there exists some $\underline{N}$ so that whenever the
receiver has a data set of size $n\ge\underline{N}$ on type $\theta$'s
play, his Bayesian posterior as to the probability that $\theta$
plays $s^{*}$ differs from the empirical distribution by no more
than $\epsilon$. Put $N^{\text{age}}\coloneqq\frac{2\underline{N}}{\min_{\theta\in\Theta}\lambda(\theta)}$. 

Consider any steady state $\psi^{(k)}$ with $k\ge K$. With probability
no smaller than $1-\sum_{\theta\in\Theta}e_{n,\theta}^{\text{bias},k}$,
an age-$n$ receiver who has seen at least $\frac{1}{2}n\lambda(\theta)$
instances of type $\theta$ for every $\theta\in\Theta$ will have
an empirical distribution such that every type's probability of playing
$s^{*}$ differs from $\pi_{1}^{*}(s^{*}|\theta)$ by less than $2\epsilon$.
If, furthermore, $n\ge N^{\text{age}}$, then in fact $\frac{1}{2}n\lambda(\theta)\ge\underline{N}$
for each $\theta$ so the same probability bound applies to the event
that the receiver's Bayesian posterior on every type $\theta$ playing
$s^{*}$ is closer than $3\epsilon$ to $\pi_{1}^{*}(s^{*}|\theta)$.
By the construction of $\epsilon$, playing $a^{*}$ after $s^{*}$
is the unique best response to such a posterior. 

Therefore, for $k\ge K$, the probability that the sender population
plays some action other than $a^{*}$ after $s^{*}$ in $\psi^{(k)}$
is bounded by 
\[
N^{\text{age}}(1-\gamma_{k})+(1-\gamma_{k})\cdot\sum_{n=0}^{\infty}\gamma_{k}^{n}\cdot\sum_{\theta\in\Theta}\left(e_{n,\theta}^{\text{obs}}+e_{n,\theta}^{\text{bias},k}\right).
\]

To explain this expression, receivers aged $N^{\text{age}}$ or younger
account for no more than $N^{\text{age}}(1-\gamma_{k})$ of the population.
Among the age $n$ receivers, no more than $\sum_{\theta\in\Theta}e_{n,\theta}^{\text{obs}}$
fraction has a sample size smaller than $\frac{1}{2}n\lambda(\theta)$
for any type $\theta$, while $\sum_{\theta\in\Theta}e_{n,\theta}^{\text{bias},k}$
is an upper bound on the probability (conditional on having a large
enough sample) of having a biased enough sample so that some type's
empirical frequency of playing $s^{*}$ differs by more than $2\epsilon$
from $\pi_{1}^{*}(s^{*}|\theta)$. 

But since $\gamma_{k}\in[0,1),$ 
\[
\sum_{n=0}^{\infty}\gamma_{k}^{n}\cdot\sum_{\theta\in\Theta}e_{n,\theta}^{\text{obs}}<\sum_{n=0}^{\infty}\sum_{\theta\in\Theta}e_{n,\theta}^{\text{obs}}\le N^{\text{obs}}
\]
and 
\[
\sum_{n=0}^{\infty}\gamma_{k}^{n}\cdot\sum_{\theta\in\Theta}e_{n,\theta}^{\text{bias},k}<\sum_{n=0}^{\infty}\sum_{\theta\in\Theta}e_{n,\theta}^{\text{bias},k}\le N^{\text{bias}}.
\]

We conclude that whenever $k\ge K$, 
\[
\pi_{2}^{(k)}(a^{*}|s^{*})\ge1-(1-\gamma_{k})\cdot(N^{\text{age}}+N^{\text{obs}}+N^{\text{bias}}).
\]

Finally, observe that none of $N^{\text{age}},N^{\text{obs}},N^{\text{bias}}$
depends on the sequence $\pi^{(k)}$, so $N$ is chosen independent
of the sequence $\pi^{(k)}$. 
\end{proof}
\begin{lem}
\label{lem:appendix_affine_bound}Assume $g$ is regular. Suppose
there is some $a^{*}\in A$ and $v\in\mathbb{R}$ so that $u_{1}(\theta,s^{*},a^{*})>v$.
Then, there exist $C_{1}\in(0,1)$, $C_{2}>0$ so that in every sender
history $y_{\theta}$, $\#(s^{*},a^{*}|y_{\theta})\ge C_{1}\cdot\#(s^{*}|y_{\theta})+C_{2}$
implies $\mathbb{E}\left[u_{1}(\theta,s^{*},\pi_{2}(\cdot|s^{*}))|y_{\theta}\right]>v.$
\end{lem}
\begin{proof}
Write $\underline{u}\coloneqq\min_{a\in A}u_{1}(\theta,s^{*},a)$.
There exists $q\in(0,1)$ so that 
\[
q\cdot u_{1}(\theta,s^{*},a^{*})+(1-q)\cdot\underline{u}>v.
\]
Find a small enough $\epsilon>0$ so that $0<\frac{q}{1-\epsilon}<1$. 

Since $g$ is regular, Proposition 1 of \citet*{fudenberg_he_imhof_2016}
tells us there exists some $C_{0}$ so that the posterior mean belief
of sender with history $y_{\theta}$, is no less than 
\[
(1-\epsilon)\cdot\frac{\#(s^{*},a^{*}|y_{\theta})}{\#(s^{*}|y_{\theta})+C_{0}}.
\]

Whenever this expression is at least $q$, the expected payoff to
$\theta$ playing $s^{*}$ exceeds $v$. That is, it suffices to have
\[
(1-\epsilon)\cdot\frac{\#(s^{*},a^{*}|y_{\theta})}{\#(s^{*}|y_{\theta})+C_{0}}\ge q\iff\#(s^{*},a^{*}|y_{\theta})\ge\frac{q}{1-\epsilon}\#(s^{*}|y_{\theta})+\frac{q}{1-\epsilon}\cdot C_{0}.
\]

Putting $C_{1}\coloneqq\frac{q}{1-\epsilon}$ and $C_{2}\coloneqq\frac{q}{1-\epsilon}\cdot C_{0}$
proves the lemma. 
\end{proof}
\begin{lem}
\label{lem:appendix_gambler}Let $Z_{t}$ be i.i.d. Bernoulli random
variables, where $\mathbb{E}[Z_{t}]=1-\epsilon$. Write $S_{n}\coloneqq\sum_{t=1}^{n}Z_{t}.$
For $0<C_{1}<1$ and $C_{2}>0$, there exist $\bar{\epsilon},G_{1},G_{2}>0$
such that whenever $0<\epsilon<\bar{\epsilon}$, 
\[
\mathbb{P}\left[S_{n}\ge C_{1}n+C_{2}\ \forall n\ge G_{1}\right]\ge1-G_{2}\epsilon.
\]
\end{lem}
\begin{proof}
We make use of a lemma from \citet{fudenberg_superstition_2006},
which in turn extends some inequalities from \citet{billingsley1995probability}.

\textbf{FL06 Lemma A.1}: \emph{Suppose $\{X_{k}\}$ is a sequence
of i.i.d. Bernoulli random variables with $\mathbb{E}[X_{k}]=\mu$,
and define for each $n$ the random variable
\[
S_{n}\coloneqq\frac{|\sum_{k=1}^{n}(X_{k}-\mu)|}{n}.
\]
}

\emph{Then for any $\underline{n},\bar{n}\in\mathbb{N}$, 
\[
\mathbb{P}\left[\max_{\underline{n}\le n\le\bar{n}}S_{n}>\epsilon\right]\le\frac{2^{7}}{3}\cdot\frac{1}{\underline{n}}\cdot\frac{\mu}{\epsilon^{4}}.
\]
}

For every $G_{1}>0$ and every $0<\epsilon<1$, 
\begin{align*}
\mathbb{P}[S_{n}\ge C_{1}n+C_{2}\ \forall n\ge G_{1}] & =1-\mathbb{P}\left[(\exists n\ge G_{1})\sum_{t=1}^{n}Z_{t}<C_{1}n+C_{2}\right]\\
 & =1-\mathbb{P}\left[(\exists n\ge G_{1})\sum_{t=1}^{n}(X_{t}-\epsilon)>(1-\epsilon-C_{1})n-C_{2}\right],
\end{align*}
where $X_{t}\coloneqq1-Z_{t}$. Let $\bar{\epsilon}\coloneqq\frac{1}{2}(1-C_{1})$
and $G_{1}\coloneqq2C_{2}/\bar{\epsilon}$. Suppose $0<\epsilon<\bar{\epsilon}$.
Then for every $n\ge G_{1}$, $(1-\epsilon-C_{1})n-C_{2}\ge\bar{\epsilon}n-C_{2}\ge\frac{1}{2}\bar{\epsilon}n$.
Hence, 
\[
\mathbb{P}\left[S_{n}\ge C_{1}n+C_{2}\ \forall n\ge G_{1}\right]\ge1-\mathbb{P}\left[(\exists n\ge G_{1})\sum_{t=1}^{n}(X_{t}-\epsilon)>\frac{1}{2}\bar{\epsilon}n\right]
\]
and, by FL06 Lemma A.1, the probability on the right-hand side is
at most $G_{2}\epsilon$ with $G_{2}\coloneqq2^{11}/(3G_{1}\bar{\epsilon}^{4})$. 
\end{proof}
We now prove Theorem \ref{thm:eqm_dominated_types}. 

\textbf{Theorem }\ref{thm:eqm_dominated_types}: Suppose $\pi^{*}$
is on-path strict for the receiver and patiently stable. Then it satisfies
the strong compatibility criterion. 
\begin{proof}
Let some $a^{'}\notin\text{BR}(\Delta(\tilde{J}(s^{'},\pi^{*})),s^{'})$
and $h>0$ be given. We will show that $\pi_{2}^{*}(a^{'}|s^{'})\le3h$. 

\textbf{\noun{Step 1}}: Defining the constants $\xi,\theta^{J},a_{\theta},s_{\theta},C_{1},C_{2},G_{1},G_{2},\text{ and }N^{\text{recv}}$. 

(i) For each $\xi>0$, define the $\xi$-approximations to $\Delta(\tilde{J}(s^{'},\pi^{*}))$
as the probability distributions with weight no more than $\xi$ on
types outside of $\tilde{J}(s^{'},\pi^{*})$, 
\[
\Delta_{\xi}(\tilde{J}(s^{'},\pi^{*}))\coloneqq\left\{ p\in\Delta(\Theta):p(\theta)\le\xi\ \forall\theta\notin\tilde{J}(s^{'},\pi^{*})\right\} .
\]

Because the best-response correspondence has closed graph, there exists
some $\xi>0$ so that $a^{'}\notin\text{BR}(\Delta_{\xi}(\tilde{J}(s^{'},\pi^{*})),s^{'})$. 

(ii) Since $\tilde{J}(s^{'},\pi^{*})$ is nonempty, we can fix some
$\theta^{J}\in\tilde{J}(s^{'},\pi^{*})$. 

(iii) For each equilibrium-dominated type $\theta\in\Theta\backslash\tilde{J}(s^{'},\pi^{*})$,
identify some on-path signal $s_{\theta}$ so that $\pi_{1}^{*}(s_{\theta}|\theta)>0$.
By assumption of on-path strictness for the receiver, there is some
$a_{\theta}\in A$ so that $\pi_{2}^{*}(a_{\theta}|s_{\theta})=1$,
and furthermore, $a_{\theta}$ is the strict best response to $s_{\theta}$
in $\pi^{*}$. By the definition of equilibrium dominance, 
\[
u_{1}(\theta,s_{\theta},a_{\theta})>\max_{a\in A}u_{1}(\theta,s^{'},a)\eqqcolon v_{\theta}.
\]

By applying Lemma \ref{lem:appendix_affine_bound} to each $\theta\in\Theta\backslash\tilde{J}(s^{'},\pi^{*})$,
we obtain some $C_{1}\in(0,1)$, $C_{2}>0$ so for every $\theta\in\Theta\backslash\tilde{J}(s^{'},\pi^{*})$
and in every sender history $y_{\theta}$, $\#(s_{\theta},a_{\theta}|y_{\theta})\ge C_{1}\cdot\#(s_{\theta}|y_{\theta})+C_{2}$
implies $\mathbb{E}\left[u_{1}(\theta,s_{\theta},\pi_{2}(\cdot|s_{\theta}))|y_{\theta}\right]>v_{\theta}.$

(iv) By Lemma \ref{lem:appendix_gambler}, find $\bar{\epsilon},G_{1},G_{2}>0$
such that if $\mathbb{E}[Z_{t}]=1-\epsilon$ are i.i.d. Bernoulli
and $S_{n}\coloneqq\sum_{t=1}^{n}Z_{t}$, then whenever $0<\epsilon<\bar{\epsilon}$,
\[
\mathbb{P}\left[S_{n}\ge C_{1}n+C_{2}\ \forall n\ge G_{1}\right]\ge1-G_{2}\epsilon.
\]

(v) Because at $\pi^{*}$, $a_{\theta}$ is a strict best response
to $s_{\theta}$ for every $\theta\in\Theta\backslash\tilde{J}(s^{'},\pi^{*})$,
from Lemma \ref{lem:appendix_receiver_rate} we may find a $N^{\text{recv}}$
so that for each sequence $\pi^{(k)}\in\Pi^{*}(g,\delta_{k},\gamma_{k})$
where $\gamma_{k}\to1$,$\pi^{(k)}\to\pi^{*}$, there corresponds
$K^{\text{recv}}\in\mathbb{N}$ so that $k\ge K^{\text{recv}}$ implies
$\pi_{2}^{(k)}(a_{\theta}|s_{\theta})\ge1-(1-\gamma_{k})\cdot N^{\text{recv}}$
for every $\theta\in\Theta\backslash\tilde{J}(s^{'},\pi^{*})$. 

\textbf{\noun{Step 2}}: Two conditions to ensure that all but $3h$
receivers believe in $\Delta_{\xi}(\tilde{J}(s^{'},\pi^{*}))$. 

Consider some steady state $\psi\in\Psi^{*}(g,\delta,\gamma)$ for
$g$ regular, $\delta,\gamma\in[0,1)$. 

In Theorem 2 of \citet*{fudenberg_he_imhof_2016}, put $c=\frac{2}{\xi}\cdot\frac{\max_{\theta\in\Theta}\lambda(\theta)}{\lambda(\theta^{J})}$
and $\delta=\frac{1}{2}$. We conclude that there exists some $N^{\text{rare}}$
(not dependent on $\psi$) such that whenever $\pi_{1}(s^{'}|\theta^{J})\ge c\cdot\pi_{1}(s^{'}|\theta^{\text{D}})$
for every equilibrium-dominated type $\theta^{\text{D}}\notin\tilde{J}(s^{'},\pi^{*})$
and 
\begin{equation}
n\cdot\pi_{1}(s^{'}|\theta^{J})\ge N^{\text{rare}},\label{eq:eqd_sample}
\end{equation}
then an age-$n$ receiver in steady state $\psi$ where $\pi=\sigma(\psi)$
has probability at least $1-h$ of holding a posterior belief $g_{2}(\cdot|y_{2})$
such that $\theta^{J}$ is at least $\frac{1}{2}c$ times as likely
to play $s^{'}$ as $\theta^{\text{D}}$ is for every $\theta^{\text{D }}\notin\tilde{J}(s^{'},\pi^{*})$.
Thus history $y_{2}$ generates a posterior belief after $s^{'},$
$p(\cdot|s^{'};y_{2})$ such that 
\[
\frac{p(\theta^{\text{D}}|s^{'};y_{2})}{p(\theta^{J}|s^{'};y_{2})}\le\frac{\lambda(\theta^{\text{D}})}{\lambda(\theta^{J})}\cdot\xi\cdot\frac{\lambda(\theta^{J})}{\max_{\theta\in\Theta}\lambda(\theta)}\le\xi.
\]

In particular, $p(\cdot|s^{'};y_{2})$ must assign weight no greater
than $\xi$ to each type not in $\tilde{J}(s^{'},\pi^{*})$; therefore,
the belief belongs to $\Delta_{\xi}(\tilde{J}(s^{'},\pi^{*}))$. By
construction of $\xi$, $a^{'}$ is then not a best response to $s^{'}$
after history $y_{2}$. 

A receiver whose age $n$ satisfies Equation (\ref{eq:eqd_sample})
plays $a^{'}$ with probability less than $h$, provided $\pi_{1}(s^{'}|\theta^{J})\ge c\cdot\pi_{1}(s^{'}|\theta^{\text{D}})$
for every $\theta^{\text{D}}\notin\tilde{J}(s^{'},\pi^{*})$. However,
to bound the overall probability of $a^{'}$ in the entire receiver
population in steady state $\psi$, we ensure that Equation (\ref{eq:eqd_sample})
is satisfied for all except $2h$ fraction of receivers in $\psi$.
We claim that when $\gamma$ is large enough, a sufficient condition
is for $\pi=\sigma(\psi)$ to satisfy $\pi_{1}(s^{'}|\theta^{J})\ge(1-\gamma)N^{*}$
for some $N^{*}\ge N^{\text{rare}}/h$. This is because under this
condition, any agent aged $n\ge\frac{h}{1-\gamma}$ satisfies Equation
(\ref{eq:eqd_sample}), while the fraction of receivers younger than
$\frac{h}{1-\gamma}$ is $1-\left(\gamma^{\frac{h}{1-\gamma}}\right)\le2h$
for $\gamma$ near enough to 1. 

To summarize, in Step 2 we have found a constant $N^{\text{rare}}$
and shown that if $\gamma$ is near enough to 1, then $\pi=\sigma(\psi)$
has $\pi_{2}(a^{'}|s^{'})\le3h$ if the following two conditions are
satisfied: 

(\textbf{C1}) $\pi_{1}(s^{'}|\theta^{J})\ge c\cdot\pi_{1}(s^{'}|\theta^{\text{D}})$
for every equilibrium-dominated type $\theta^{\text{D}}\notin\tilde{J}(s^{'},\pi^{*})$

(\textbf{C2}) $\pi_{1}(s^{'}|\theta^{J})\ge(1-\gamma)N^{*}$ for some
$N^{*}\ge N^{\text{rare}}/h.$

In the following step, we show there is a sequence of steady states
$\psi^{(k)}\in\Psi^{*}(g,\delta_{k},\gamma_{k})$ with $\delta_{k}\to1,$
$\gamma_{k}\to1$, and $\sigma(\psi^{(k)})=\pi^{(k)}\to\pi^{*}$ such
that, in every $\pi^{(k)}$, the above two conditions are satisfied.
Using the fact that $\gamma_{k}\to1,$ we conclude that, for large
enough $k$, we get $\pi_{2}^{(k)}(a^{'}|s^{'})\le3h$, which in turn
shows $\pi^{*}(a^{'}|s^{'})\le3h$ due to the convergence $\pi^{(k)}\to\pi^{*}$. 

\textbf{\noun{Step 3}}: Extracting a suitable subsequence of steady
states. 

In the statement of Lemma \ref{lem:nondom_message}, put $\theta^{'}\coloneqq\theta^{J}$.
We obtain some number $\epsilon$ and functions $\bar{\delta}(N),$
$\bar{\gamma}(N,\delta)$. Put $N^{\text{ratio}}\coloneqq\frac{2}{\xi}G_{2}\cdot N^{\text{recv}}\frac{\max_{\theta\in\Theta}\lambda(\theta)}{\lambda(\theta^{J})}$
and $N^{*}\coloneqq\max(N^{\text{ratio}},N^{\text{rare}}/h)$. 

Since $\pi^{*}$ is patiently stable, it can be written as the limit
of some strategy profiles $\pi^{*}=\lim_{k\to\infty}\pi^{(k)}$, where
each $\pi^{(k)}$ is $\delta_{k}$-stable with $\delta_{k}\to1$.
By the definition of $\delta$-stable, each $\pi^{(k)}$ is the limit
$\pi^{(k)}=\lim_{j\to\infty}\pi^{(k,j)}$ with $\pi^{(k,j)}\in\Pi^{*}(g,\delta_{k},\gamma_{k,j})$
with $\lim_{j\to\infty}\gamma_{k,j}=1$. It is without loss to assume
that for every $k\ge1,$ $\delta_{k}\ge\bar{\delta}(N^{*})$, and
that the $L_{1}$ distance between $\pi^{(k)}$ and $\pi^{*}$ is
less than $\epsilon/2$. Now, for each $k$, find a large enough index
$j(k)$ so that (i) $\gamma_{k,j(k)}\ge\gamma(N^{*},\delta_{k})$,
(ii) $L_{1}$ distance between $\pi^{(k,j)}$ and $\pi^{(k)}$ is
less than $\min(\frac{\epsilon}{2},\frac{1}{k})$, and (iii) $\lim_{k\to\infty}\gamma_{k,j(k)}=1$.
This generates a sequence of $k$-indexed steady states, $\psi^{(k,j(k))}\in\Psi^{*}(g,\delta_{k},\gamma_{k,j(k)})$.
We will henceforth drop the dependence through the function $j(k)$
and just refer to $\psi^{(k)}$ and $\gamma_{k}$. The sequence $\psi^{(k)}\in\Psi^{*}(g,\delta_{k},\gamma_{k})$
satisfies: (1) $\delta_{k}\to1,\gamma_{k}\to1$; (2) $\delta_{k}\ge\bar{\delta}(N^{*})$
for each $k$; (3) $\gamma_{k}\ge\bar{\gamma}(N^{*},\delta_{k})$
for each $k$; (4) $\pi^{(k)}\to\pi^{*}$; (5) the $L_{1}$ distance
between $\bar{\psi}^{(k)}$ and $\pi^{*}$ is no larger than $\epsilon$.
Lemma \ref{lem:nondom_message} implies that, for every $k$, $\pi_{1}^{(k)}(s^{'}|\theta^{J})\ge(1-\gamma_{k})N^{*}$.
So, every member of the sequence thus constructed satisfies condition
(\textbf{C2}). 

\textbf{\noun{Step 4}}: An upper bound on experimentation probability
of equilibrium-dominated types. 

It remains to show that eventually condition (\textbf{C1}) is also
satisfied in the sequence constructed in \textbf{\noun{Step 3}}\noun{.}\textbf{\noun{
}}We first bound the rate at which the aggregate receiver strategy
$\pi_{2}^{(k)}$ converges to $\pi_{2}^{*}$. By Lemma \ref{lem:appendix_receiver_rate},
there exists some $K^{\text{recv}}$ so that $k\ge K^{\text{recv}}$
implies $\pi_{2}^{(k)}(a_{\theta}|s_{\theta})\ge1-(1-\gamma_{k})\cdot N^{\text{recv}}$
for every $\theta\in\Theta\backslash\tilde{J}(s^{'},\pi^{*})$. Find
next a large enough $K^{\text{error}}$ so that $k\ge K^{\text{error}}$
implies $(1-\gamma_{k})\cdot N^{\text{recv}}<\bar{\epsilon}$ (where
$\bar{\epsilon}$ was defined in \textbf{\noun{Step 1}}). 

We claim that when $k\ge\max(K^{\text{recv}},K^{\text{error}})$,
a type $\theta\notin\tilde{J}(s^{'},\pi^{*})$ sender who always sends
signal $s_{\theta}$ against a receiver population that plays $\pi_{2}^{(k)}(\cdot|s_{\theta})$
has less than $(1-\gamma_{k})\cdot N^{\text{recv}}\cdot G_{2}$ chance
of ever having a posterior belief that the expected payoff to $s_{\theta}$
is no greater than $v_{\theta}$ in some period $n\ge G_{1}$. This
is because by Lemma \ref{lem:appendix_gambler}, 
\[
\mathbb{P}\left[S_{n}\ge C_{1}n+C_{2}\ \forall n\ge G_{1}\right]\ge1-G_{2}\cdot\pi_{2}^{(k)}(\{a\ne a_{\theta}\}|s_{\theta})\ge1-G_{2}\cdot(1-\gamma_{k})\cdot N^{\text{recv}}
\]
where $S_{n}$ refers to the number of times that the receiver population
responded to $s_{\theta}$ with $a_{\theta}$ in the first $n$ times
that $s_{\theta}$ was sent. But Lemma \ref{lem:appendix_affine_bound}
guarantees that provided $S_{n}\ge C_{1}n+C_{2}$, sender's expected
payoff for $s_{\theta}$ is strictly above $v_{\theta}$, so we have
established the claim. 

Finally, find a large enough $K^{\text{Gittins}}$ so that $k\ge K^{\text{Gittins}}$
implies the effective discount factor $\delta_{k}\gamma_{k}$ is so
near 1 that for every $\theta\notin\tilde{J}(s^{'},\pi^{*}),$ the
Gittins index for signal $s_{\theta}$ cannot fall below $v_{\theta}$
if $s_{\theta}$ has been used no more than $G_{1}$ times. (This
is possible since the prior is non-doctrinaire.) Then for $k\ge\max(K^{\text{recv}},K^{\text{error}},K^{\text{Gittins}})$,
there is less than $G_{2}\cdot(1-\gamma_{k})\cdot N^{\text{recv}}$
chance that the equilibrium-dominated sender $\theta\notin\tilde{J}(s^{'},\pi^{*})$
will play $s^{'}$ even once. To see this, we observe that according
to the prior, the Gittins index for $s_{\theta}$ is higher than that
of $s^{'}$, whose index is no higher than its highest possible payoff
$v_{\theta}$. This means the sender will not play $s^{'}$ until
her Gittins index for $s_{\theta}$ has fallen below $v_{\theta}$.
Since $k\ge K^{\text{recv}}$, this will not happen before the sender
has played $s_{\theta}$ at least $G_{1}$ times, and since $k\ge\max(K^{\text{error}},K^{\text{recv}}$),
the previous claim establishes that the probability of the expected
payoff to $s_{\theta}$ (and, a fortiori, the Gittins index for $s_{\theta})$
ever falling below $v_{\theta}$ sometime after playing $s_{\theta}$
for the $G_{1}$-th time is no larger than $G_{2}\cdot(1-\gamma_{k})\cdot N^{\text{recv}}$. 

This shows that, for $k\ge\max(K^{\text{recv}},K^{\text{error}},K^{\text{Gittins}})$,
$\pi_{1}^{(k)}(s^{'}|\theta)\le G_{2}N^{\text{recv}}\cdot(1-\gamma_{k})$
for every $\theta\notin\tilde{J}(s^{'},\pi^{*})$. But since $\pi_{1}^{(k)}(s^{'}|\theta^{J})\ge N^{*}\cdot(1-\gamma_{k})$
where $N^{*}\ge N^{\text{ratio}}=\frac{2}{\xi}G_{2}\cdot N^{\text{recv}}\frac{\max_{\theta\in\Theta}\lambda(\theta)}{\lambda(\theta^{J})}$,
we see that condition (\textbf{C1}) is satisfied whenever $k\ge\max(K^{\text{recv}},K^{\text{error}},K^{\text{Gittins}})$.
\end{proof}

\end{document}